\documentclass[3p]{elsarticle}

\usepackage{hyperref}
\usepackage{amssymb}
\usepackage{graphicx}
\graphicspath{{Figures/}}
\usepackage{url}
\usepackage{amsmath}
\usepackage[scriptsize]{subfigure}
\usepackage{float}
\usepackage{algpseudocode}

\begin{document}

\begin{frontmatter}

\title{Artificial Intelligence Based Malware Analysis}

\author[cra]{Avi Pfeffer\corref{corr}}
\ead{apfeffer@cra.com}
\author[cra]{Brian Ruttenberg\corref{corr}}
\ead{bruttenberg@cra.com}
\author[cra]{Lee Kellogg}
\author[cra]{Michael Howard}
\author[cra]{Catherine Call}
\author[cra]{Alison O'Connor}
\author[cra]{Glenn Takata}
\author[cra]{Scott Neal Reilly}
\author[cra]{Terry Patten}
\author[cra]{Jason Taylor}
\author[cra]{Robert Hall}
\address[cra]{Charles River Analytics\\625 Mt. Auburn St.\\Cambridge, MA, 02138}

\author[ull]{Arun Lakhotia}
\author[ull]{Craig Miles}
\address[ull]{Software Research Lab\\University of Louisiana at Lafayette\\Lafayette, LA}

\author[ais]{Dan Scofield}
\author[ais]{Jared Frank}
\address[ais]{Assured Information Security\\Rome, NY}

\cortext[corr]{Corresponding authors}
%\ead{bruttenberg@cra.com}

\begin{abstract}
Artificial intelligence methods have often been applied to perform specific functions or tasks in the cyber--defense realm. However, as adversary methods become more complex and difficult to divine, piecemeal efforts to understand cyber--attacks, and malware--based attacks in particular, are not providing sufficient means for malware analysts to understand the past, present and future characteristics of malware.

In this paper, we present the Malware Analysis and Attributed using Genetic Information (MAAGI) system. The underlying idea behind the MAAGI system is that there are strong similarities between malware behavior and biological organism behavior, and applying biologically inspired methods to corpora of malware can help analysts better understand the ecosystem of malware attacks. Due to the sophistication of the malware and the analysis, the MAAGI system relies heavily on artificial intelligence techniques to provide this capability. It has already yielded promising results over its development life, and will hopefully inspire more integration between the artificial intelligence and cyber--defense communities.
\end{abstract}


\end{frontmatter}

%\linenumbers

\newdefinition{definition}{Definition}
\newtheorem{theorem}{Theorem}
\newtheorem{corollary}{Corollary}
\newproof{proof}{Proof}

\section{Introduction}

Artificial intelligence (AI) ideas and methods have been successfully applied in countless domains to learn complex processes or systems, make informed decisions, or to model cognitive or behavior processes . While there is a tremendous need for artificial intelligence applications that can automate tasks and supplant human analysis, there is also still a place for artificial intelligence technology to enhance human analysis of complex domains and help people understand the underlying phenomenon behind human processes. For instance, artificial intelligence is often applied in the medical field to assist in diagnostic efforts by reducing the space of possible diagnoses for a patient and helping medical professionals understand the complex processes underlying certain diseases~\cite{szolovits1988artificial}.

Cyber security and defense is one field that has tremendously benefited from the explosive growth of artificial intelligence. Spam filters, commercial malware software, and intrusion detection systems are just small example of applications of artificial intelligence to the cyber security realm~\cite{tyugu2011artificial}. However, like in the medical diagnosis field, human judgment and analysis is often needed to reason about complex processes in the cyber security domain. Nevertheless, artificial intelligence tools and concepts can be used help people discover, understand and model these complicated domains. Malware analysis is one task within the cyber defense field that can especially benefit from this type of computer--based assistance.

Loosely defined, malware is just unwanted software that performs some non--benign, often nefarious, operation on any system. Unfortunately, most modern computer users are all too familiar with the concept of malware; McCafee reports almost 200 million unique malware binaries in their archive and increasing at nearly 100,000 binaries per day~\cite{cruz2014mcafee}. As a result of this threat, there has been a concerted push within the cyber defense community to improve the mitigation of malware propagation and reduce the number of successful malware attacks by detailed analysis and study of malware tactics, techniques and procedures.

Principally, malware analysis is focused on several key ideas: Understanding how malware operates; determining the relationship between similar pieces of malware; modeling and predicting how malware changes over time; attributing authorship; and how do new ideas and tactics propagate to new malware. These malware analysis tasks can be very difficult to perform. First, the shear number of new malware generated makes it extremely hard to automatically track patterns of malware over time, let alone allow a human to manually explore different malware binaries. Second, the relationships and similarities between different pieces of malware can be very complex and difficult to find in a large corpus of malware. Finally, malware authors also purposefully obfuscate their methods to prevent analysis, such as by encrypting their binaries (also known as packing). 

Due to these difficulties, malware analysis could tremendously benefit from artificial intelligence techniques, where automated systems could learn models of malware evolution over time, cluster similar types of malware together, and predict the behavior or prevalance of future malware. Cyber defenders and malware analysts could use these AI enabled malware analysis systems to explore the ecosystem of malware, determine new and significant threats, or develop new defenses. 

In this paper, we describe the Malware Analysis and Attribution Using Genetic Information (MAAGI) system. This malware analysis system relies heavily on AI techniques to assist malware analysts in their quest to understand, model and predict the malware ecosystem. The underlying idea behind the MAAGI system is that there are strong similarities between malware behavior and biological organism behavior. As such, the MAAGI system borrows analysis methods and techniques from the genomics, phylogenetics, and evolution fields and applies them to malware ecosystems. Given the large amount of existing (and future) malware, and the complexity of malware behavior, AI methods are a critical tool needed to apply these biological concepts to malware analysis. While the MAAGI system is moving towards a production level deployment, it has already yielded promising results over its development life and has the potential to greatly expand capabilities of malware analysts

\section{Background and Motivation}
\label{background}
The need for AI applications in malware analysis arises from the typical work flow of malware analysts, and the limitations of the methods and tools they currently use to support that work flow. After malware is collected, the analysis of that malware occurs in two main phases. The first is a filtering or triage stage, in which malware are selected for analysis according to some criteria. That criteria includes whether or not they have been seen before, and if not, whether or not they are interesting in some way. If they can be labeled as a member of a previously analyzed malware family, an previous analysis of that family can be utilized and compared against to avoid repetition of effort. If they are novel, and if initial inspection proves them to be worthy of further analysis, they are passed to the second phase where a deep--dive analysis is performed.

The triage phase typically involves some signature or hash--based filtering~\cite{jang_bitshred:_2011}. Unfortunately, this process is only based on malware that has been explicitly analyzed before, and is not good at identifying previously--seen malware that has been somehow modified or obfuscated to form a new variant. In these cases the responsibility generally falls to a human analyst to recognize and label the malware, where it risks being mistaken as truly novel malware and passed on for deep--dive analysis. When strictly relying on human--driven processes for such recognition and classification tasks there are bound to be oversights due to the massive volume and velocity with which new malware variants are generated and collected. In this case, oversights lead to repetition of human effort, which exacerbates the problem, as well as missed opportunities of recognizing and analyzing links between similar malware. Intelligent clustering and classification methods are needed to support the triage process, not to replace the human analysts but to support them, by learning and recognizing families of similar malware among very large collections where signatures and hashes fail, suggesting classifications of incoming malware and providing evidence for those distinctions. A lot of work has been done towards applying clustering techniques to malware, including approaches based on locality--sensitive hashing \cite{bayer2009scalable}, prototype--based hierarchical clustering and classification \cite{rieck2011automatic}, and other incremental hierarchical approaches \cite{perdisci2013scalable,sahoo2006incremental}. The question of whether a truly novel sample is interesting or not is also an important one. With so many novel samples arriving every day it becomes important to prioritize those samples to maximize the utilization of the best human analysts on the greatest potential threats. Automated prioritization techniques based on knowledge learned from past data could provide a baseline queue of incoming malware without any initial human effort required.

The deep--dive phase can also strongly benefit from AI applications. A typical point of interest investigated during this phase is to recognize high--level functional goals of malware; that is, to identify the purpose of the malware, and the intention or motivation of its author. Currently this task is performed entirely by a human expert through direct inspection of the static binary using standard reverse--engineering tools such as a disassembler or a decompiler, and perhaps a debugger or a dynamic execution environment. This has been a very effective means of performing such analysis. However, it relies heavily on the availability of malware reverse--engineering experts. Also, characterizing complex behaviors and contextual motivation of an attacker requires the ability of those experts to recognize sometimes complicated sequences of interrelated behaviors across an entire binary, which can be a very difficult task, especially in large or obfuscated malware samples. 

Besides such single--malware factors, other goals of the deep--dive phase could benefit from intelligent cross--malware sample analysis. For example, the evolution of malware families is of great interest to analysts since it indicates new capabilities being added to malware, new tactics or vulnerabilities being exploited, the changing goals of malware authors, and other such dynamic information that can be used to understand and perhaps even anticipate the actions of the adversary. While it is possible for a human analyst to manually compare the malware samples in a small family and try to deduce simple single--inheritence patterns, for example, it quickly becomes a very difficult task when large families or complex inheritance patterns are introduced~\cite{dumitras2011experimental}. Research has been done towards automatically learning lineages of malware families \cite{karim2005malware}, and even smarter methods based on AI techniques could provide major improvements and time savings over human analyst--driven methods. Besides the technical evolution of malware, other cross--family analyses could benefit from automated intelligent techniques. For example, the life cycle patterns of malware, which can be useful to understand and anticipate growth rates and life spans of malware families, can be difficult for an analyst to accurately identify since the data sets are so large.

The goals of the malware analyst, and the various types of information about binaries gathered during the triage and deep--dive phases, are broad. The level of threat posed by a piece of malware, the nature of that threat, the motivation of the attacker, the novelty of the malware, how it fits into the evolution of a family and other factors can all be of interest. Also, these factors may be of particular interest when combined with one another. For example, combining the output of clustering, lineage and functional analyses could reveal the method by which the authors of a particular malware family made the malware's credential stealing module more stealthy in recent versions, suggesting that other families may adopt similar techniques. For this reason, an automated malware analysis tool should take all of these types of information into account, and should have solutions for detecting all of them and presenting them to the analyst in a cohesive manner.

\section{The Cyber Genome Program}

The MAAGI system was developed under the US Defense Advanced Research Projects Agencies' (DARPA) Cyber Genome program. One of the core research goals of this program is the development of new analysis techniques to automate the discovery, identification, and characterization of malware variants and thereby accelerate the development of effective responses~\cite{cybergenome}. 

The MAAGI system provides the capabilities described in Sec.~\ref{background} based on two key insights, each of which produces a useful analogy. The first is that malware is very rarely created \textit{de novo} and in isolation from the rest of the malware community. Attackers often reuse code and techniques from one malware product to the next. They want to avoid having to write malware from scratch each time, but they also want to avoid detection, so they try to hide similarities between new and existing malware. If we can understand reuse patterns and recognize reused elements, we will be able to connect novel malware to other malware from which it originated, create a database of malware with relationships based on similarity that analysts can interactively explore, and predict and prepare to defend against future attacks.

This insight suggests a biological analogy, in which a malware sample is compared to a living organism. Just like a living organism, the malware sample has a phenotype, consisting of its observed properties and behavior (e.g. eye color for an organism, type of key--logger used by malware). The phenotype is not itself inherited between organisms; rather it is the expression of the genotype which is inherited. Likewise, it is the source code of the malware that is inherited from one sample to the next. Similar to biological evolution, malware genes that are more "fit" (i.e., more successful) are more likely to be copied and propagated by other malware authors. By reverse engineering the malware to discover its intrinsic properties, we can get closer to the genetic features that are inherited.

The second insight is that the function of malware has a crucial role to play in understanding reuse patterns. The function of malware – what it is trying to accomplish – is harder to change than the details of how it is accomplished. Therefore, analysis of the function of malware is a central component of our program. In our analysis of function, we use a linguistic analogy: a malware sample is like a linguistic utterance. Just like an utterance, malware has temporal structure, function, and context in which it takes place. The field of functional linguistics studies utterances not only for their meaning, but also for the function they are trying to accomplish and the context in which they occur.

The MAAGI system provides a \textit{complete} malware analysis framework based on these biological and linguistic analogies. While some of the ideas in biology and linguistics have been applied to malware analysis in the past, the MAAGI system integrates all of these novel analysis techniques into a single workflow, where the output of one analysis can be fed to another type of analysis. While biology and linguistics provide the foundational theory for these malware analysis schemes, they could not be implemented or integrated without using sophisticated artificial intelligence methods.

\section{Overview of MAAGI system}

\begin{figure}[t]
\centering % was full columnwidth
\includegraphics[width=.6\columnwidth]{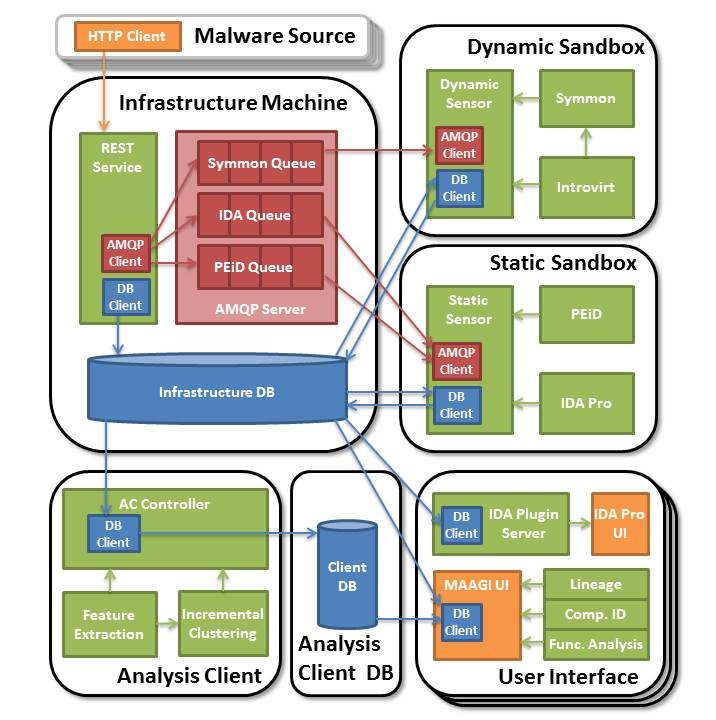}
\vspace{-1mm}
\caption{High level architecture of the MAAGI system.\label{arch}}
\end{figure}

The Malware Analysis and Attribution using Genetic Information (MAAGI) system is a complete, automated malware analysis system that applies AI techniques to support all phases of traditional malware analysis work flows, including reverse--engineering, high--level functional analysis and attacker motivational characterization, clustering and family classification of malware, evolutionary lineage analysis of malware families, shared component identification, and predicting the future of malware families. An overview of the system is shown in Fig.~\ref{arch}.

Malware is first uploaded to the system, then submitted for static reverse--engineering and dynamic execution and tracing, via HTTP requests. The static analysis is built on top of the Hex-Rays IDA disassembler. In addition to other features of the binary including strings and API imports, it generates a semantic representation of the function of code blocks, called BinJuice \cite{lakhotia2013fast}. This code abstraction has been shown to be more robust to obfuscation, compiler changes, and other cross--version variation within malware families, which improves performance on clustering, lineage, and other types of analysis that rely on measuring the similarity between binaries.

As malware are reversed, they are fed into our clustering system where they are incrementally organized into a hierarchy of malware families. The system first classifies the new malware into the hierarchy at the leaves, then re--shapes sections of the hierarchical structure as needed to maintain an optimal organization of the malware into families. Besides providing the basis for an efficient method of incremental clustering, maintaining a database of malware as a hierarchy offers other benefits, such as the ability to perform analysis tasks on nodes at any level of the hierarchy, which are essentially representatives of malware families at various levels of granularity.

The hierarchical clustering of malware can then be viewed and explored remotely through the MAAGI user interface, our analysis and visualization tool. To understand malware of interest in the context of their family, users can view and compare the features of various malware and their procedures, including strings, API calls, code blocks, and BinJuice. Search functions are also available to allow users to find related samples according to their features, including searches for specific features, text--based searches, and searches for the most similar procedures in the database to a particular procedure, which uses MinHash \cite{broder1997resemblance} as an approximation of the similarity between sets of BinJuice blocks.

% Is this relevant? I feel like the UI stuff might be off topic. I put it in because it does play a part in the control of the system, but it could certainly be removed.

From the user interface, users can initiate our various other analyses over subsets of the database. For example, lineage analysis can be run over the set of malware currently selected within the tool. Lineage analysis uses a probabilistic model of malware development, built using the Figaro probabilistic programming language \cite{pfeffer2011practical}, to learn an inheritance graph over a set of malware. This directed graph represents the pattern of code adoption and revision across a family of malware, and can include complex structures with branches, merges, and multiple roots. The returned lineage is visualized in the user interface where the malware can once again be explored and compared to one another. Learning a lineage over a family can assist malware analysts in understand the nature of the relationships between similar samples, and provide clues concerning the strategy and motivation behind changes in code and functionality adopted during development of the malware.

To further understand how code sharing occurs between malware samples, users can also run the component identification process, which identifies cohesive, functional code components shared across malware binaries. Component identification uses a multi--step clustering process. In the first step it clusters together similar procedures across the set of binaries using Louvain clustering, which is a greedy agglomerative clustering method \cite{blondel2008fast}. In the second step it proceeds to cluster these resulting groups of procedures together based on the samples in which they appear. This results in the specification of a number of components, which are collections of procedures within and across malware binaries that together comprise a functional unit.

MAAGI also includes a functional analysis capability that can determine the functional goals of malware, or the motivational context behind the attack. When executed over any sample in the database, it analyzes the static features of the malware using systemic functional grammars (SFGs). Systemic functional linguistics is an approach to the study of natural language that focuses on the idea of speech as a means to serve a specific function. The application of SFGs in this case leverages the analogy of malware as an application of the language of Win32 executable code for the purpose of exhibiting particular malicious behaviors. The grammars themselves are created by subject matter experts and contain structured descriptions of the goals of malicious software and the context in which such attacks occur. The result of MAAGI's functional analysis is a human--readable report describing the successful parses through the grammar, ranked according to the level of confidence, along with the evidence found within the binary to justify those functional designations.

The final element of the MAAGI system is a trend analysis and prediction capability. This tool predicts future attributes of malware families, given only the first malware sample discovered from that family. These attributes include such factors as the number of distinct versions within the family, the amount of variability within the family, and the length of time until the final variant will be seen. The regression model for these predictions are learned from data using machine learning techniques including support vector machines (SVM) \cite{gunn1998support}. 

%The second type of prediction concerns the life cycle patterns of malware families, specifically the number of variants of a malware family that appear in the wild per time period over the course of the family's lifetime. Given the variants that have already been seen, the system predicts the number of new variants that will be discovered within the next time period. This dynamic life cycle prediction is based on probabilistic models of malware development patterns built using the Figaro language, and the parameters of the model were learned using the Metropolis--Hastings \cite{chib1995understanding} algorithm.
 
\section{MAAGI Persistent Malware Hierarchy}
\label{clustering}

While the focus of most malware analysis is to learn as much as possible about a particular malware binary, the analysis is both facilitated and improved by first clustering incoming malware binaries into families along with all previously seen malware. First, it assists in the triage process by immediately associating new malware binaries with any similar, previously--seen malware. This allows those binaries to be either assigned to analysts familiar with that family or lowered in priority if the family to which they are assigned is deemed uninteresting, or if the binary is only superficially different from previously--analyzed malware. Second, providing clustering results to an analyst improves the potential quality of the information gathered by the deep--dive analysis process. Rather than looking at the malware in isolation, it can be analyzed within the context of its family, allowing analysts to re--use previous analyses while revealing potentially valuable evolutionary characteristics of the malware, such as newly--added functionality or novel anti--analysis techniques.

To support clustering within a real--world malware collection and analysis environment, however, it is important that the operation be performed in an online fashion. As new malware are collected and added to the database, they need to be incorporated into the existing set of malware while maintaining the optimal family assignments across the entire database. The reason for this is that clustering of high--dimensional data such as malware binaries is by nature a highly time--complex operation. It is simply not reasonable to cluster the entire database of malware each time a new binary is collected. To support the online clustering of malware we propose a novel incremental method of maintaining a hierarchy of malware families. This hierarchical representation of the malware database supports rapid, scalable online clustering; provides a basis for efficient classification of new malware into the database; describes the complex evolutionary relationships between malware families; and provides a means of performing analysis over entire families of malware at different granularities by representing non--leaf nodes in the hierarchy as agglomerations of the features of their descendants.

% maybe a reduced version of the related work section here

In the following sections, we will outline the hierarchical data structure itself, the online algorithm, the specific choices made for our implementation of the algorithm, the results of some of the experiments performed over actual malware data sets, and conclusions we have drawn based on our work. For a more thorough description of our approach, see our paper titled Hierarchical Management of Large–Scale Malware Data \cite{?????}, which has been submitted for publication, and from which much of this material originated.

\subsection{Hierarchy}
\label{hier}

At the center of our approach is a hierarchical data structure. Malware are inserted into the hierarchy as they are introduced to the system. The various steps of our incremental algorithm operate over this data structure, and it has a number of properties and restrictions to support these operations.

The hierarchy is a tree structure with nodes and edges, and no restriction on the number of children a node may have. There is a single root at the top and the malware themselves are at the leaves. All of the nodes that are not leaves are called exemplars, and they act as representatives of all of the malware below them at the leaves. Just as the binaries are represented in the system by some features, exemplars are represented in the same feature space based on the binaries below them. Sec.~\ref{impl} describes some methods of building exemplars. This common representation of exemplars allows them to be compared to malware binaries, or to one another, using common algorithms and distance functions. The exemplars that have binaries as their children are called leaf exemplars and simultaneously define and represent malware families. Leaf exemplars must only have binaries as children and not any other exemplars. There is no restriction that the hierarchy must be balanced and if the hierarchy is lopsided binaries in different families may have very different path distances from the root.

This hierarchy design is based on the concept of a dendogram, which is the output of many hierarchical clustering algorithms. In a dendogram, like our hierarchy, each node represents some subset of the dataset, with the most general nodes at the top and the most specific nodes at the bottom. Each level represents some clustering of the leaves. As you traverse up the tree from the leaf exemplars, the exemplars progressively represent more and more families, and the cohesion between the binaries within those families becomes weaker. Finally, the root represents the entire corpus of malware.

This structure provides some very important insight into the data. Specifically, distance in the hierarchy reflects dissimilarity of the malware. For example, any given exemplar will typically be more similar to an exemplar that shares its parent than an exemplar that shares its grandparent. This relationship holds for binaries as well. We use this insight to our advantage to represent malware families as more binaries are introduced.

\subsection{Online Operation}
\label{online}

The hierarchical organization of malware is an effective approach to managing the volume of data, but we also need to consider the velocity at which new binaries are discovered and added to the system. As new malware binaries are rapidly encountered, we need a means to incorporate them into the hierarchy that is fast, flexible and accurately places new binaries into new or existing families.

One naive approach to managing the influx of new binaries is to re--cluster the entire database upon arrival of some minimum amount of new binaries to create a new hierarchy. This has the obvious disadvantage that large amounts of the hierarchy may not change and unnecessary work is being performed at each re--clustering. Another naive approach is to add new malware binaries to existing families in the hierarchy using some classification algorithm. This too is problematic, as the families would be fixed. For example, two malware binaries that are originally assigned to different families would never end up in the same family, even if the influx of new binaries later suggested that their differences were relatively insignificant and that they should in fact be siblings.

\begin{figure*}[t]
\centering
\subfigure[Insertion: A new malware binary ($S_{new}$) is inserted into the hierarchy below $E_5$, the closest leaf exemplar according to the distance function $\delta$. The exemplar has features $F_5$, and it is marked as modified. When all new binaries have been classified, a post--order traversal of the hierarchy is initiated.]{
\includegraphics[width=.45\textwidth]{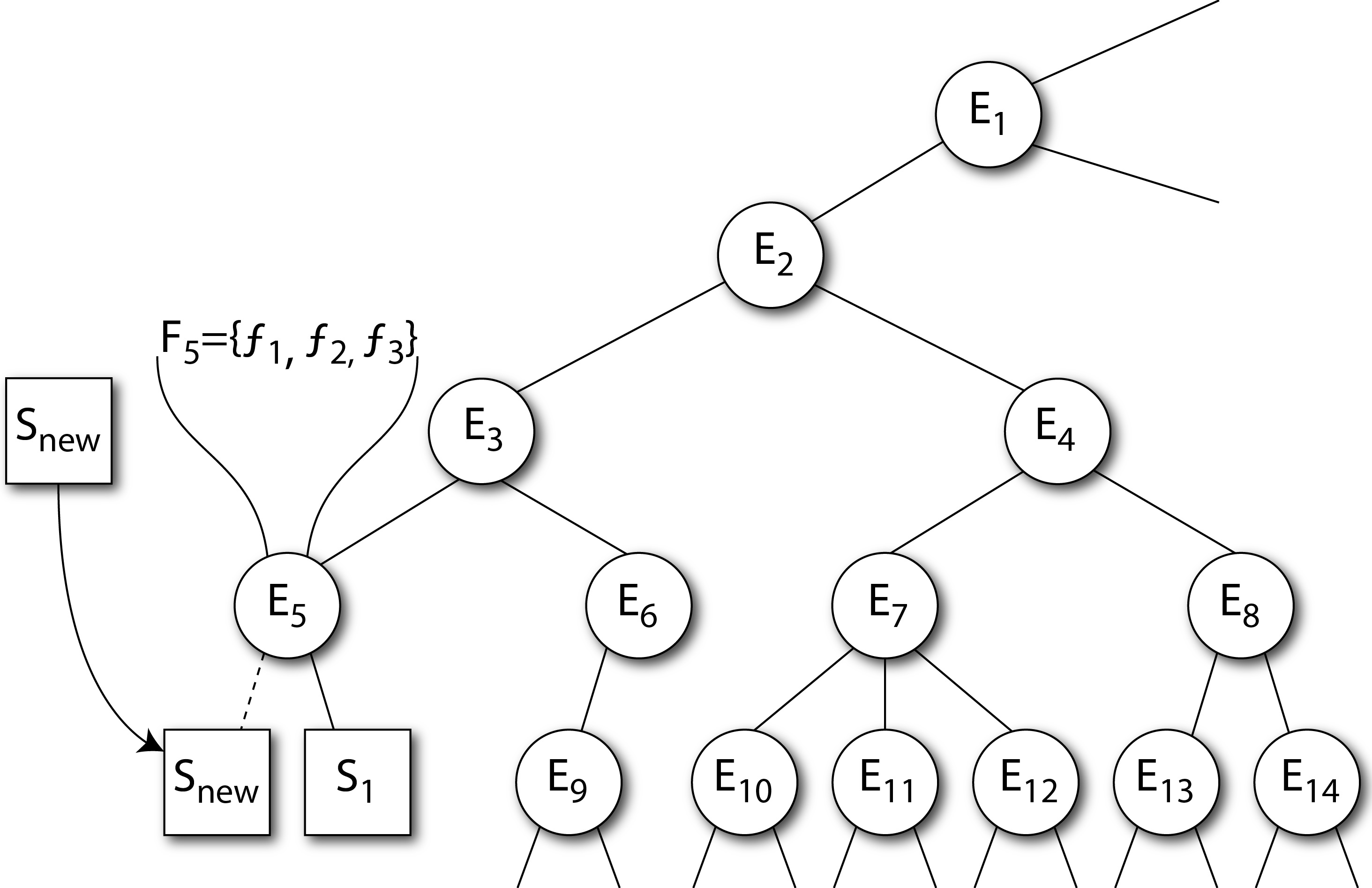}
\label{fig:online1}
}\hspace{0.40em}
\subfigure[Comparison and Marking: When the modified parent $E_5$ is reached during the post--order traversal, its new features $F_5'$ are calculated using the $\theta$ function over its children. If the distance function $\delta$ determines that the node has changed enough since it was originally created, its $(d_r - 1)$\textsuperscript{th} ancestor $E_2$ is marked for re-clustering.]{
\includegraphics[width=.45\textwidth]{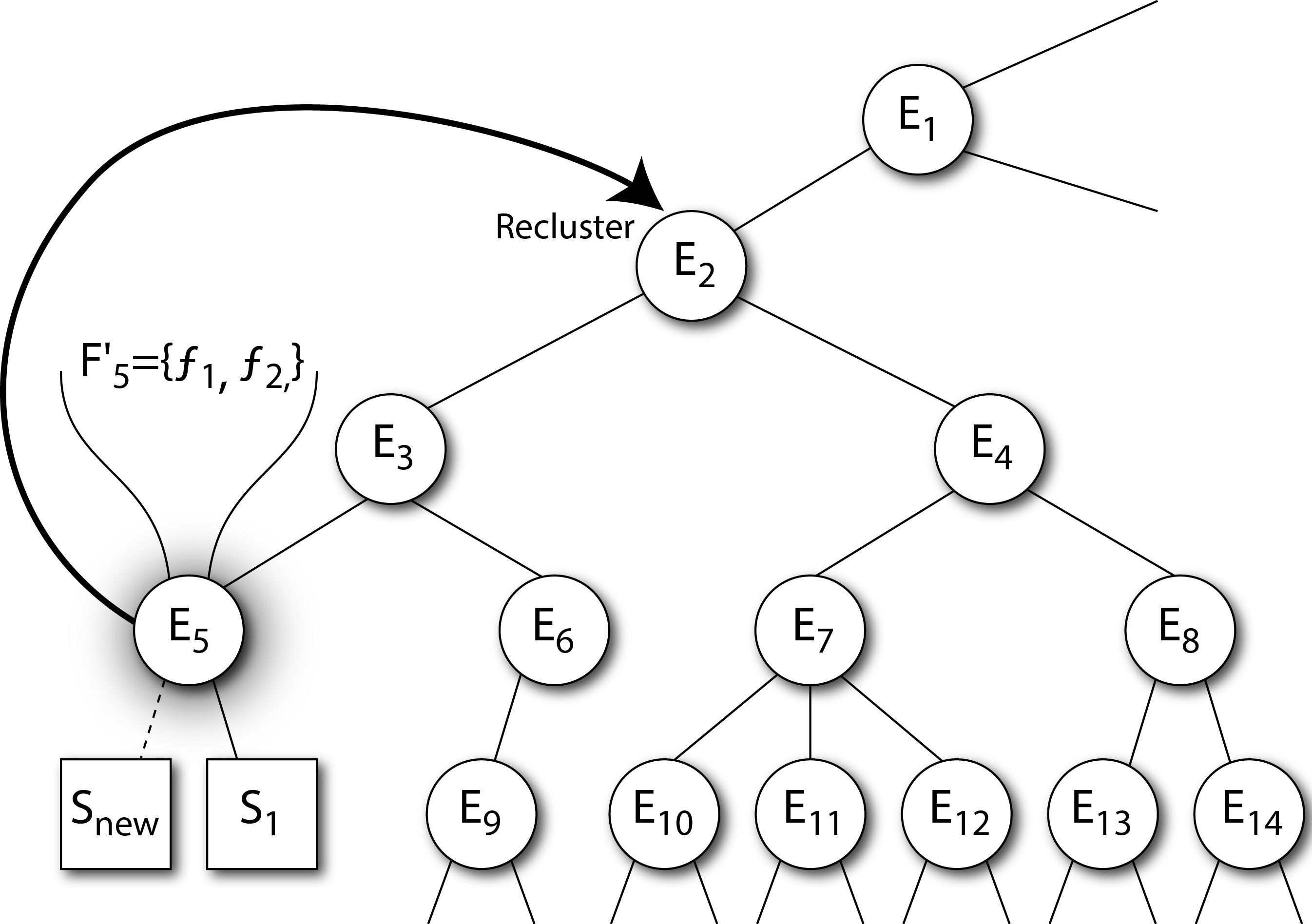}
\label{fig:online2}
}\hspace{0.40em}
\subfigure[Re-clustering: When the post--order traversal reaches $E_2$, which is marked for re-clustering, its $d_r$\textsuperscript{th} descendants are passed to the clustering algorithm, in this case its great--grandchildren. The old exemplars in between are removed.]{
\includegraphics[width=.45\textwidth]{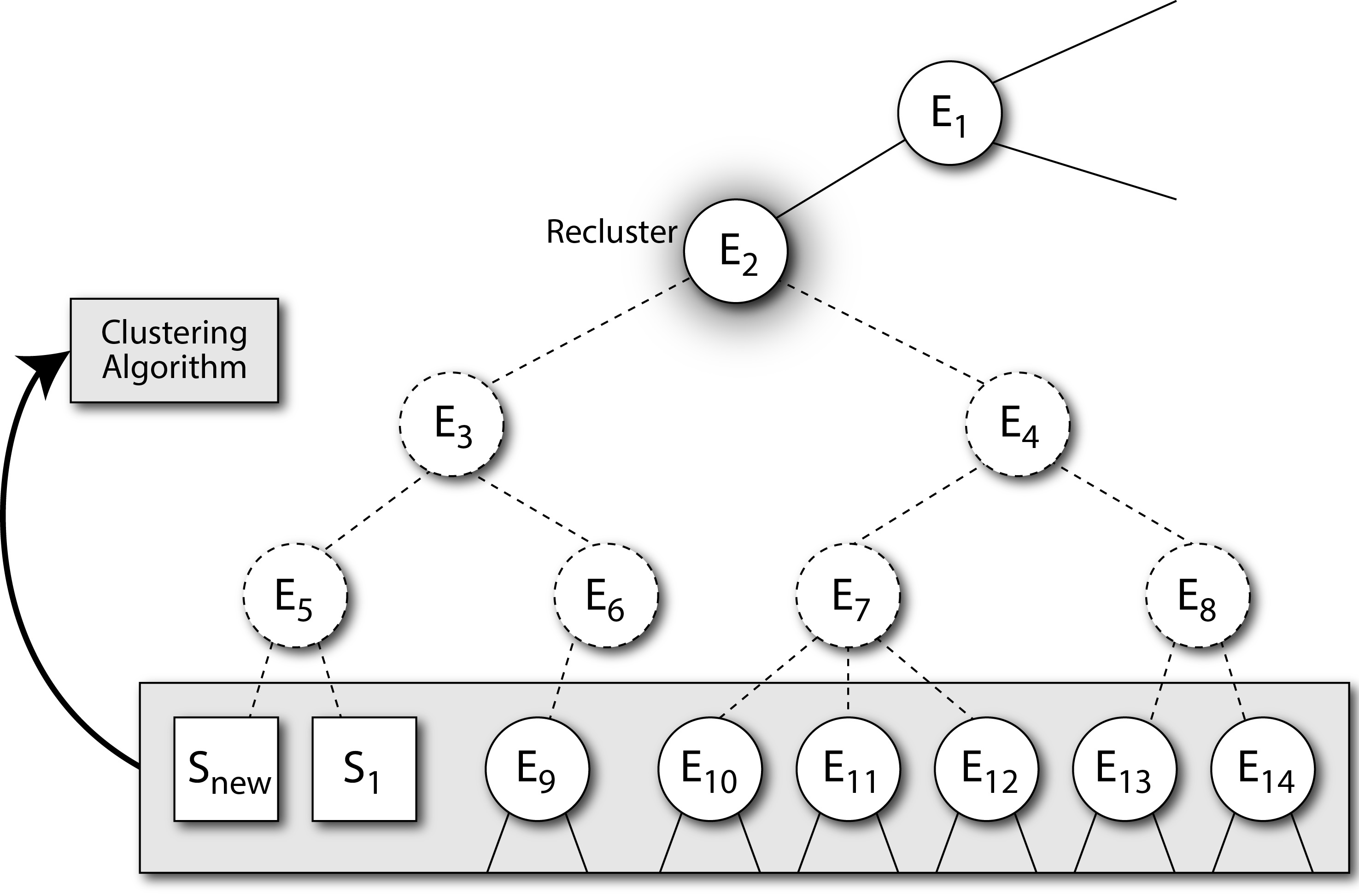}
\label{fig:online3}
}\hspace{0.40em}
\subfigure[Sub-hierarchy Replacement: The clustering algorithm returns a sub-hierarchy, which is inserted back into the full hierarchy below $E_2$, replacing the old structure. There is no restriction on the depth of the returned sub-hierarchy. The features of $E_2$ are unchanged.]{
\includegraphics[width=.45\textwidth]{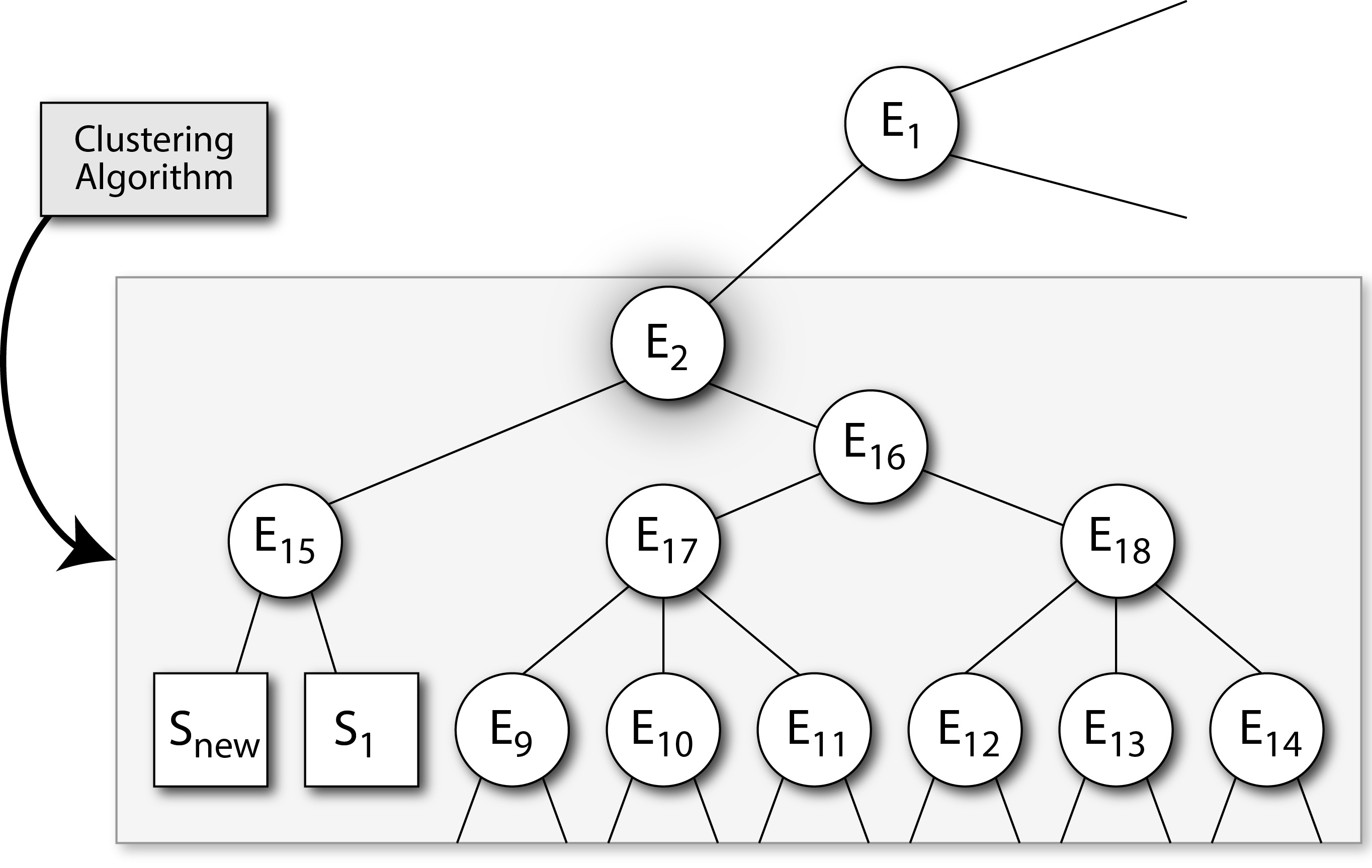}
\label{fig:online4}
}%\hspace{0.25em}
\vspace{0mm}
\caption[]{A diagram of the online hierarchical clustering algorithm operating over a portion of the hierarchy. In this case the re-clustering depth $d_r$ is set to 3.\label{fig:online}}
\end{figure*}

 % TODO: 
 %	- \mathbb{F} should be changed to a simple capital italic F
 %	- Apostrophes should be changed to primes
 %	- Node that is reclustered on should be named
 %	- Arrows should only be for traversal, or they should be different styles
 %	- Dotted vs. solid line in Figure 3
 % 	- Show parent as modified
 % 	- Show rest of traversal

To alleviate these problems, we propose an incremental classification and clustering algorithm that is run online as new binaries are ingested by the system. The main idea behind our incremental update is that we only re--cluster portions of the hierarchy that have significantly changed since the last update, and we have designed the re--clustering process so that no malware binary is ever fixed in a family. A summary of the algorithm is illustrated in Fig.~\ref{fig:online}. 
%In summary, the process is as follows:
%\begin{enumerate}
% \item Upon collecting a minimum number of new malware binaries, initiate the incremental classification and clustering process.
% \item Classify each new binary into an existing family using the exemplars from each family (leaf exemplars).
% \item Visit each exemplar in the hierarchy using a post--order traversal.
% \item At each exemplar, recompute its features if any of its children have changed.
% \item If the difference between the current features of the exemplar and its previous features is greater than some threshold, mark an ancestor of the exemplar for re--clustering.
% \item If an exemplar is marked for re--clustering, re--cluster the descendants at a particular depth below the exemplar, creating a new sub--hierarchy extending from the exemplar to those descendants which replaces that section of the hierarchy.
%\end{enumerate}

\subsubsection{ Algorithm Basics}

Before we go into more detail on the incremental clustering algorithm operation, we first define some basic terms and concepts. Let $\mathcal{H}$ be some hierarchy of malware binaries, where each node $i$ in the hierarchy is represented by an exemplar $E_i$ and let $\mathbb{C}^n_i = \{E_1, \dots, E_k\}$ denote the descendants of exemplar $E_i$ at depth $n$. Furthermore, let $F_i$ denote the features of $E_i$ in some feature space $\mathbb{F}$. $F^0_{i}$ are the original features of $E_i$ at its time of creation. We further define a function $\theta:\mathbb{F}\times\dots\times\mathbb{F}\rightarrow\mathbb{F}$ that aggregates features from a set of exemplars into a single feature representation. We now define what it means for a function $\theta$ to be \textit{transitive} on the hierarchy.

\begin{definition}
\label{def:comm}
\textbf{Transitive Feature Function.} A function $\theta$ is transitive on hierarchy $\mathcal{H}$ if for every exemplar $E_i$,
\begin{equation*}
\begin{split}
 & \theta(F_{1,n}, \dots, F_{k,n}) = \theta(F_{1,m}, \dots, F_{j,m})\mbox{ where } \\
 & E_{1,n}, \dots, E_{k,n} \in \mathbb{C}^n_i, \,\, E_{1,m}, \dots, E_{j,m} \in \mathbb{C}^m_i,
\end{split}
\end{equation*}
\end{definition}
That is, a transitive function on the hierarchy is one where application of $\theta$ to the features of an exemplar's children at depth $n$ results in a new feature that is the same as the application of $\theta$ to children at depth $m$. Two examples of transitive functions are the intersection function and the weighted average function, described in Sec.~\ref{impl}.

In addition to $\theta$, there are three additional parameters that are used in the clustering algorithm. The first, $d_r$, limits the portion of the hierarchy that is re--clustered when an exemplar has been marked for re--clustering. $\delta$ is a distance function defined between the features $F_i$ and $F_j$ of two exemplars (or binaries). Finally, $\tau$ specifies a distance threshold that triggers a re--clustering event at some exemplar.

\subsubsection{Algorithm Detail}

When a sufficient amount of new malware binaries have been collected, the algorithm is initiated to integrate them into the hierarchy. Note that the timing of this process is subjective and different uses of the hierarchy may require more or less initiations of the algorithm. The first thing the algorithm does is to insert each new malware binary into an existing family in the hierarchy (Fig. \ref{fig:online1}). As the existing families in the hierarchy are defined as the leaf exemplars, the distance from each new malware binary to each leaf exemplar is computed using the distance function $\delta$, and a binary is added as a child to the minimum distance exemplar. When a binary has been added to a leaf exemplar, the exemplar is marked as modified.

After the classification step, we perform a post--order traversal of the hierarchy. A post--order traversal is a depth--first pass in which each node is only visited after all of its children has been visited. At each node, we execute a function, \textit{Update}, that updates the exemplars with new features and controls when to re--cluster portions of the hierarchy. The first thing the function does is to check if the exemplar has been modified, meaning it has had new malware inserted below it. If it has been modified, it marks its parent as modified as well. It then creates the new features of the exemplar, $F_i'$, by combining the child exemplars of the node using the $\theta$ function. If the distance between $F_i'$ and the original features of the exemplar, $F^0_{i}$, is greater than $\tau$, when we mark the $(d_r - 1)$\textsuperscript{th} ancestor as needing re--clustering (Fig. \ref{fig:online2}).

Finally, we check if the exemplar has been marked for re--clustering. If so, we perform a clustering operation using $\mathbb{C}^{d_r}_i$, the children at a depth of $d_r$ from $E_i$ (Fig. \ref{fig:online3}). The clustering operation returns a sub--hierarchy $\mathcal{H}'$, which has $E_i$ as its root and $\mathbb{C}^{d_r}_i$ as its leaves. The old exemplars between $E_i$ and $\mathbb{C}^{d_r}_i$ are discarded, and $\mathcal{H}'$ is inserted back into $\mathcal{H}$ between $E_i$ and $\mathbb{C}^{d_r}_i$ (Fig. \ref{fig:online4}). There are no assumptions made about the operation and algorithm employed in the re--clustering step. We only enforce that the clustering operation returns a sub--hierarchy with $E_i$ as its root and $\mathbb{C}^{d_r}_i$ as its leaves; the internal makeup of this sub--hierarchy is completely dependent upon the clustering algorithm employed. This means that the clustering algorithm is responsible for creating new exemplars between the leaves and the root of the sub--hierarchy (and hence we must pass $\theta$ to the algorithm). We also do not enforce that the depth of the new sub--hierarchy equals $d_r$. That is, after re--clustering, the distance from $E_i$ to any one of $\mathbb{C}^{d_r}_i$ may no longer be $d_r$ -- it could be more or less depending upon the clustering algorithm and the data.

There are several key properties of the algorithm that enable it manage the large amounts of malware effectively. First, the $d_r$ parameter controls how local changes in portions of the hierarchy are propagated to the rest of the hierarchy. For instance, if $d_r = \infty$, then any change in the hierarchy will trigger a re--clustering event on \textit{all} malware binaries in the hierarchy; such an operation can be extremely time consuming with a large amount of binaries. Using a reasonable value of $d_r$, however, lets us slowly propagate changes throughout the hierarchy to regions where the impact of the changes will be felt the most. Clearly, higher values of $d_r$ will allow modified exemplars to potentially cluster with a larger set of exemplars, but at the cost of higher computational effort.

Another key property is the ability to create new structure in the hierarchy upon re--clustering a set of exemplars. If we enforced that upon re--clustering, the distance between $E_i$ and $\mathbb{C}^{d_r}_i$ remained fixed at $d_r$, then we would never add any new families to the hierarchy. This property is critical, as certain parts of the hierarchy may contain more binaries which should be reflected in a more diverse hierarchical organization.

The detailed pseudocode for the algorithm is in Fig. \ref{fig:alg}. In the pseudocode, the $applyPostOrder$ function applies the \textit{Update} function to the exemplars in $\mathcal{H}$ in post--order fashion, passing along the arguments $d_r$, $\theta$, $\delta$, and $\tau$ (line \ref{postOrder}). The $ancestor(E_i,n)$ function retrieves the $n$\textsuperscript{th} ancestor of $E_i$ (lines \ref{modified}, \ref{mark}). The $cluster(\mathbb{C}^{d_r}_i, \delta, \theta)$ function uses a standard hierarchical clustering algorithm to cluster the exemplars or binaries at a depth of $d_r$ below the exemplar $E_i$. Finally, the $insert(\mathcal{H}, \mathcal{H}', E_i, d_r)$ function inserts a sub-hierarchy $\mathcal{H}'$ into the full hierarchy $\mathcal{H}$, replacing the nodes between $E_i$ and $\mathbb{C}^{d_r}_i$ (line \ref{insert}).

\begin{figure}
\centering
\begin{algorithmic}[1]
\State $\mathbb{M} \gets$ new malware
%\State $\mathbb{L} \gets$ leaf exemplars

\Procedure{IncrementalCluster}{$\mathbb{M}, \mathcal{H}, d_r, \theta, \delta, \tau$}\label{incrProc}
\State $\mathbb{L} \gets \mbox{leaf exemplars} \in \mathcal{H}$
\For{$m \in \mathbb{M}$}\label{classStart}
  \State $E_{min} \gets \underset{E_i \in \mathbb{L}}{argmin}\,\, \delta(m, E_i)$\label{minLeaf}
  \State $addChild(E_{min}, m)$ \label{addChild}
  \State mark $E_{min}$ as modified
\EndFor\label{classEnd}
\State $applyPostOrder(\mathcal{H}, d_r, \theta, \delta, \tau, \mbox{Update})$\label{postOrder}
\EndProcedure

\State
\Procedure{Update}{$E_i, d_r, \theta, \delta, \tau$}\label{update}
\If {$isModified(E_i)$}\label{isModified}
  \State mark $ancestor(E_i, 1)$ as modified\label{modified}
  \State $F_i' \gets \theta(\mathbb{C}^1_i)$\label{features}  
  \If {$\delta(F_i', F^0_{i}) > \tau$}\label{distCheck} 
    \State mark $ancestor(E_i, d_r - 1)$ for re--clustering\label{mark}
  \EndIf
\EndIf
\If {$needsReclustering(E_i)$} \label{needsReclustering}
  \State $\mathcal{H}' \gets cluster(\mathbb{C}^{d_r}_i, \delta, \theta)$\label{recluster}
  \State $\mathcal{H} \gets insert(\mathcal{H}, \mathcal{H}', E_i, d_r)$\label{insert}
\EndIf
\EndProcedure
\end{algorithmic}
\caption{The incremental classification and clustering algorithm.\label{fig:alg}}
\end{figure}

\subsubsection{Clustering Correctness}

The hierarchy is a representation of a database of malware, and changes to the hierarchy should only reflect changes in the underlying database. That is, an exemplar should be marked for re--clustering because of some change in the composition of the malware database, and \textit{not} due to the implementation of the clustering algorithm. This means that when the clustering algorithm returns $\mathcal{H}'$, the features of the leaves of $\mathcal{H}'$ and the root (i.e., $E_i$ that initiated the re--clustering) must remain unchanged from before the clustering. For instance, if the clustering algorithm returns $\mathcal{H}'$ with the root features changed, then we may need to mark one of $E_i$'s ancestors for re--clustering, but this re--clustering was triggered by how $\mathbb{C}^{d_r}_i$ was clustered, and not because of any significant change in the composition 
of the malware below $E_i$. This idea leads us to Thm.~\ref{thm:1} below.

\begin{theorem}
\label{thm:1}
 Let $F_i$ be the features of exemplar $E_i$ whose children at depth $d_r$ need to be re--clustered (i.e., the root of $\mathcal{H}'$) and let $F_i'$ be the features of $E_i$ after re--clustering. If the re--clustering algorithm uses $\theta$ to generate the features of exemplars, and $\theta$ is a transitive feature function, then $F_i = F_i'$.
\end{theorem}
\begin{proof}
 The proof follows naturally from the definition of a transitive feature function. Let $F^{d_r}_i = \theta(\mathbb{C}^{d_r}_i)$  be the features computed for $E_i$ using the children of $E_i$ at depth $d_r$. After re--clustering, let ${F^{d_r-1}_i}'$ be the features computed for $E_i$ using children of $E_i$ at depth $d_r-1$. By the definition of a transitive feature function,
\begin{equation*}
\begin{split}
  & F^{d_r}_i = F^{d_r-1}_i = \dots = F^{1}_i\\
  & F^{d_r}_i = {F^{d_r-1}_i}' = \dots = {F^{1}_i}'
\end{split}
\end{equation*}
As $F^{1}_i = {F^{1}_i}'$, then $F_i$ remains unchanged before and after clustering. Hence, if $\theta$ is a transitive function, then it is guaranteed that the insertion of the new sub--hierarchy into $\mathcal{H}$ will itself not initiate any new re--clustering.
\end{proof}

\subsubsection{Clustering Time Complexity}

It would also be beneficial to understand how the theoretical time complexity of our incremental clustering algorithm compares with the time complexity of the true clustering of the database. Without loss of generality, we assume the hierarchy $\mathcal{H}$ has a branching factor $b$ and let $\mathcal{H}_d$ indicate the total depth of the hierarchy. We also assume that we have some clustering algorithm that runs in complexity $O(n^p)$, for some number of exemplars $n$ and exponent $p$. For all of this analysis, we take a conservative approach and assume that each time the incremental clustering algorithm is invoked with a new batch of malware, that the entire hierarchy is marked for re--clustering.

We first consider the basic case and compare creating a hierarchy from a single, large clustering operation to creating one incrementally. This is not a realistic scenario, as the incremental algorithm is intended to operate on batches of malware over time, but it serves to base our discussion of time complexity. With a re--clustering depth of $d_r$, on the entire hierarchy, we would perform $b^{\mathcal{H}_d-d_r+1}$ clustering operations. Each operation in turn would take $O(b^{p \cdot d_r})$ time. In contrast, clustering the entire database at once would take $O(b^{p \cdot \mathcal{H}_d})$ time. So our incremental clustering algorithm is faster if:`
\begin{equation}
\begin{split}
 b^{(\mathcal{H}_d-d_r+1)} b^{p d_r} & < b^{p \mathcal{H}_d}\\
 (\mathcal{H}_d-d_r+1)+(p d_r) & < p\mathcal{H}_d\\
 d_r(p-1)+ 1 & < \mathcal{H}_d(p-1)
\end{split}
\end{equation}
So, for any value of $p > 1.5$ our method is faster if $d_r$ is at most $\mathcal{H}_d-2$, which is a very reasonable assumption. In practice we typically use $d_r$ values of 2 or 4, and $\mathcal{H}_d$ might be on the order of 100 edges deep when clustering 10,000 binaries collected from the wild, for example.

Of course, we don't perform a single clustering operation but perform incremental clustering online when a sufficient number of new malware binaries have been collected. For purposes of comparison, however, let assume that there are a finite amount of binaries that comprise $\mathcal{H}$. Let us further assume that the incremental clustering is broken up into $S$ batches, so that each time the algorithm is run, it is incorporating $\frac{b^{\mathcal{H}_d}}{S}$ new malware binaries into the database. That is, the first time incremental clustering is run, we perform $\frac{b^{\mathcal{H}_d-d_r+1}}{S}$ clustering operations, the second time  $\frac{2b^{\mathcal{H}_d-d_r+1}}{S}$ clustering operations and so forth. In this case, our method is faster if:
\begin{equation}
\label{eqn:2}
\begin{split}
 \frac{1}{S} b^{(\mathcal{H}_d-d_r+1)} b^{p  d_r} + \frac{2}{S} b^{(\mathcal{H}_d-d_r+1)} b^{p  d_r} + \dots & < b^{p  \mathcal{H}_d}\\
 b^{(\mathcal{H}_d-d_r+1)} b^{p  d_r} (\frac{1}{S} + \frac{2}{S} + \dots) & < b^{p  \mathcal{H}_d}\\
 \frac{1}{S} b^{(\mathcal{H}_d-d_r+1)} b^{p  d_r} \sum_{i=1}^S i & < b^{p  \mathcal{H}_d}\\
 \frac{1}{S} b^{(\mathcal{H}_d-d_r+1)} b^{p  d_r} \frac{S(S+1)}{2} & < b^{p  \mathcal{H}_d}\\
 \frac{(S+1)}{2} b^{(\mathcal{H}_d-d_r+1)} b^{p  d_r}  & < b^{p  \mathcal{H}_d}\\
 \frac{log\, \frac{(S+1)}{2}}{log\, b} +  d_r(p-1)+ 1 & < \mathcal{H}_d(p-1)
\end{split}
\end{equation}
In this situation, we add a constant time operation to our method that depends on the number of incremental clustering batches and the branching factor of the hierarchy. While it is difficult to find a closed form expression of when our method is faster in this situation, we note that in Eqn.~\ref{eqn:2}, the single clustering method scales linearly with the overall depth of the hierarchy, whereas the incremental method scales logarithmically with the number of batches (assuming a constant $b$). When all the other parameters are known, we can determine the maximum value of $S$ where our method is faster. This indicates that in an operational system, we should not initiate the incremental clustering operation too often or we would incur too much of a penalty.

\subsection{Implementation}
\label{impl}
A number of details need to be specified for an actual implementation of the system, including a way to represent the items in memory, a method of comparing those representations, a $\theta$ function for aggregating them to create representations of exemplar nodes, and a hierarchical clustering algorithm.

\subsubsection{Feature--Based Representation}

The feature--based representation of the binaries and exemplars is controls how they will be both compared and aggregated.  For brevity, we only provide a brief overview of the feature generation and representation for malware binaries, and refer the reader to previous works for a complete description of malware feature extraction and analysis~\cite{lakhotia2013fast, pfeffer2012malware, ruttenberg2014}.

While our system is defined to handle different types of features, a modified bag--of--words model of representing an item's distinct features provides the basis for our implementation. We extract three types of features from the malware to use as the words. The first two are the strings found in a binary and the libraries it imports. The third type of feature are n--grams over the Procedure Call Graph (PCG) of the binary. The PCG is a directed graph extracted from a binary where the nodes are procedures in the binary and the edges indicate that the source procedure made a call to the target procedure. Each n--gram is a set of $n$ procedure labels indicating a path through the graph of length $n-1$. In our implementation, we used 2--grams, which captures the content and structure of a binary while keeping the complexity of the features low.

The labels of the n--grams are based on MinHashing~\cite{broder1997resemblance}, a technique typically used to approximate the Jaccard Index \cite{jaccard1912distribution} between two sets. Each procedure is composed of a set of code blocks, and each block is represented by its BinJuice ~\cite{lakhotia2013fast}, a representation of the semantic function of a code block that also abstracts away registers and memory addresses. A hash function is used to generate a hash value for the BinJuice of every block in a procedure, and the lowest value is selected as the label for that procedure. The equivalence between procedure's hash value is an approximation of the Jaccard index between the procedures, hence we use more hash functions (four) for a better approximation. As a result, we have multiple bags of n--grams, one per hash function. So in our implementation a malware binary is represented by four bags of n--grams over the PCG, along with a bag of its distinct strings, and another bag of its imported libraries. 

%we used a MinHash  over the basic blocks of the procedure \cite{broder1997resemblance}. MinHash is typically used to estimate Jaccard Index \cite{jaccard1912distribution} between sets, and is used here to represent a set of items (basic blocks) with a label based on a subset of its contents. Through this approach, a hash function is used to hash every block in the procedure, with the lowest value used as the label for that procedure. The value of the basic block that the hash function is run over is the BinJuice \cite{lakhotia2013fast}, which is a generalized interpretation of the semantics of the block's code. In MinHash, using more hash functions results in a closer approximation of Jaccard Index, and here we can use multiple hash functions to build a more accurate labeling over the procedures. This results in having multiple bags of NGrams, one per hash function. We used four hash functions, so in our implementation a malware binary is represented by four bags of NGrams over the PCG, along with a bag of its distinct strings, and another bag of its imported APIs. We were ultimately able to represent procedures using the same data structure, which allowed us to reuse our comparison and combination methods. Procedures are represented by a single NGram per hash function, where that NGram is a 1-gram containing the hash of that procedure. We do not include strings or API calls in the procedure representation.

\subsubsection{Creating Exemplars and Defining Distance}

We also must define methods to create exemplars from a set of binaries or other exemplars (e.g., $\theta$). As explained in Sec.~\ref{online}, we must define $\theta$ as a transitive function for proper behavior of the algorithm. We propose two methods in our implementation; intersection and weighted average. For intersection, $\theta$ is defined simply as $\theta(F_{1}, \dots, F_{n}) = F_{1} \wedge F_{2} \wedge \dots \wedge F_{n}$.

Slightly more complicated is the definition of average--based exemplars, as it is necessary to consider the number of malware binaries represented by each child of an exemplar. In this case, we must also assign a weight to each word in the bag--of--words model. We define $\theta$ for the weighted average method as
\begin{equation*}
\begin{split}
\theta(F_{1}, \dots, F_{n}) = \{(f, w_f)\}, \,\,
w_f = \frac{\sum\limits_{i=1}^n s_i*w_{f,i}}{\sum\limits_{i=1}^n s_i}
\end{split}
\end{equation*}
where $w_{f,i}$ is the weight of feature $f$ in feature set $F_i$ (or 0 if $f \notin F_i$), and $s_i$ is the number of actual malware binaries represented by $F_i$.

%Our feature--based representation is actually a slightly modified version of a bag-of-words model. The difference was introduced to support one of the two methods we use to represent the features of the exemplars in the hierarchy, whose features must reflect all of the binaries, or procedures, below them. One method is to represent exemplars as the intersection of the features of their descendants, or only the features that all of their descendants share. The other is to use the average features of the descendants. Using the average means that, for example, if half of an exemplars descendant binaries had a particular feature, the exemplar would be said to have that feature, but with a weight of 0.5. To support this method, our bags of words contain a weight alongside each feature. These weights will always be 1.0 for the malware binaries and procedures themselves.

The distance function, $\delta$, is needed for classification, sub--hierarchy re-clustering, and the exemplar distance threshold calculation. Since exemplars must be compared to the actual malware binaries, as well as to other exemplars, the comparison method is also driven by the exemplar representation method. In our implementation, we use a weighted Jaccard distance measure on the bag--of--word features~\cite{manasse2007consistent}. The weighted Jaccard distance is calculated by dividing the sum of the minimum weights in the intersection (where it is assumed the lack of a feature is a weight of 0), by the sum of maximum weights in the union. Note that in the case where $\theta$ is the intersection function, the weights of feature words are all one, and hence the weighted Jaccard is the normal Jaccard distance. The distance between each of the three bags of features (i.e., n--grams, strings and libraries) for two binaries/exemplars are computed using $\delta$, and these three distances are combined using a weighted average to determine the overall distance.

\subsubsection{Other Parameters}

Beyond the preferred measure of distance and exemplar creation, the user is faced with additional decisions in regard to the preferred clustering algorithm, along with its settable parameters. The Neighbor Join \cite{saitou1987neighbor} and Louvain methods \cite{blondel2008fast} are both viable hierarchical clustering methods implemented in our system. While both methods have proven to be successful, specific conditions or preferences may lead a user to run one instead of the other.

Another adjustable parameter comes from the fact that the hierarchical clustering algorithms return a hierarchy that extends down to each individual item, so no items share a parent. At this point every item is in its own family. In order to produce a hierarchy that represents the optimal family assignments, the hierarchy needs to be flattened by merging families from the bottom up to the point where the items are optimally grouped under the same parent nodes. To flatten a hierarchy we perform a post--order traversal, and at each exemplar a condition is applied that determines whether to delete its children and become the parent of its grandchildren, placing them into the same family and flattening that part of the hierarchy by one level (note this will also not change the exemplar features due to exemplar transitivity). The condition we apply is to flatten whenever $\lambda * cohesion(\mathbb{C}^{2}_i) < cohesion(\mathbb{C}^{1}_i)$, where $\mathbb{C}^{n}_i$ are all the children at a depth $n$ from $E_i$, and $cohesion(\mathbb{C}^{n}_i)$ is defined as the average of the distances between all pairs of exemplars in $\mathbb{C}^{n}_i$. The other value in this condition, $\lambda$, is another parameter that can be adjusted in order to influence the flatness of the hierarchy. This operation is performed at the end of the re--clustering step before a sub--hierarchy is inserted back into the hierarchy.

The other parameter worth discussing is the re-clustering depth parameter, $d_r$, which controls the depth of the sections of hierarchy modified during re-clustering. Initially, we performed re-clustering with $d_r=2$. In more recent experiments, we have obtained higher clustering performance using a greater depth. These experiments, among others, are discussed in the following section.

\subsection{Experiments}
\label{expr}

To test our system, we performed different types of experiments to evaluate its effectiveness in several areas including the scalability of the system, the correctness of the families it produces, its sensitivity to various parameters, and its ability to adapt to changes in implementation details and how these changes affect accuracy and speed. Here we will describe the results of the correctness experiments.

\subsubsection{Datasets and Metrics}

For our correctness tests, we used a set of 8,336 malware binaries provided by MIT Lincoln Labs during the Cyber Genome program's Independent Verification and Validation (IV\&V). The dataset included 128 clean-room binaries built by MIT Lincoln Labs comprising six distinct families of malware. These 128 binaries provided the ground truth. The others were collected from the wild and provided background for the test.

The accuracy of a clustering was determined by running cluster evaluation metrics over the cluster assignments for the 128 ground truth binaries. The metrics we used were Jaccard Index, Purity, Inverse Purity, and Adjusted Rand Index. Jaccard Index is defined as the number of pairs of items of the same truth family that were correctly assigned to the same cluster, divided by the number of pairs that belong  to the same truth family and/or assigned cluster. Purity is defined as the percentage of items that are members of the dominant truth family in their assigned cluster, and Inverse Purity is the percentage of items that were assigned to their truth family's dominant cluster. Adjusted Rand Index is a version of the Rand Index cluster accuracy measure that is adjusted to account for the expected success rate when clusters are assigned at random.

\subsubsection{Base Accuracy Results}

To evaluate the correctness of our incremental system, we compared the clustering produced by our system to that produced by simply clustering the entire dataset at once, both using the Neighbor Join algorithm for clustering. The incremental run used batches of 1,000 binaries and set $d_r$ to 6 (to put that into perspective, a resulting cluster hierarchy over this dataset is typically on the order of 50-100 generations deep at its maximum). Fig.~\ref{fig:correct_eval} shows the evaluation metrics for both methods. Both score well compared to the families on the 128 ground truth binaries, where the incremental method is only slightly worse than the single batch version. The entire incremental result (all 8,336 binaries) was also scored using the single--batch run as the truth, which shows the overall high similarity between the two cluster assignments for the entire dataset, not just the clean room binaries.

\begin{figure*}[t]
\centering
\subfigure[Correctness evaluation]{
\includegraphics[width=.48\textwidth]{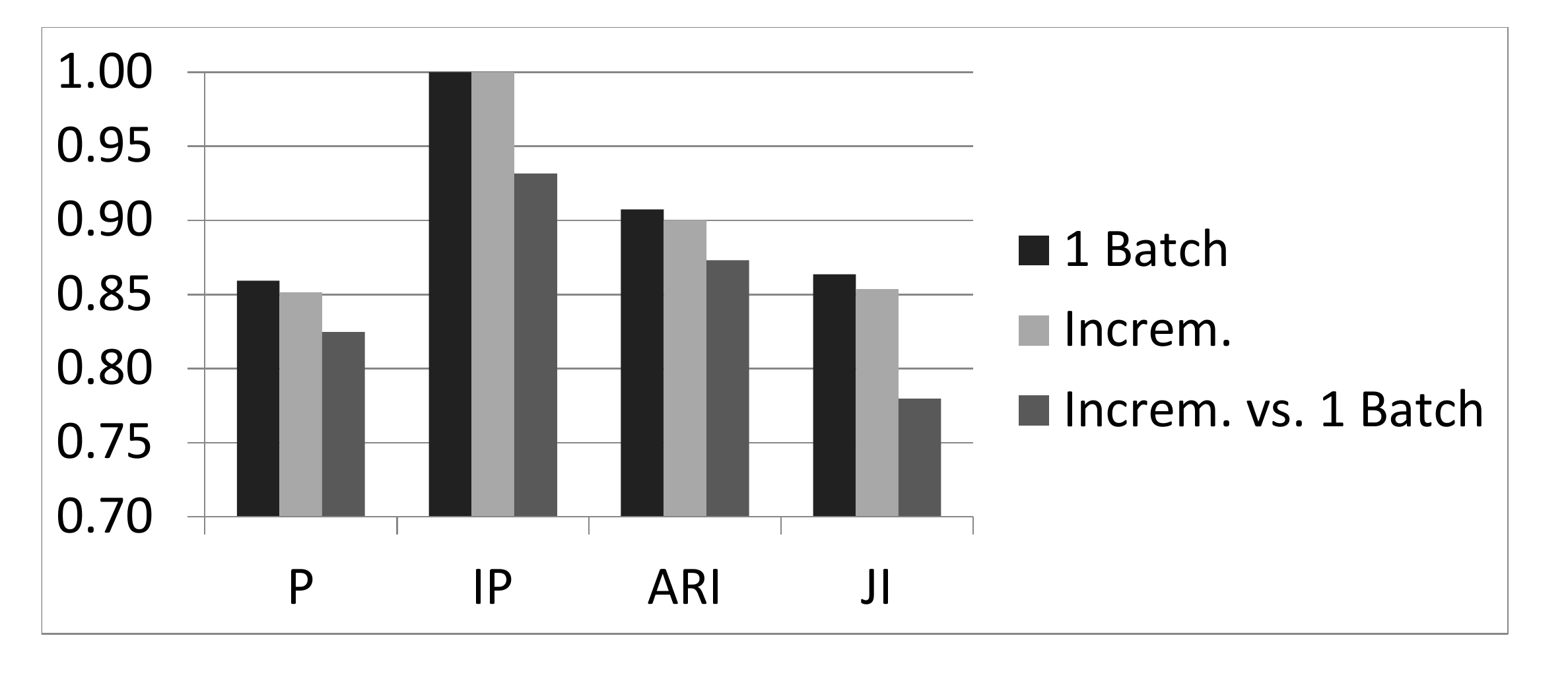}
\label{fig:correct_eval}
}
\subfigure[Correctness run time]{
\includegraphics[width=.48\textwidth]{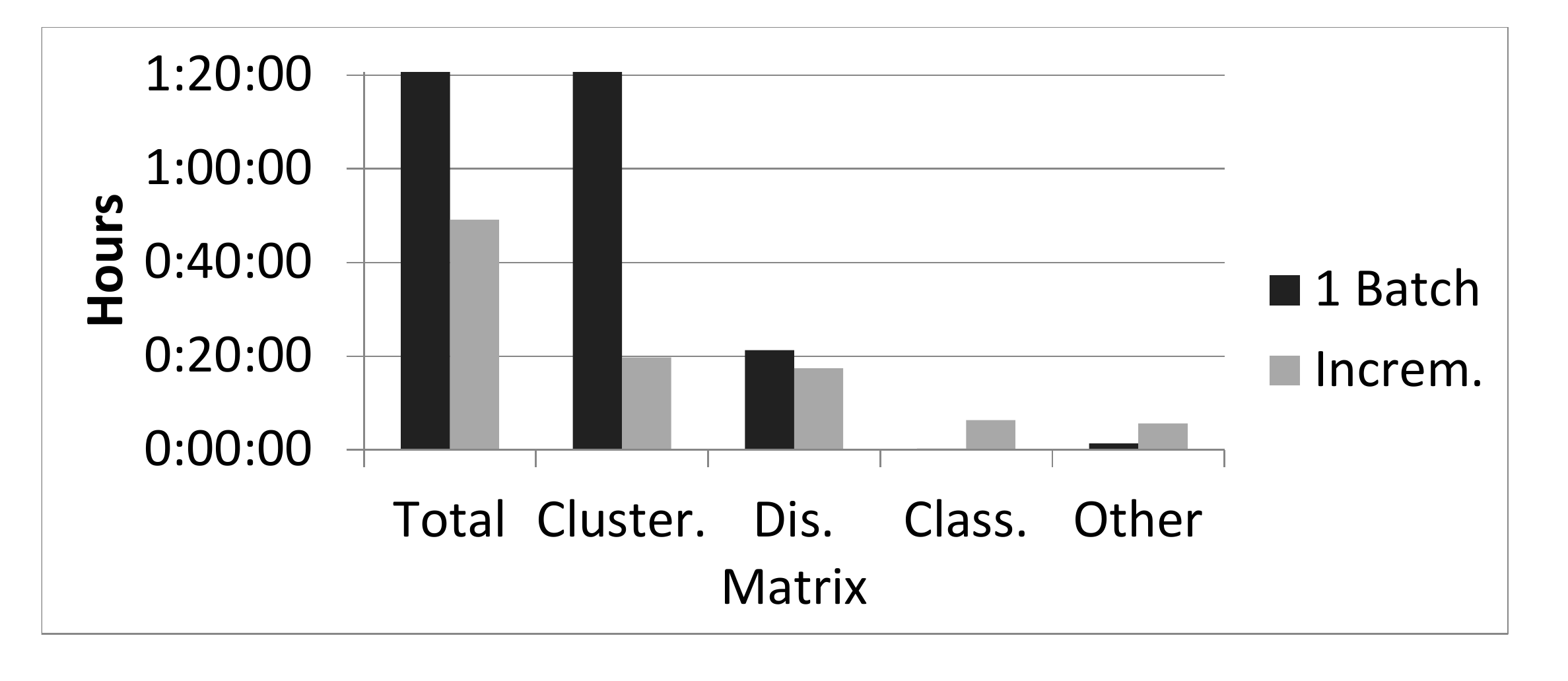}
\label{fig:correct_timing}
}
\vspace{-3mm}
\caption[]{Correctness tests: performance agaist the cluster evaluation metrics (Purity, Inverse Purity, Adjusted Rand Index, and Jaccard Index; higher is better) and run times (total run time, time spent clustering, calculating distance matrices, classifying new binaries, and other operations; lower is better).\label{fig:synth_all}}
\end{figure*}

The most obvious measurable difference between the two approaches is speed, shown in Fig.~\ref{fig:correct_timing}, where the time spent on each task in the algorithm is detailed. Clustering the binaries in a single batch took over five hours, whereas the incremental system completed in under an hour. In the single--batch run, the application spent 21 minutes calculating the distance matrix, after which the vast majority of the time was spent on clustering. In the incremental run, the application only spent a total of about 20 minutes clustering, which occurred when the first batch was introduced to the system and whenever subsequent batches initiated re--clustering. Neighbor Join is a polynomial time algorithm, so the many smaller clustering operation performed during the incremental run finish orders of magnitude faster than a single clustering of the entire dataset. Besides the time spent building the distance matrix prior to clustering, the incremental run also spent some time classifying new binaries as they were introduced. In both runs, some other time was spent persisting the database, manipulating data structures, initializing and updating exemplar features, and other such overhead.

%\subsection{Conclusion}
%\label{conc}
%We have described our scalable, incremental approach to the management of malware based on our hierarchical data structure. We laid out our online method of analyzing new malware while maintaining a meaningful representation of the malware families, and have proven its effectiveness, both in theory and through the results of actual experiments comparing it to traditional hierarchical clustering methods. We have also described the flexibility of the system, and have outlined the various implementation choices we have made to perform our analysis.\ 

As can be seen, this hierarchical clustering approach can assist in malware analysts' need to efficiently and effectively group together and classify malware in an online fashion. In this way we can support both the triage and deep--dive stages of the malware analysis work flow. Organizing malware into families also provides a basis for performing some of MAAGI's other types of analysis, including lineage generation and component identification, which can provide insight into the specific relationships between malware binaries within and across families.
 
\section{MAAGI Analysis Framework}

The MAAGI system incorporates several malware analysis tools that rely heavily on AI techniques to analysis capabilities to the user. Each of these techniques build upon the malware hierarchy described in Sec.~\ref{hier} that serves as the persistant organization of detected malware into families. We present the details of each of the analyses in the following sections.

\subsection{Component Identification}

Malware evolution is often guided by the sharing and adaptation of functional components that perform a desired purpose. Hence, the ability to identify shared \textit{components} of malware is central to the problem of determining malware commonalities and evolution. A component can be thought of as a region of binary code that logically implements a ``unit'' of malicious operation. For example, the code responsible for key--logging would be a component. Components are intent driven, directed at providing a specific malware capability. Sharing of functional components across a malware corpus would thus provide an insight into the functional relationships between the binaries in a corpus and suggest connection between their attackers.

Detecting the existence of such shared components is not a trivial task. The function of a component is the same in each malware binary, but the \textit{instantiation} of the component in each binary may vary. Different authors or different compilers and optimizations may cause variations in the binary code of each component, hindering detection of the shared patterns. In addition, the detection of these components is often unsupervised, that is the number of common components, their size, and their functions may not be known {\it a priori}. 

The component identification analysis tool helps analysts find shared functional components in a malware corpus. Our approach to this problem is based on an ensemble of techniques from program analysis, functional analysis, and artificial intelligence. Given a corpus of (unpacked) malware binaries, each binary is reverse engineered and decomposed into a set of smaller functional units, namely procedures (approximately corresponding to language functions). The code in each procedure is then analyzed and converted into a representation of the code's semantics. Our innovation is in the development of a unsupervised method to learn the components in a corpus. First, similar procedures from across the malware corpus are clustered.  The clusters created are then used as features for a second stage clustering where each cluster represents a component. %sOur two-stage clustering is different from classic multi--stage clustering in which each stage refines the clusters created in the previous stage. In contrast, in our two-stage clustering, the clusters in each stage consists of different elements, representing different groupings.

Under supervision of DARPA, an independent verification and validation (IV\&V) of our system was performed by Massachusetts Institute of Technology Lincoln Laboratory (MITLL). Their objective was to measure the effectiveness of the techniques under a variety of obfuscations used in real malware. The team constructed a collection of malware and experiments on this collection indicate that our method is very effective in detecting shared components in malware repositories.

\subsubsection{Generative Process}
\label{sec:clustering}

\begin{figure}[t]
\centering % was full columnwidth
\includegraphics[width=.75\columnwidth]{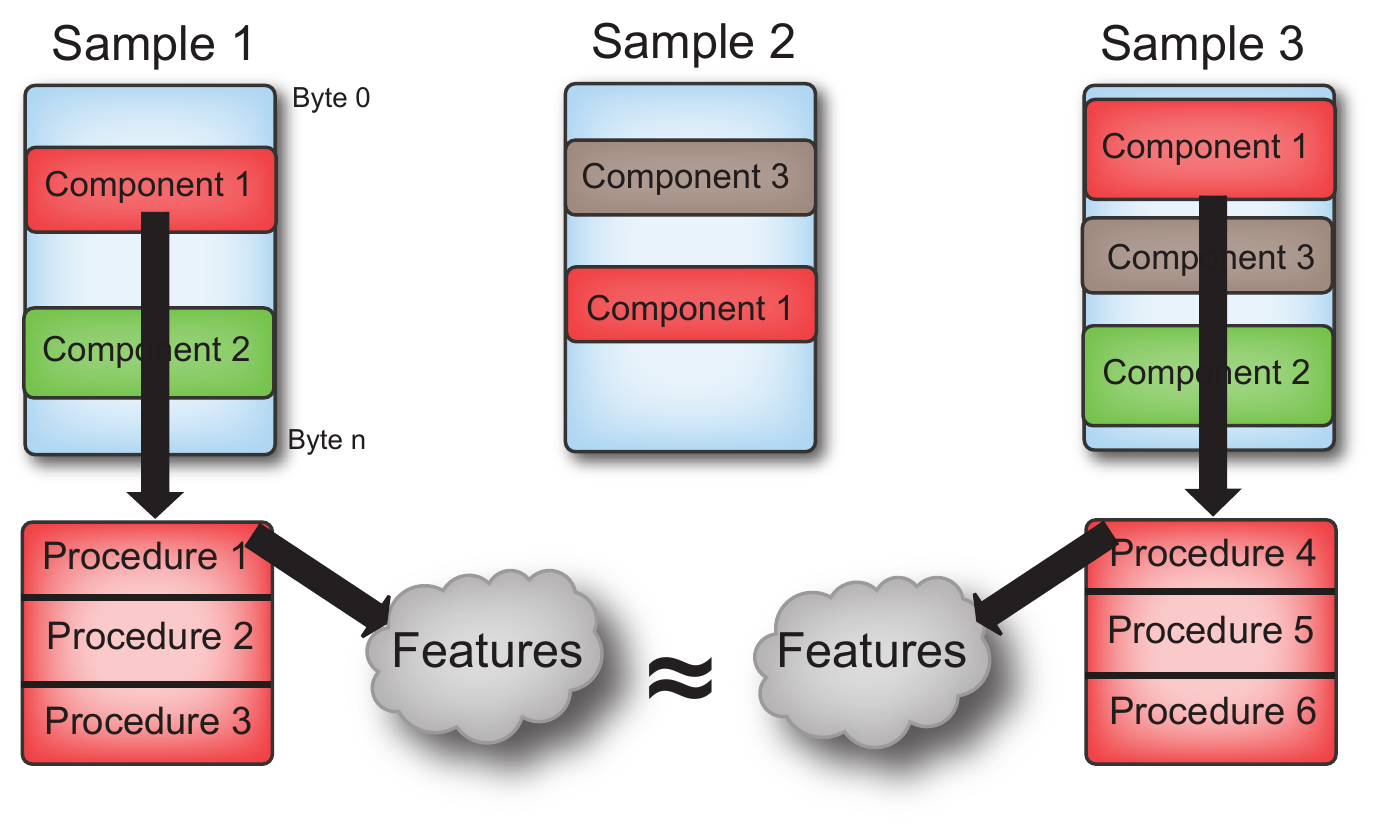}
\vspace{-5mm}
\caption{Component generative process in malware. Instantiations of components in different malware binaries should have similar features. \label{fig:GenProcess}}
\end{figure}

The basic idea of the unsupervised learning task can be thought of as reverse engineering a malware generative process, as shown in Fig.~\ref{fig:GenProcess}. A malware binary is composed of several shared components that perform malicious operations. Each component in turn is composed of one or more procedures. Likewise, we assume each procedure is represented by a set of features; in our system, features are extracted from the code blocks in the procedure, but abstractly, can be any meaningful features from a procedure. 

The main idea behind our method is that features from the procedures should be similar between \textit{instances} of the same component found in different binaries. Due to authorship, polymorphism, and compiler variation or optimizations, they may not be exact. However, we expect that two functionally similar procedures instantiated in different binaries should be more similar to each other than to a random procedure from the corpus. This generative process provides the foundation for our learning approach to discovery and identification of components.

\subsubsection{Basic Algorithm}
\begin{figure}[t]
\centering % was .9 columnwidth
\includegraphics[width=.65\columnwidth]{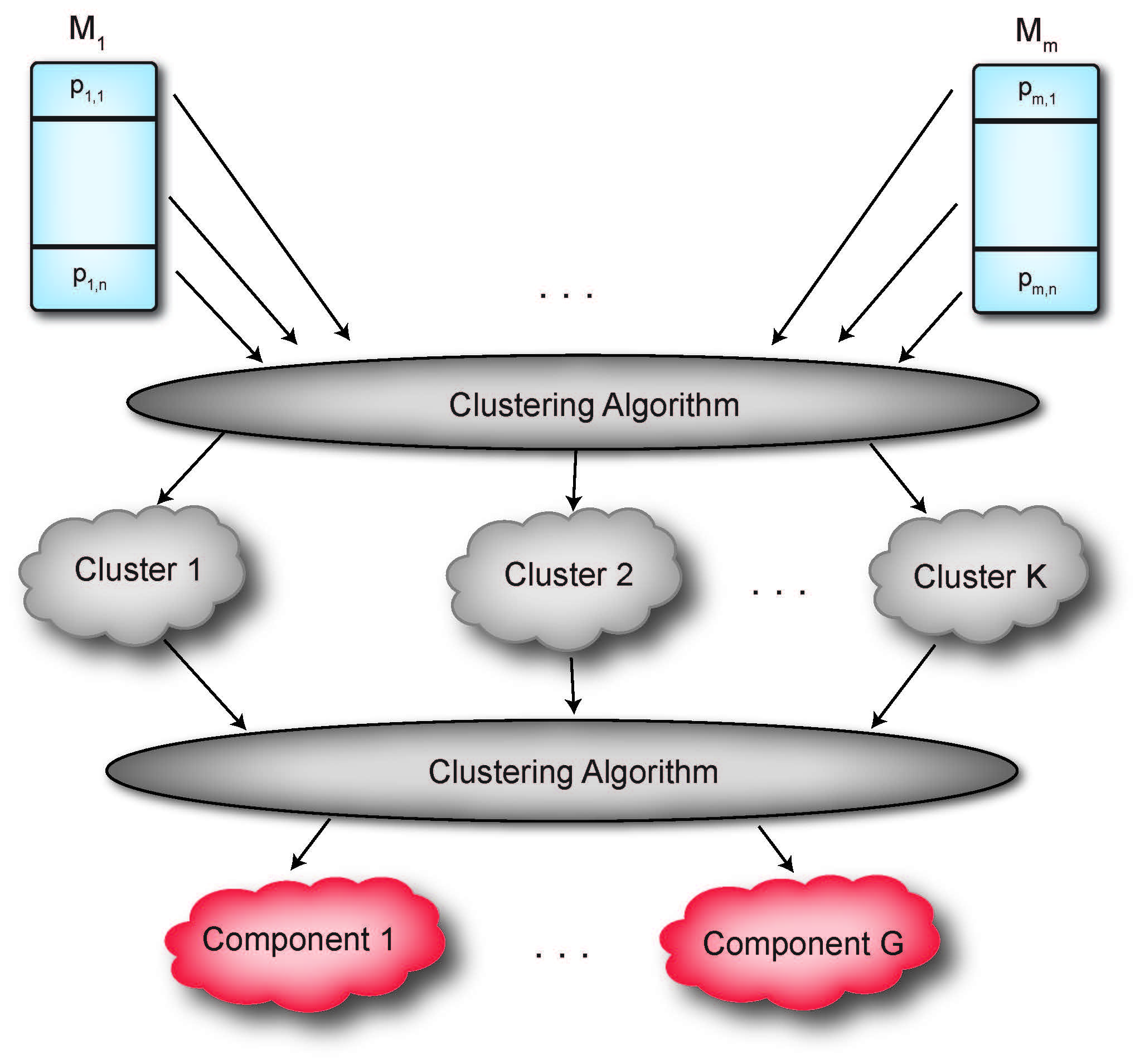}
\vspace{-5mm}
\caption{Two-stage clustering procedure to identify shared components in a malware corpus. Procedures are clustered based on feature similarity, then the centroids are converted to the space of binaries and clustered again. The resulting clusters are the shared functional components. \label{fig:TwoStage}}
\end{figure}

Building off of the generative process that underlies components, we develop a two--stage clustering method to identify shared components in a set of malware binaries, outlined in Fig.~\ref{fig:TwoStage}. 
For the moment, we assume that each procedure is composed of a set of meaningful features that describe the function of the procedure. Details on the features we employ in our implementation and evaluation can be found in Sections~\ref{sec:overview} and \ref{sec:testing}.

Given a corpus of malware binaries $\mathbb{M} = \{M_1, \dots, M_{|\mathbb{M}|}\}$, we assume it contains a set of shared functional components $\mathbb{T} = \{T_1, \dots, T_{|\mathbb{T}|}\}$. However, we are only able to observe $T_{i,j}$, which is the \textit{instantiation} of the $i^{th}$ component in $M_j$. If the component is not part of the binary $M_j$, then $T_{i,j}$ is undefined. We also denote $T_{i,*}$ as the set of all instantiations of the $i^{th}$ component in the entire corpus. Note that $T_{i,j}$ may not be an exact code replica of $T_{i,k}$, since the components could have some variation due to authorship and compilation. Each $M_j$ consists of a set of procedures $p_{i,j}$, denoting the $i^{th}$ procedure in the $j^{th}$ binary.

\paragraph{Procedure-Based Clustering}
The first stage of clustering is based on the notion that if $T_{i,j} \in M_{j}$ and $T_{i,k} \in M_{k}$, then at least one procedure in $M_{j}$ must have high feature similarity to a procedure in $M_k$.  Since components are shared across a corpus and represent common functions, even among different authors and compilers, it is likely that there is some similarity between the procedures. We first start out with a strong assumption and assert that the components in the corpus satisfy what we term as the \textit{component uniqueness} property.
\begin{definition}
\textbf{Component Uniqueness}. A component satisfies the component uniqueness property if the following relation holds true for all instantiations of $T_{i,*}$:
\begin{equation*}
\begin{split}
 \forall \: p_{x,j} \in T_{i,j}, & \: \exists \: p_{a,k} \in T_{i,k} \:  | \: d(p_{x,j}, p_{a,k}) \ll d(p_{x,j}, p_{*,*}), \\
 & \forall \: p_{*,*} \in T_{*,k}, \: T_{i,j}, T_{i,k} \in T_{i,*}
 \end{split}
\end{equation*}
\label{def:unique}
\end{definition}
\vspace{-2mm}
where $d(p_{*,*}, p_{*,*})$ is a distance function between the features of two procedures. Informally, this states that \textit{all} procedures in each instantiation of a component are much more similar to a single procedure from the same component in a different binary than to all other procedures.

Given this idea, the first step in our algorithm is to cluster the \textit{entire} set of procedures in a corpus. These clusters represent the common functional procedures found in \textit{all} the binaries, and by the component uniqueness property, similar procedures in instantiations of the same component will tend to cluster together. Of course, even with component uniqueness, we cannot guarantee that all like procedures in instantiations of a component will be clustered together; this is partially a function of the clustering algorithm employed. However, as we show in later sections, with appropriately discriminative distance functions and agglomerative clustering techniques, this clustering result is highly likely. %This all assumes that component uniqueness holds, which we will relax later.

These discovered clusters, however, are not necessarily the common components in the corpus. Components can be composed of multiple procedures, which may exhibit little similarity to each other (uniqueness does not say anything about the similarity between procedures in the same component). In Fig.~\ref{fig:GenProcess}, for example, Component 1 contains three procedures in binary 1 and binary 3. After clustering, three clusters are likely formed, each with one procedure from binary 1 and 3. This behavior is often found in malware: A component may be composed of a procedure to open a registry file and another one to compute a new registry key. Such overall functionality is part of the same component, yet the procedures could be vastly dissimilar based on the extracted features. However, based on component uniqueness, procedures that are part of the same shared component should appear in the same \textit{subset} of malware binaries. 

\paragraph{Binary-Based Clustering}

\begin{figure}[t]
\centering
\includegraphics[width=.65\columnwidth]{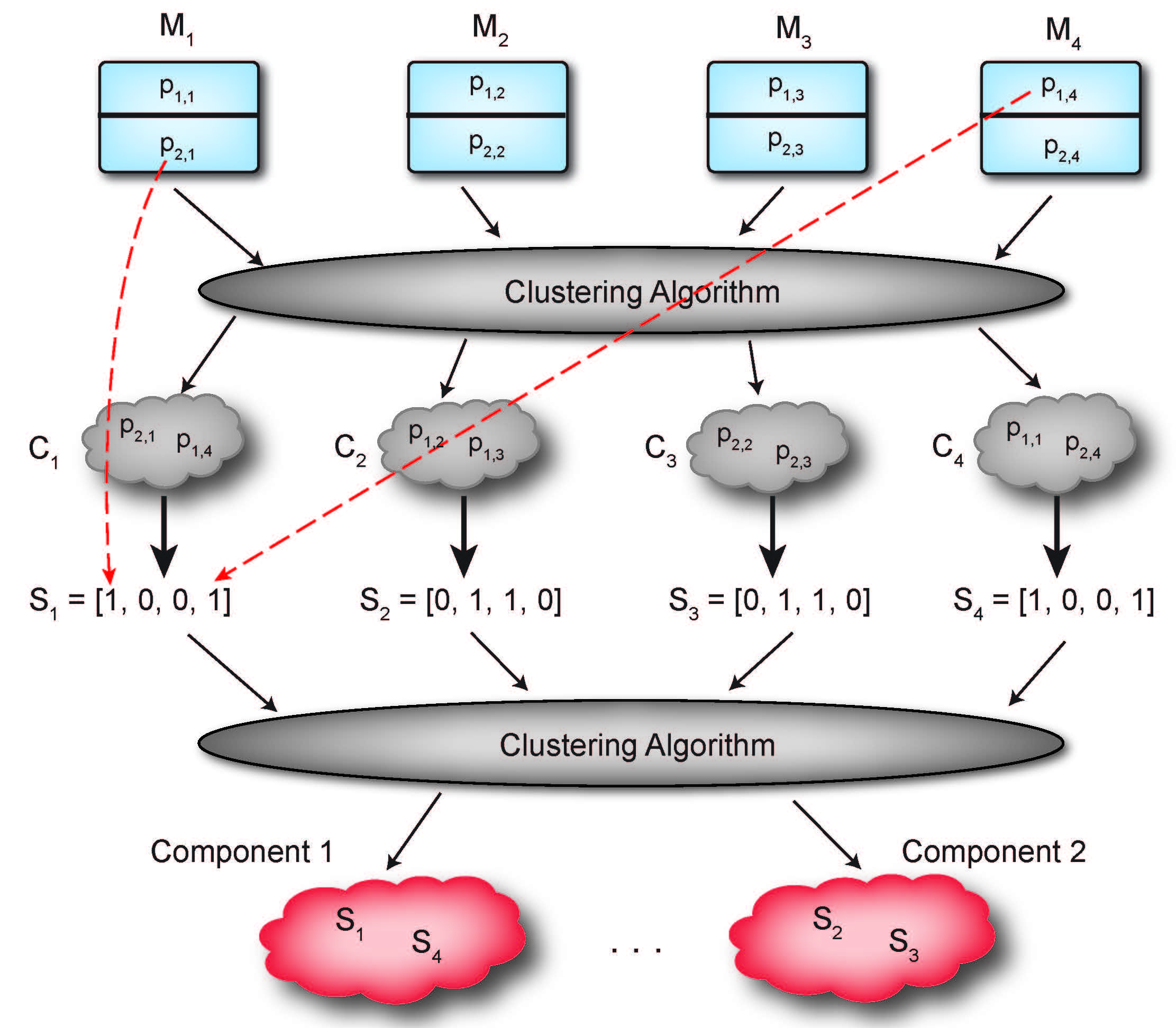}
\vspace{-5mm}
\caption{Conversion of procedure clusters into a vector space representing the presence of a procedure from each binary in the cluster, and then the subsequent clustering of the vector space representations. The clusters of vector space representations are the shared components. \label{fig:SpaceConversion}}
\end{figure}

Next, we perform a second step of clustering on the results from the first stage, but first convert the clusters from the space of procedure similarity to what we denote as binary similarity. Let $C_i$ represent a procedure cluster, where $p_{x_1,y_1}, \dots, p_{x_k,y_k} \in C_i$ are the procedures from the corpus that were placed into the cluster. We then represent $C_i$ by a vector $\vec{S_i}$, where
\begin{equation}
\label{eqn:ca1}
 S_i[j] = \begin{cases}
   1 & \text{if } \exists \,\, p_{x_k, y_k} \in C_i \,\, |\, y_k = j \\
   0       & \text{otherwise}
  \end{cases}
\end{equation}
That is, $S_i[j]$ represents the presence of a procedure from $M_j$ in the cluster. In this manner, each procedure cluster is converted into a point in an $|\mathbb{M}|$-dimensional space, where $|\mathbb{M}|$ is the number of malware binaries in the corpus. Consider the example shown in Fig.~\ref{fig:SpaceConversion}. The procedures in the corpus have been clustered into four unique procedure clusters. Cluster $C_1$ contains procedures $p_{2,1}$ from $M_1$, and $p_{1,4}$ from binary $M_4$. Hence, we convert this cluster into the point $\vec{S_1} = [1, 0, 0, 1]$, and convert the other clusters into the vector space representation as well. 

This conversion now allows us to group together clusters that appear in the same \textit{subset} of malware binaries. Using component uniqueness and a reasonable clustering algorithm in the first step, it is likely that a $p_{x,j} \in T_{i,j}$ has been placed in a cluster $C_v$ with other like procedures from $T_{i,*}$. Similarly, it is also likely that a different procedure in the same component instance, $p_{y,j}$, is found in cluster $C_w$ with other procedures from $T_{i,*}$. Since $C_v$ and $C_w$ both contain procedures from $T_{i,*}$, then they will contain procedures from the same set of binaries, and therefore their vector representations $\vec{S_v}$ and $\vec{S_w}$ will be very similar. We can state this intuition more formally as
\begin{equation*}
\begin{split}
 d(\vec{S_v}, \vec{S_w}) \approx 0 \Rightarrow p_{x,j},p_{y,k} \, \in \, T_{i,*} \, \forall \,\, p_{x,j},p_{y,k} \, \in\,  \{C_v, C_w\}
 \end{split}
\end{equation*}
Based on these intuitions, we then cluster the newly generated $\vec{S_i}$ together to yield our components. Looking again at Fig.~\ref{fig:SpaceConversion}, we can see that when cluster $C_1$ and $C_4$ are converted to $S_1$ and $S_4$, they contain the same set of binaries. These two procedure groups therefore constitute a single component instantiated in two binaries, and would be combined into a single cluster in the second clustering step, as shown.
\paragraph{Analysis}
Provided the component uniqueness property holds for all components in the data set, then the algorithm is very likely to discover all shared components in a malware corpus. However, if two components in the corpus are found in the exact same subset of malware binaries, then they become indistinguishable in the second stage of clustering; the algorithm would incorrectly merge them both into a single cluster. Therefore, if each component is found in the corpus according to some prescribed distribution, we can compute the probability that two components are found in the same subset of malware. 

Let $\mathcal{T}_i$ be a random variable denoting the set of malware binaries that contain the $i^{th}$ component. If $\mathcal{T}_i$ is distributed according to some distribution function, then for some $t = \{M_x, \dots, M_y\} \subseteq \mathbb{M}$, we denote the probability of the component being found in exactly the set $t$ as $Pr(\mathcal{T}_i = t)$. Assuming uniqueness holds, we can now determine the probability that a component is detected in the corpus.
\begin{theorem}
\label{thm:thm1}
The probability that the $i^{th}$ component is detected in a set of malware binaries is
\begin{equation*}
\underset{t \in \text{all subsets of}\, \mathbb{M}}{\sum} \!\!\!\!\!\!\!\!\!\! Pr(\mathcal{T}_k \neq t, \dots, \mathcal{T}_i = t, \dots, \mathcal{T}_k \neq t) 
\end{equation*}
\end{theorem}
\begin{proof}
 If $\mathcal{T}_i = t_j$ for some $t_j \subseteq \mathbb{M}$, the component will be detected if no other component in the corpus is found in the exact same subset. That is, $\mathcal{T}_k \neq t_j$ for all other components in the corpus. Assuming nothing about component or binary independence, the probability of no other component sharing $t_j$ is the joint distribution $Pr(\mathcal{T}_k \neq t_j, \dots, \mathcal{T}_i = t_j, \dots, \mathcal{T}_k \neq t_j)$. Summing over all possible subsets of $\mathbb{M}$ then yields Thm~\ref{thm:thm1}.
\end{proof}

Thm.~\ref{thm:thm1} assumes nothing about component and binary independence. However, if we do assume that components are independent of each other and a component $T_i$ appears independently in each binary with probability $p_i$, then $\mathcal{T}_i$ is distributed according to a binomial distribution. As such, we can compute a lower bound for the probability of detection by ignoring equality between distribution sets as
\begin{equation*}
\begin{split}
 Pr(\text{Dete}&\text{ction of $T_i$})\\
  = & \underset{t \in \text{all subsets of}\, \mathbb{M}}{\sum} \!\!\!\!\!\!\!\!\!\! Pr(\mathcal{T}_k \neq t, \dots, \mathcal{T}_i = t, \dots, \mathcal{T}_k \neq t)  \\
               = & \underset{t \in \text{all subsets of}\, \mathbb{M}}{\sum} \!\!\!\!\!\!\!\!\!\! Pr(\mathcal{T}_i = t) \prod_{k \neq i} (1-Pr(\mathcal{T}_k = t))  \\
               \geq & \sum_{x=0}^{|\mathbb{M}|} Pr(|\mathcal{T}_i| = x) \prod_{k \neq i} (1-Pr(|\mathcal{T}_k| = x))\\
               = & \sum_{x=0}^{|\mathbb{M}|} Bin(x,|\mathbb{M}|, p_i)  \prod_{k \neq i} (1-Bin(x,|\mathbb{M}|, p_k))\\
\end{split}
\end{equation*}
where $Bin(\cdot)$ is the binomial probability distribution function. This lower bound can provide us with reasonable estimates on the probability of detection. For instance, even in a small data set of 20 binaries with two components that both have a 20\% chance of appearing in any binary, the probability of detection is at least 0.85.

Based on component uniqueness, the basic algorithm can easily locate the shared components in a set of malware binaries. However, in practice, component uniqueness rarely holds in a malware corpus. That is, it is likely that some procedures in different components are quite similar. This situation can be quite problematic for the basic algorithm. In the next section, we relax the component uniqueness assumption and detail a more sophisticated algorithm intended to identify components.

\subsubsection{Assumption Relaxation}
\begin{figure}[t]
\centering
\includegraphics[width=.65\columnwidth]{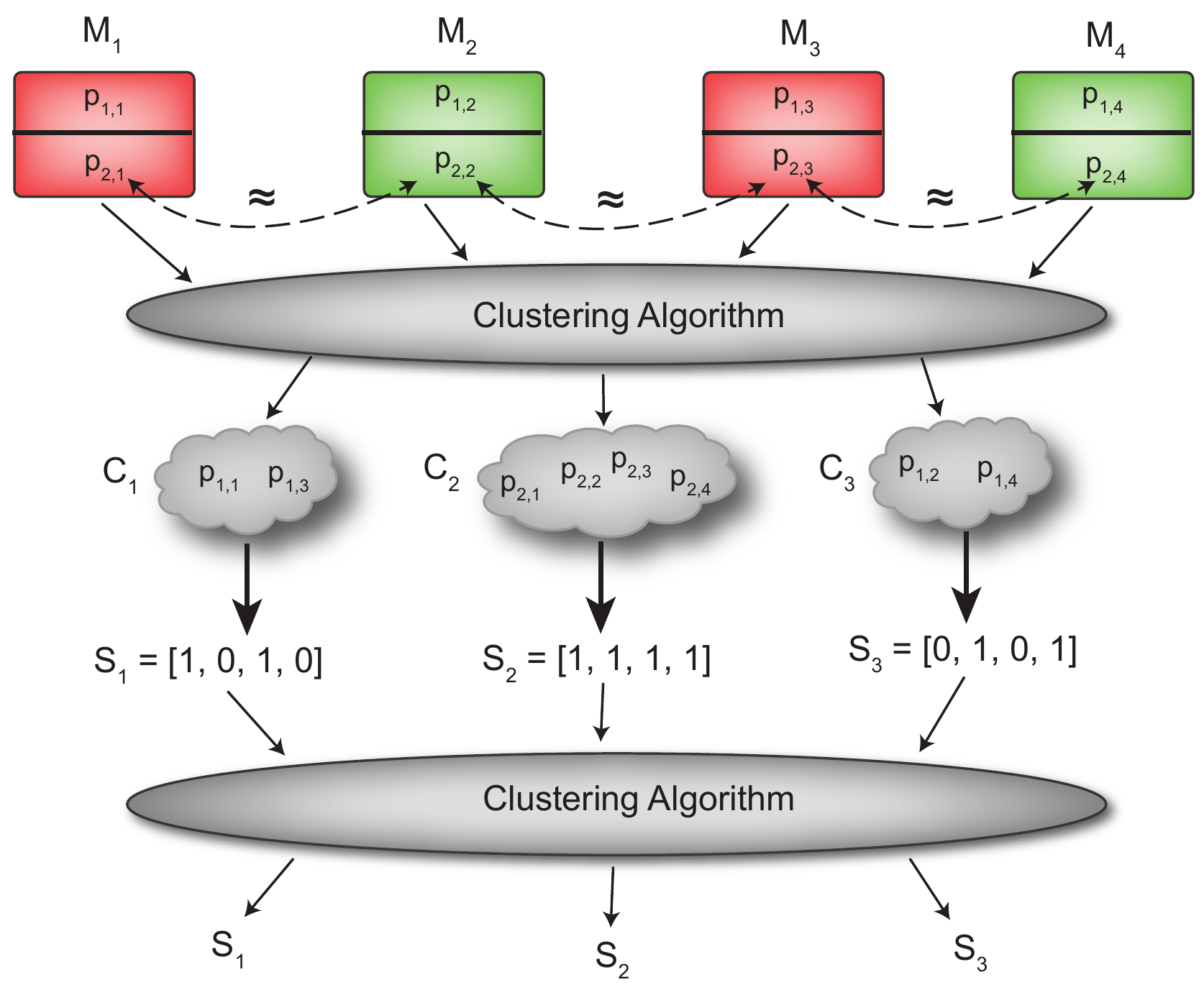}
\vspace{-7mm}
\caption{Basic algorithm without the component uniqueness property. There are two components in the set, yet the end result clustering produces three, of which one is a combination of two components\label{fig:NonUnique}}
\end{figure}
When component uniqueness does not hold, the basic algorithm may not correctly identify components. Consider the example shown in Fig.~\ref{fig:NonUnique}. There are two components in four binaries, each composed of two procedures. Assuming component uniqueness does not hold, then the second procedure in each binary could show high similarity to each other (it is possible they perform some basic function to set up malicious behavior). After the first step, $p_{2,1}$, $p_{2,2}$, $p_{2,3}$, and $p_{2,4}$ are placed in the same cluster; this results in creation of $\vec{S_2}$ that does not resemble any other cluster vectors. Hence, any clustering of $\vec{S_2}$ with $\vec{S_1}$ or $\vec{S_3}$ will result in a misidentification of the procedures in each component.

To remediate this error, we utilize an algorithm that ``splits'' clusters discovered in the first step of the algorithm before the second stage of clustering is performed. This requires that we relax the component uniqueness assumption in Def.~\ref{def:unique} to what we term as \textit{procedure uniqueness}, which states that only \textit{one} procedure in an instantiation of a component must exhibit high similarity to a procedure in another instantiation. The intuition is that after the first stage of clustering, there are clusters $C'_{1} \dots C'_{|\mathbb{T}|}$ that each contain the set of procedures $p'_{i,*}$. That is, from procedure uniqueness and a reasonable clustering algorithm, it is highly likely that we get $|\mathbb{T}|$ clusters, one for each component in the data set. We do not discuss the detailed algorithmic process and theoretic limits of splitting in this work, but refer the reader to Ruttenberg et. al.~\cite{ruttenberg2014} for more details.

\subsubsection{Implementation}
\label{sec:overview}
Our component identification system is intended to discover common functional sections of binary code shared across a corpus of malware. The number, size, function and distribution of these components is generally unknown, hence our system architecture reflects a combination of unsupervised learning methods coupled with semantic generalization of binary code. 
The system uses two main components:
\begin{enumerate}
\item{BinJuice}: To extract the procedures and generate a suitable $Juice$ features.
\item{Clustering Engine}: To perform the actual clustering based on the features.
\end{enumerate}

We use Lakhotia et al.'s BinJuice system~\cite{lakhotia2013fast} to translate the code of each basic block into four types of features: {\tt code}, {\tt semantics}, {\tt gen\_semantics}, and {\tt gen\_code}. Since each of the features are strings, they may be represented using a fixed size hash, such as md5. As a procedure is composed of a set of basic blocks, a procedure is represented as an unordered set of hashes on the blocks of the procedure.  We measure similarity between a pair of procedures using the Jaccard index~\cite{theodoridis2008pattern} of their sets of features.

\subsubsection{Clustering Engine}
For the first stage of clustering, we choose to use a data driven clustering method. Even if we know the number of shared components in a corpus, it is far less likely that we will know how many procedures are in each component. Thus, it makes sense to use a method that does not rely on prior knowledge of the number of procedure clusters. 

We use Louvain clustering for the procedure clustering step~\cite{blondel2008fast}. Louvain clustering is a greedy agglomerative clustering method, originally formulated as a graph clustering algorithm that uses modularity optimization~\cite{newman2006modularity}. We view procedures as nodes in a graph and the weights of the edges between the nodes as the Jaccard index between the procedure features. Modularity optimization attempts to maximize the modularity of a graph, which is defined as groups of procedures that have higher intra--group similarity and lower inter--group similarity than would be expected at random. Louvain clustering iteratively combines nodes together that increases the overall modularity of the graph until no more increase in modularity can be attained. 

For the second stage, we experimented with two different clustering methods: Louvain and K--means. These methods represent two modalities of clustering, and various scenarios may need different methods of extracting components. For instance, in situations where we know a reasonable upper bound on the number of components in a corpus, we wanted to determine if traditional iterative clustering methods (i.e., K--means) could outperform a data driven approach. In the second step of clustering, the $L_2$ distance between vectors was used for K--means, and since Louvain clustering operates on similarity (as opposed to distance), an inverted and scaled version of the $L_2$ distance was employed for the second stage Louvain clustering.

\subsubsection{Experiments}
\label{sec:testing}

Since it is quite a challenge to perform scientifically valid controlled experiments that would estimate the performance of a malware analaysis system in the real-world, MITLL was enlisted (by DARPA) to create a collection of malware binaries with byte--labeled components. That is, for each component they know the exact virtual memory addresses of each byte that is part of the component in every binary. We tested the performance of the component identification algorithm using this data set. 

In all tests, the algorithms do not have prior knowledge of the number of components in the data set. For the K-means tests, we set a reasonable upper bound on the estimated number of components ($K=50$). We used two metrics to gauge the effectiveness of our method. The first is using the Jaccard index to measure the similarity between the bytes identified as components in our algorithm and the ground truth byte locations. We also used the Adjusted Rand Index~\cite{hubert1985comparing} to measure how effective our algorithm is at finding binaries with the same components. 

\paragraph{Results}
We ran all of our tests using the two clustering algorithms (Louvain and K-means), and additionally tested each method with and without splits to determine how much the relaxation of component uniqueness helps the results. Note that no parameters (besides $K$) were needed for evaluation; we utilize a completely data driven and unsupervised approach.
\label{sec:darpa}

\begin{figure}
\centering 
\subfigure[gen\_code]{
\includegraphics[width=.45\textwidth]{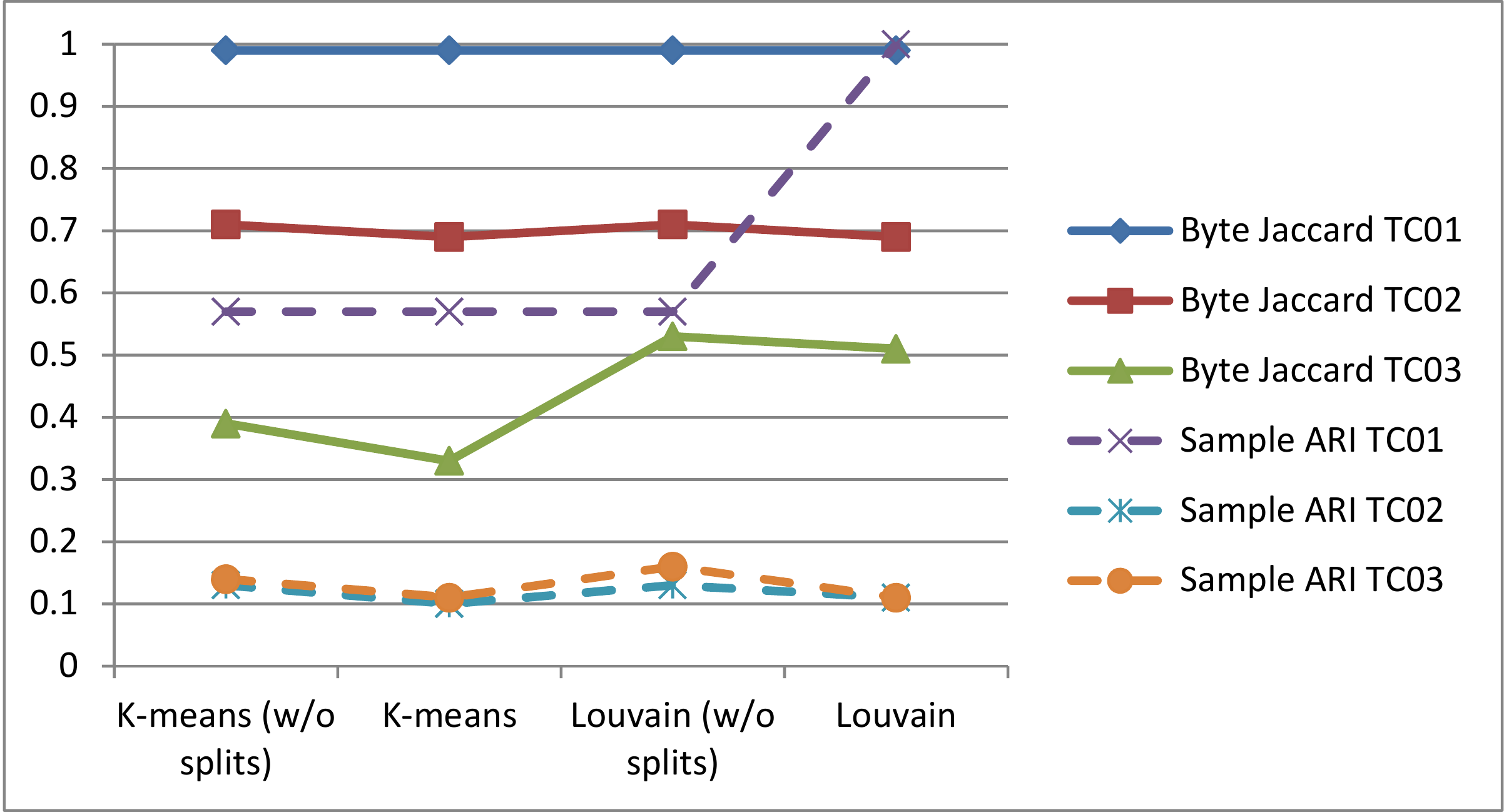}
\label{fig:RRgen_code}
}
\subfigure[gen\_semantics]{
\includegraphics[width=.45\textwidth]{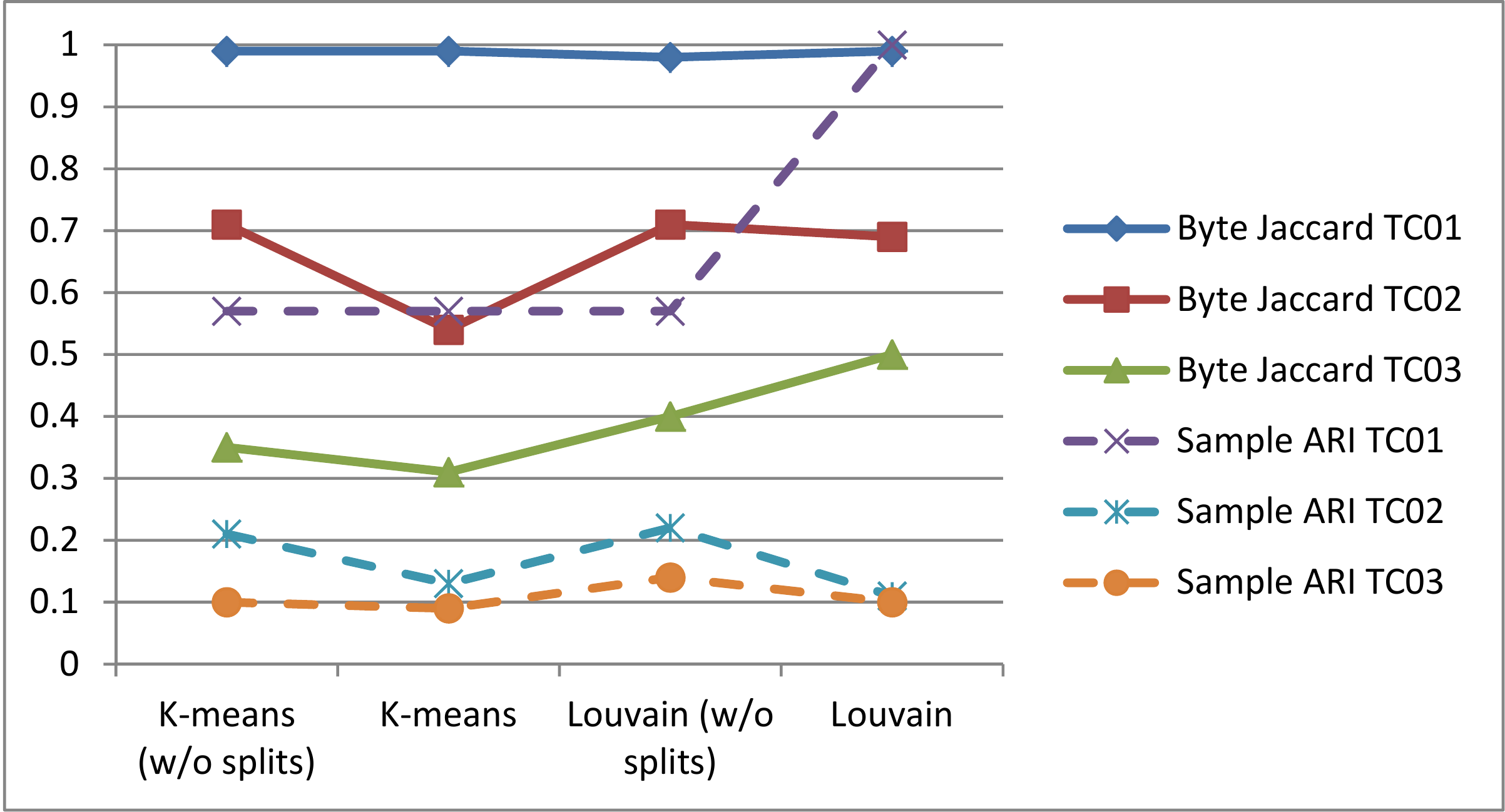}
\label{fig:RRgen_semantics}
}
\caption[]{Byte Jaccard and binary ARI comparisons of the different methods on the IV\&V data-set using two of the BinJuice features. \label{fig:RR}}
\end{figure}

The results of the component identification are shown in Fig.~\ref{fig:RR}, where each metric is shown for three different test sets derived from the IV\&V data (TC1, TC2 and TC3). The {\tt code} and {\tt semantics} features, as expected, produced inferior results as compared to {\tt gen\_code} and {\tt gen\_semantics} features during initial testing. Hence subsequent testing on those feature was discontinued.

In general, all four methods have fairly low ARIs regardless of the BinJuice feature. This indicates that in our component identifications the false positives are distributed across the malware collection, as opposed to concentrated in a few binaries. Furthermore, as indicated by the Jaccard index results, the misclassification rate at the byte level is not too high. The data also shows that the Louvain method outperforms K-means on all BinJuice features, though in some cases the splitting does not help significantly. Note that the difference between Louvain with and without splitting is mainly in the binary ARI. Since Louvain without splitting is not able to break clusters up, it mistakenly identifies non--component code in the malware as part of a real component; hence, it believes that binaries contain many more components than they actually do. These results demonstrate the robustness of the Louvain method and the strength of the BinJuice generated features. The data also shows that relaxing the component uniqueness property can improve the results in real malware.

\paragraph{Study with Wild Malware}
\label{sec:wild}
\begin{figure}[t]
\centering
\subfigure[Components per binary]{
\includegraphics[width=.4\textwidth]{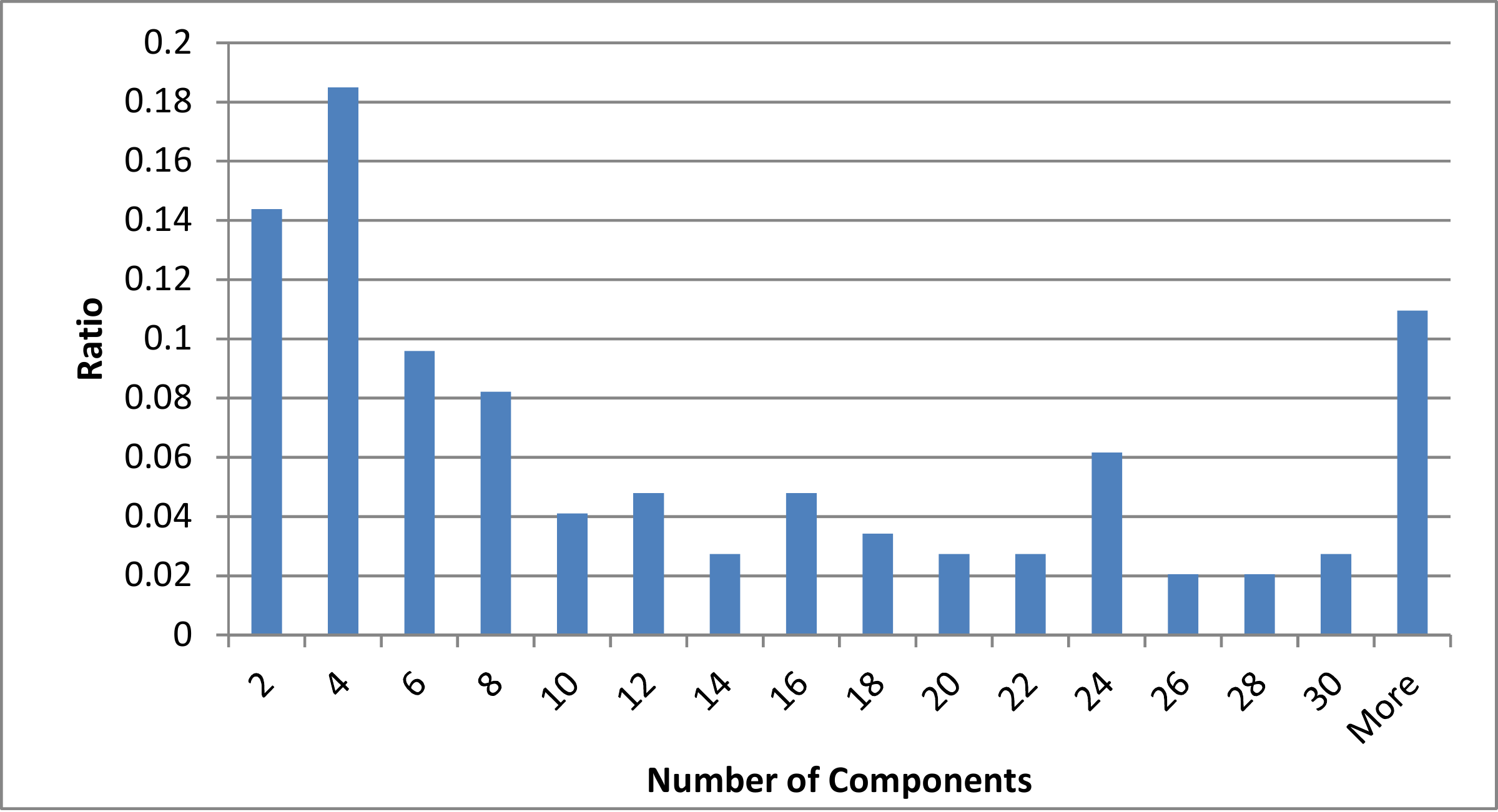}
\label{fig:darpahist}
}
\subfigure[binaries per component]{
\includegraphics[width=.4\textwidth]{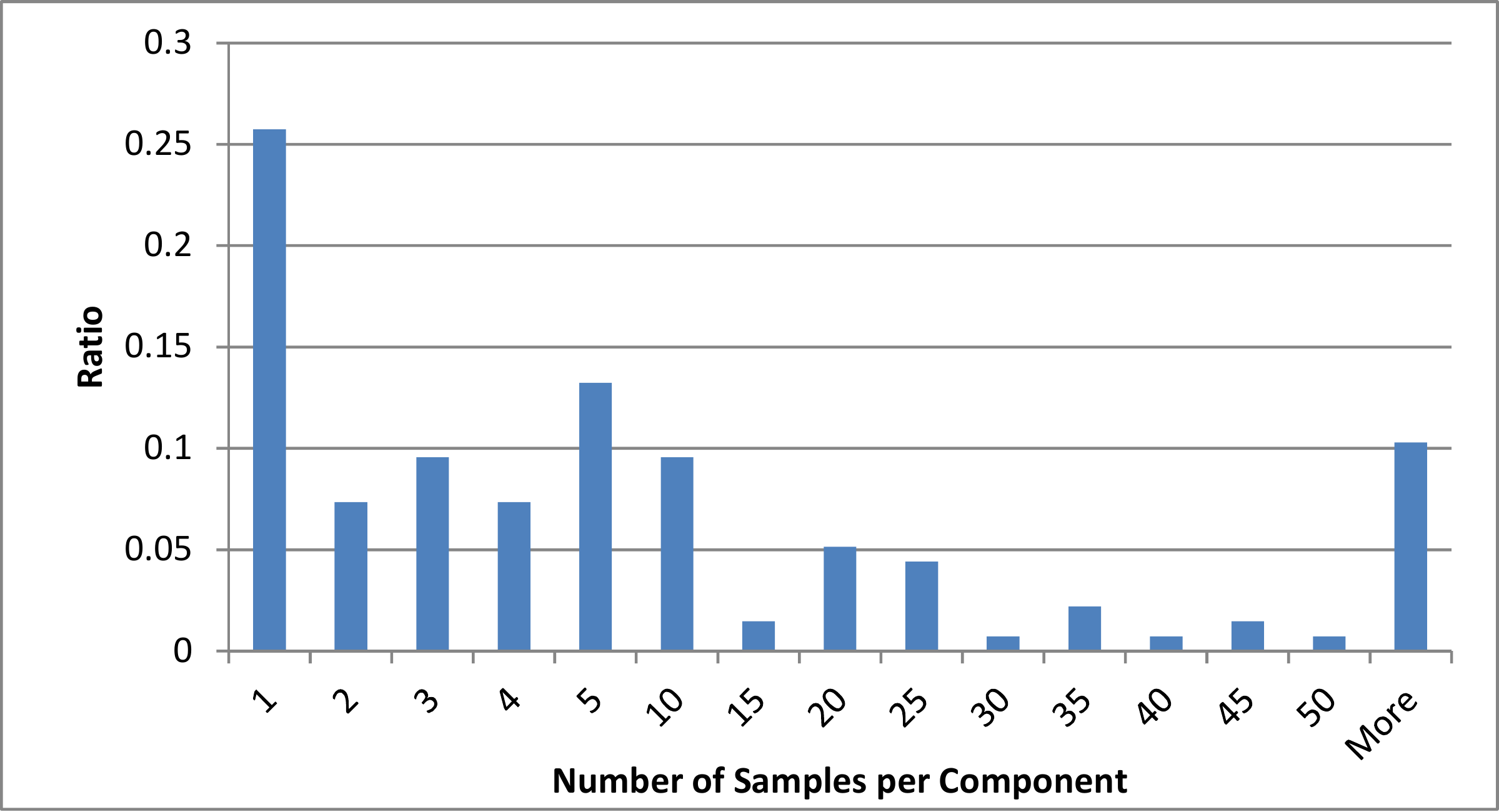}
\label{fig:comphist}
}
\subfigure[Component size compared to binaries per component]{
\includegraphics[width=.4\textwidth]{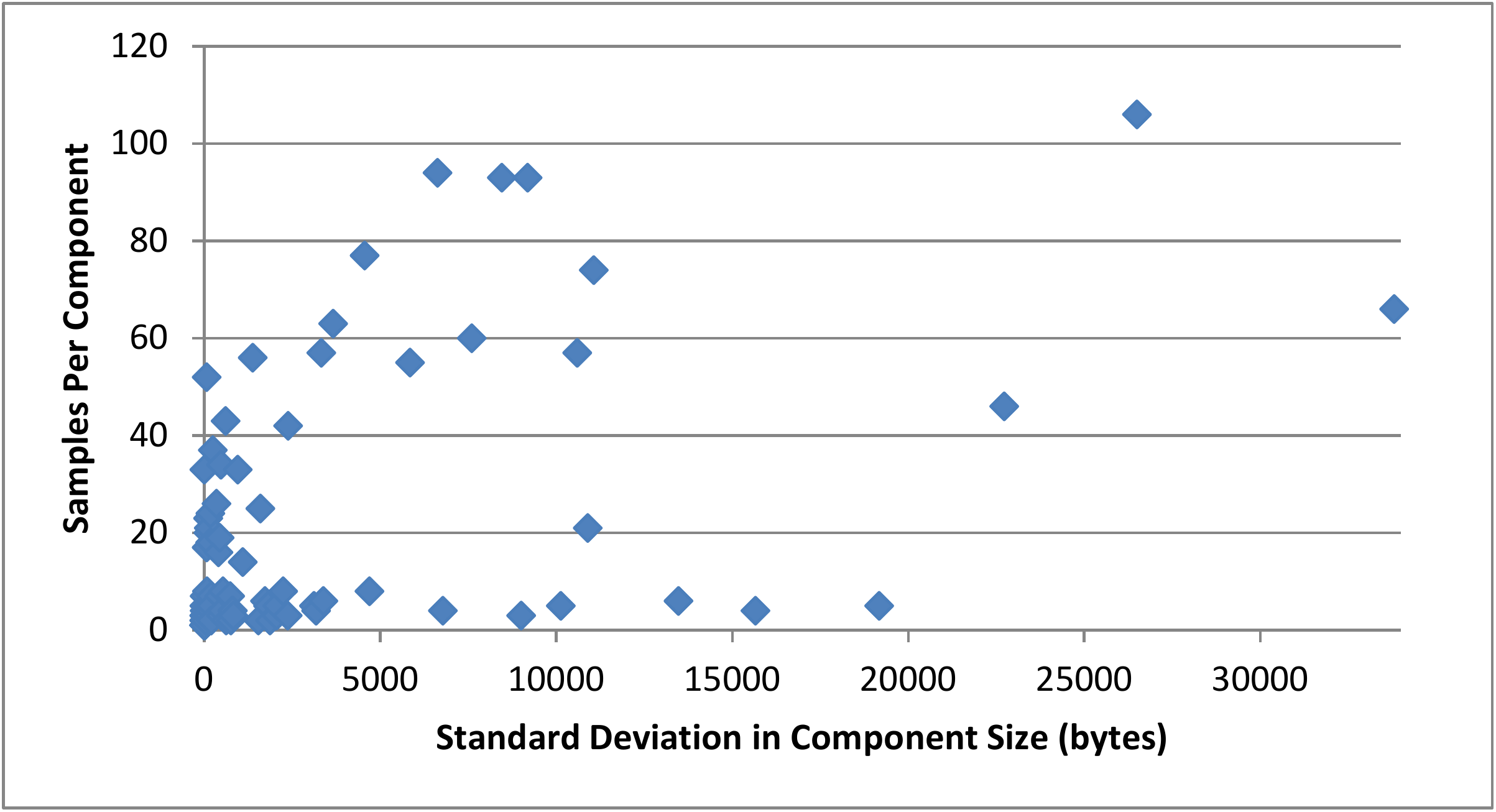}
\label{fig:darpavariance}
}
\caption[]{Histograms of the number of components found in each binary and the number of binaries per identified component for the wild malware. In addition, we also show the variation in component size as a function of the number of binaries containing each component.\label{fig:darpa}}
\end{figure}

We also performed component identification on a small data set of wild malware consisting of 136 wild malware binaries. We identified a total of 135 unique components in the data set. On an average 13 components were identified per malware binary. Fig.~\ref{fig:darpahist} shows the histogram of the number of components discovered per binary. As evident from the graph, the distribution is not uniform. Most malware binaries have few components, though some can have a very large number of shared components. In addition, we also show the number of binaries per identified component in Fig.~\ref{fig:comphist}. As can be seen, most components are only found in a few binaries. For example, 25\% of components are only found in a single binary, and thus would most likely not be of interest to a malware analyst (as components must be shared among malware binaries in our definition). 

In general, many of the identified components are similar in size (bytes), as shown in Fig.~\ref{fig:darpavariance}. In the figure, we plot the variance of the size of the instantiations in each of the 135 components against the number of binaries that contain the component. As can be seen, many of the binaries have low variance in their component size, indicating that it is likely that many of the components are representing the same shared function (components with large variation in observed instantiation size are likely false positive components). In addition, many of these low variance components are non-singleton components, meaning that the component has been observed in many malware binaries. While further investigation is needed to determine the exact function and purpose of these components, these results do indicate that our method is capable of extracting shared components in a corpus of wild malware.

The component identification analysis in the MAAGI system is clearly capability of identifying shared components across a corpus with high accuracy. As malware becomes more prevalent and sophisticated, determining the commonalities between disparate pieces of malware will be key in thwarting attacks or tracking their perpetrators. Key to the continued success of finding shared components in light of more complex malware is the adoption of even more sophisticated AI methods and concepts to assist analyst efforts.

\subsection{Lineage}

Many malware authors borrow source code from other authors when creating new malware, or will take an existing piece of malware and modify it for their needs. As a result, malware within a family of malware (i.e., malware that is closely related in function and structure) often exhibit strong parent--child relationships (or parents and children). Determining the nature of these relationships within a family of malware can be a powerful tool towards understanding how malware evolves over time, which parts of malware are transferred from parent to child, and how fast this evolution occurs. 

Analyzing the lineage of malware is a common task for malware analysis, and one that has always relied on artificial intelligence. For instance, Karim et al. used unsupervised clustering methods to construct a lineage of a well--known worm~\cite{karim2005malware}. Lindorfer et al. constructed lineages by relying on rule learning systems to construct high--level behaviors of malware. Most recently, Jang et al.~\cite{jang2013towards} also used unsupervised clustering to infer the order and a subsequent straight line lineage of a set of malware. While all of these methods have been shown to be effective, the MAAGI system relies heavily on AI techniques to mostly automate the entire process of inferring lineages with diverse structure.

In the MAAGI system, a lineage graph of a set of malware binaries is a directed graph where the nodes of the graph are the set of binaries input to the component, and the directed edge from binary A to B implies that binary B evolved partly from sample A (and by implication, was created at a later time). Lineage's can contain multiple rooted binaries, in addition to binaries that contain multiple parents. The lineage component operates over a family of malware, defined by the persistent clustering data structure. That is, all malware binaries in the same cluster constitute a family. While it is envisioned that a malware analyst would generally be interested in lineage of a particular family, there is nothing that precludes one from generating a lineage from binaries from multiple families. 

\subsubsection{Lineage as a Probabilistic Model}

To determine the lineage of malware, it is essential to know the order in which binaries were generated. Without this information, it would be hard to determine the direction of parent--child relationships. As such, the lineage of a set of binaries is conditioned upon the the creation times of each of the binaries. This knowledge allows us to create a simple, high--level probabilistic model that represents (a small part of) the generative process of malware. Fig.~\ref{lineageprob} shows the basic model.

\begin{figure}
\centering 
\includegraphics[width=.35\columnwidth]{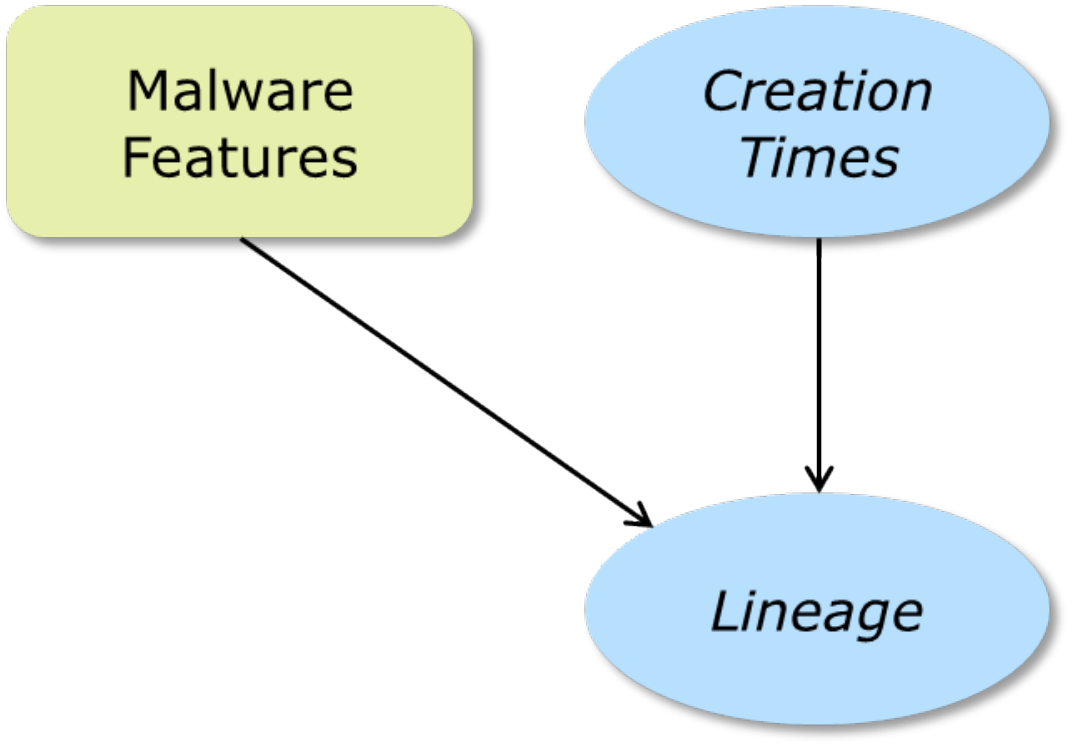}
\vspace{-5mm}
\caption{Lineage as a random variable conditioned upon the creation times of the malware and their features.\label{lineageprob}}
\end{figure}

The \textit{Lineage} variable represents a distribution over the possible lineages that can be constructed from a set of binaries, conditioned upon the \textit{Creation Times} variable and the \textit{Malware Features} variable. The \textit{Creation Times} variable represents a distribution over the dates that each binary was created. Finally, the \textit{Malware Features} is a deterministic variable that contains extracted features of the malware that is used as a soft constraint on the distribution of lineages. The more features that malware binaries share, the more likely that an edge in the lineage exists between them; the actual parent--child assignment of the two nodes depends upon the given creation times of the binaries. The lineage of a set of malware $M$ is then defined as

\begin{equation}
\label{lineageMax}
Lineage_M = \underset{Lineage_{M,i}}{argmax}\,\,P(Lineage_{M,i} | Features_M, Times_M)
\end{equation}

The compiler timestamp in the header is an indication of the time at which malware is generated. Unfortunately, it may be missing or purposely obfuscated by the malware author. However, it should not be ignored completely, as it provides an accurate signal in cases where it is not obfuscated. Therefore, we must infer the creation times of each malware binary from any available timestamp information, in addition to any other external information we have on the provenance of a malware binary. Fortunately, detailed malware databases exist that keep track of new occurrences of malware, and as a result we can also use the date that the malware was first encountered in the wild as additional evidence for the actual creation times of the binaries. 

One of the key insights in this formulation of the lineage problem is that the lineage of a set of malware and their creation times are joint processes that can inform each other; knowing the lineage can improve inference of the creation times, and knowing the creation times can help determine the lineage. As such, inferring the creation times and the lineage jointly can potentially produce better results than inferring the creation times first and conditioning the lineage inference on some property of the inferred creation times (e.g., such as the most likely creation times).

We break down the overall probabilistic model into two separate models, one to represent the creation times and another to represent the lineage construction process. The lineage inference algorithm explained in subsequent sections details how joint inference is performed using the two models. 

\paragraph{Creation Time Model}

\begin{figure}
\centering 
\includegraphics[width=.75\columnwidth]{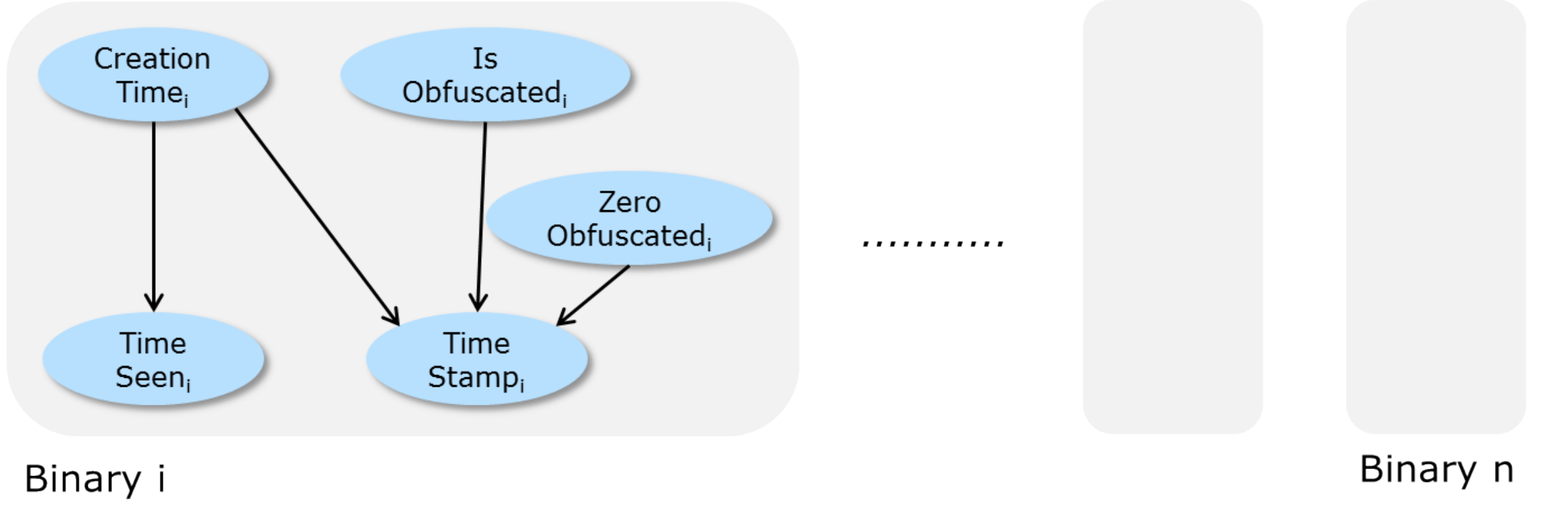}
\vspace{-5mm}
\caption{Learning a distribution over the creation times of each binary. We create a small probabilistic model for each malware that infers a lineage--independent distribution of its creation time using time stamp and time seen as evidence.\label{CreationTimeModel}}
\end{figure}

The creation time model for a set of malware binaries is shown in Fig.~\ref{CreationTimeModel}. For a set of $N$ binaries, we instantiate $N$ independent probabilistic models, one for each binary. While the creation times of each malware binary are \textit{not} truly independent, the dependence between the creation times of different binaries is enforced through the joint inference algorithm.

Each probabilistic model contains five variables. First, there is a variable to represent the actual creation time of the malware. There is also a variable to represent the time the sample was first seen in the wild, which depends upon the actual creation time of the sample. There is also a variable to represent the time stamp of the sample (from the actual binary header). This variable depends upon the creation time, as well as two additional variables that represent any obfuscation by the malware author to hide the actual creation time; one variable determines if the time stamp is obfuscated, and the other represents how the time stamp is obfuscated (either empty or some random value).

Evidence is posted to the time seen and time stamp variables and a distribution of each malware's creation time can be inferred. Note that the priors for the obfuscation variables and parameters for the conditional distributions can be learned and are discussed later.

\paragraph{Lineage Model}

\begin{figure}[t]
\centering 
\includegraphics[width=.75\columnwidth]{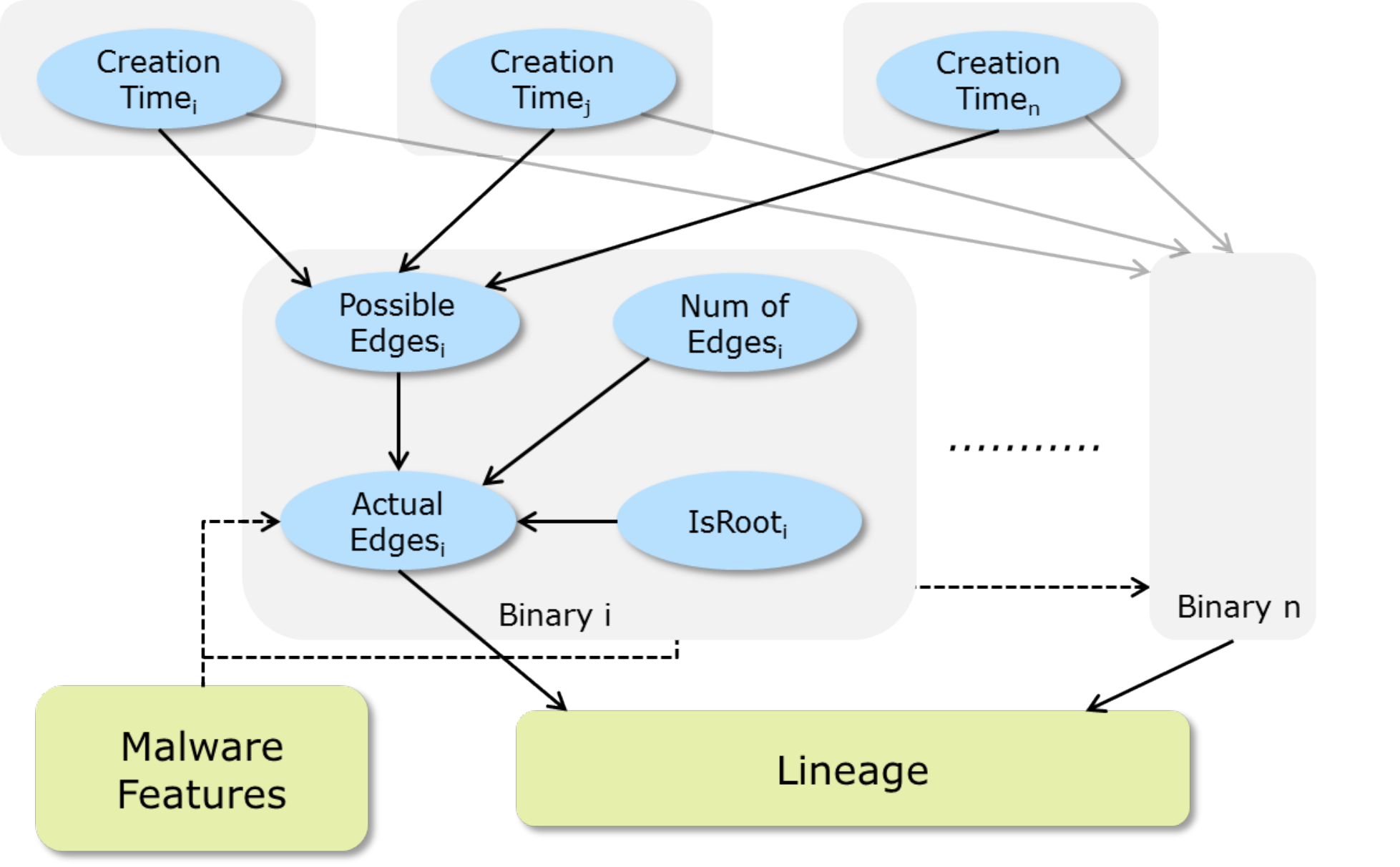}
\vspace{-1mm}
\caption{Lineage model given distributions of the creation times of each malware. The malware features are used as soft constraints on the inheritence relationships, and given the edges for each binary, constructing the lineage is deterministic\label{LineageModel}}
\end{figure}

The model for lineage is shown in Fig.~\ref{LineageModel}. For each malware binary, we create a set of variables that represent how the binary can be used in creating a lineage. First, there is a variable called $Possible Edges_i$, which represents the possible existence of edges between a binary $i$ and all the other binaries. This variable is deterministic conditioned upon the creation times of all the binaries; since a malware binary can only inherit from prior binaries, the value of this variable is simply all of the binaries with earlier creation times.  

There are also two variables that control the number of edges (i.e., parents) for each binary, as well as a variable that specifies whether a particular binary is a root in the lineage (and thus has no parents). Finally, there is a variable that represents a set of \text{actual} lineage edges of a binary $i$ which depends upon the possible edges of the binary, the number of edges it has, and whether it is a root binary. By definition, the values of the actual edges variable for all binaries defines the lineage over the set of malware (i.e., it can be deterministically construced from the edges and the creation times.)

In addition, the conditional probability distribution of the actual edges variable is constrained by the difference between the features of the binaries. That is, the higher similarity between two binaries, the more likely they are to have an edge between them. The similarity measure between binaries is based on the similarity measures used in other parts of the MAAGI system (clustering, component identification), and thus we do not discuss the measure in any further detail.

\subsubsection{Inference Algorithm}

\begin{figure}[t]
\centering 
\includegraphics[width=.75\columnwidth]{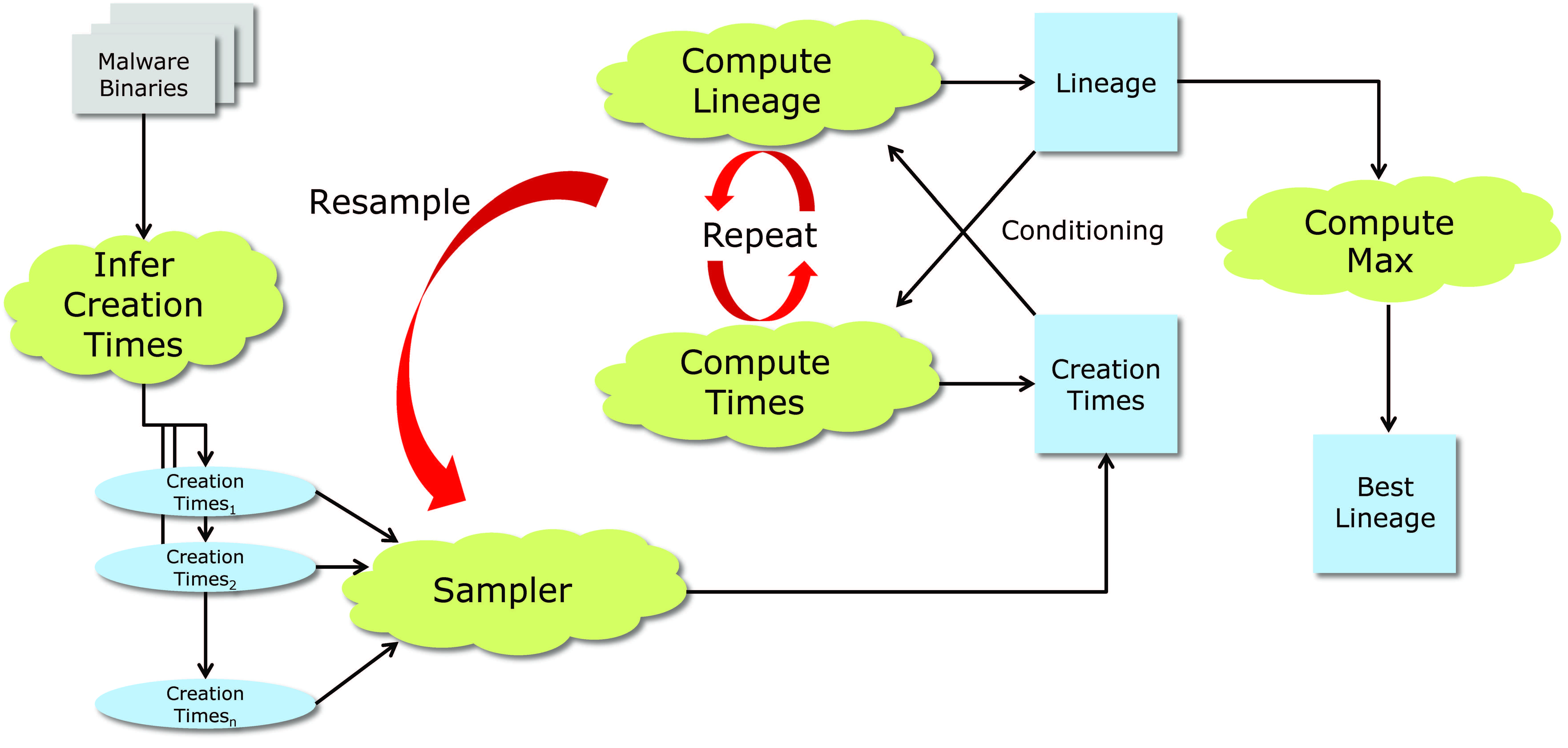}
\vspace{-1mm}
\caption{Inference algorithm to jointly infer the best lineage and creation times. Once the distribution of creation times is inferred, we sample a creation time for each binary, and iteratively maximize the lineage and creation time inference. The sampling process is repeated to avoid local maxima, the best lineage over all random starting points is selected as the algorithm output.\label{LineageAlg}}
\end{figure}

As shown in Eqn.~\ref{lineageMax}, the lineage that is provided by the MAAGI system to the user is the maximal probability lineage given the creation times and the malware features. Since the creation times are unknown, we must infer both the lineage and the creation times jointly. To accomplish this, we employed an iterative algorithm (similar to expectation--maximization) to jointly infer the most likely binary creation times and lineage. 

The algorithm runs as follows, and is also shown in Fig.~\ref{LineageAlg}:

\begin{itemize}
  \item \textbf{Infer a distribution over the creation time of each binary.} Using the observable time stamp and time seen in the wild information, we infer a distribution of the possible creation times of each binary. This distribution is still conditioned upon the lineage; this process simply marginalizes out some of the information not needed to compute a lineage (time stamps, first observations, obfuscation). 
  \item \textbf{Sample the creation times.} We take a sample from the creation time distributions of the malware binaries. This creates a fixed ordering over the binaries.
  \item \textbf{Infer the most likely lineage.} We infer the most likely lineage of the malware binaries given the fixed creation times. That is, we compute the lineage described in Eqn.~\ref{lineageMax}.  
  \item \textbf{Infer the most likely creation times.} We now infer the most likely creation times of the malware given the most likely lineage just inferred. The most likely creation times is defined as
  \begin{equation}
    Times_M = \underset{Times_{M,i}}{argmax}\,\,P(Times_{M,i} | Features_M, Lineage_M)
  \end{equation}
  Since we are conditioning on the previously computed lineage, we fix the inheritance relationship between two binaries, but the direction of the edge can change depending on the inferred creation times. This means that after finding the most likely creation times, some of the parent/child relationships in the lineage may have reversed.
  \item \textbf{Repeat steps 3 and 4 until convergence.} We repeat the dual maximization process until convergence, which can be defined as either a fixed number of iterations or until there are no more changes to both the mostly likely lineage and the creation times.
  \item \textbf{Repeat steps 2--5 until enough samples have been collected.} Since there is no guarantee that the maximization process converges on the global maximum, we restart the process to increase the likelihood that the global maximum is reached. Let $Lineage^j_M$ and $Time^j_M$ be the lineage and creation times computed on the $j^{th}$ restart of the maximization process. Then the lineage and creation times returned to the user is
  \begin{equation*}
    (Lineage_M, Times_M) = \underset{Lineage^j_M, Times^j_M}{argmax}\,\,P(Times^j_M, Lineage^j_M | Features_M)
  \end{equation*}
\end{itemize}

The algorithm is very similar to the expectation maximization algorithm, but we must re--sample several initial malware creation times to ensure that we are finding a satisfactory maximum value. At each iteration, we find the most probable lineage and then subsequently the most probable creation times. So every iteration of the algorithm increases the joint probability of both lineage and creation times. At the end of the algorithm, we select the lineage with the highest probability (from the multiple restarts of the algorithm) as the lineage of the set of malware.

\subsubsection{Learning}

Inferring the creation times of the malware depends upon some variable parameters which are generally unknown. For instance, we need to know the prior probability that a binary's time stamp has been obfuscated, and any parameters of the conditional probability of the time the malware was first seen in the wild given its creation time. To determine these parameters, we learn them on a training set of data, as shown in Fig.~\ref{LineageLearning}.

\begin{figure}[t]
\centering 
\includegraphics[width=.75\columnwidth]{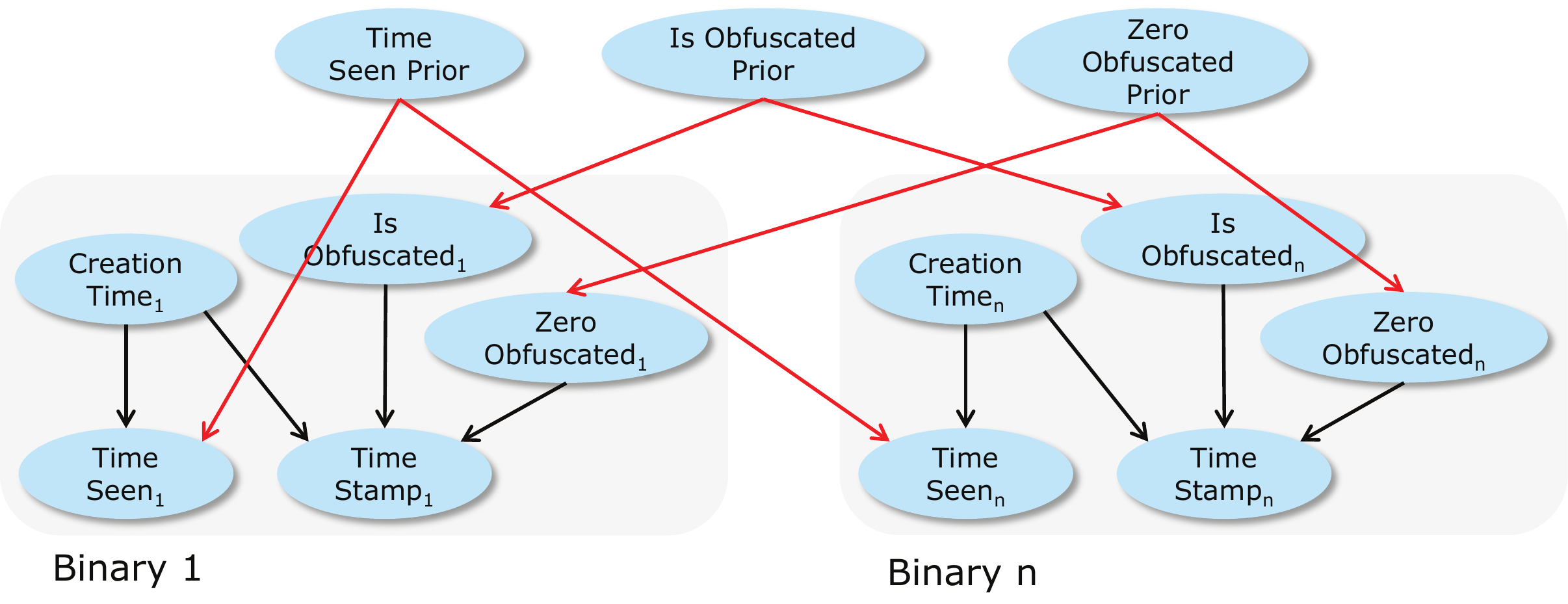}
\vspace{-5mm}
\caption{Learning the priors on the creation time model. The priors model the time between when a malware was created and when it is first seen in the wild, and the levels of obfuscation that are seen in malware.\label{LineageLearning}}
\end{figure} 

We create three prior variables, each of which is a beta distribution. The obfuscation variables that use the priors are bernoulli random variables, whereas the conditional probability distribution between the time seen and creation time is an exponential distribution (yet it still uses a beta distribution as its prior).

The learning process is as follows. We sample a value of the priors and compute the probability of the evidence of the model given the prior values, were the evidence is the time seen and time stamp information for each binary. This probability of evidence is used as a weighting factor in a maximization algorithm. That is, we explore the space of the priors and eventually converge upon the value of the priors that produces the highest probability of the observed evidence. Once the parameters are learned they can be used for any lineage on the malware.

\subsubsection{Implementation and Experiments}

We implemented the creation time model, the lineage model, and the learning model in Figaro, an open--source probabilistic programming language. All of the maximization used in the lineage analysis tool (lineage inference and learning priors) used Simulated Annealing to compute the most likely values. The Metropolis--Hastings inference algorithm was used to infer the distribution of creation times for the binaries.

\paragraph{Inferring the Creation Time Distributions}

The first step in our algorithm is to infer a distribution of the creation time for each malware binary. Although this distribution is used in the lineage algorithm, it can be computed and tested independently. Unfortunately, it is extremely difficult to determine the true creation time of wild malware binaries, hence we resorted to using synthetic data. 

For a set of $N$ synthetic binaries, we first created ground--truth lineages from randomly generated Prufer sequences~\cite{prufer1918neuer}. We then generated synthetic creation times for each synthetic binary with respect to the lineage constraints (i.e., parent time must be before child). We used a synthetic time scale from 1 to 10,000. We then generated the time seen information from each binary by modeling the distance between the creation time and time seen as either a geometric distribution or a uniform distribution. We also varied the obfuscation of the synthetic time stamps from 5\% of the data set to 95\% to replicate real--world conditions.

\begin{figure}
\centering 
\subfigure[Expected Error]{
\includegraphics[width=.3\textwidth]{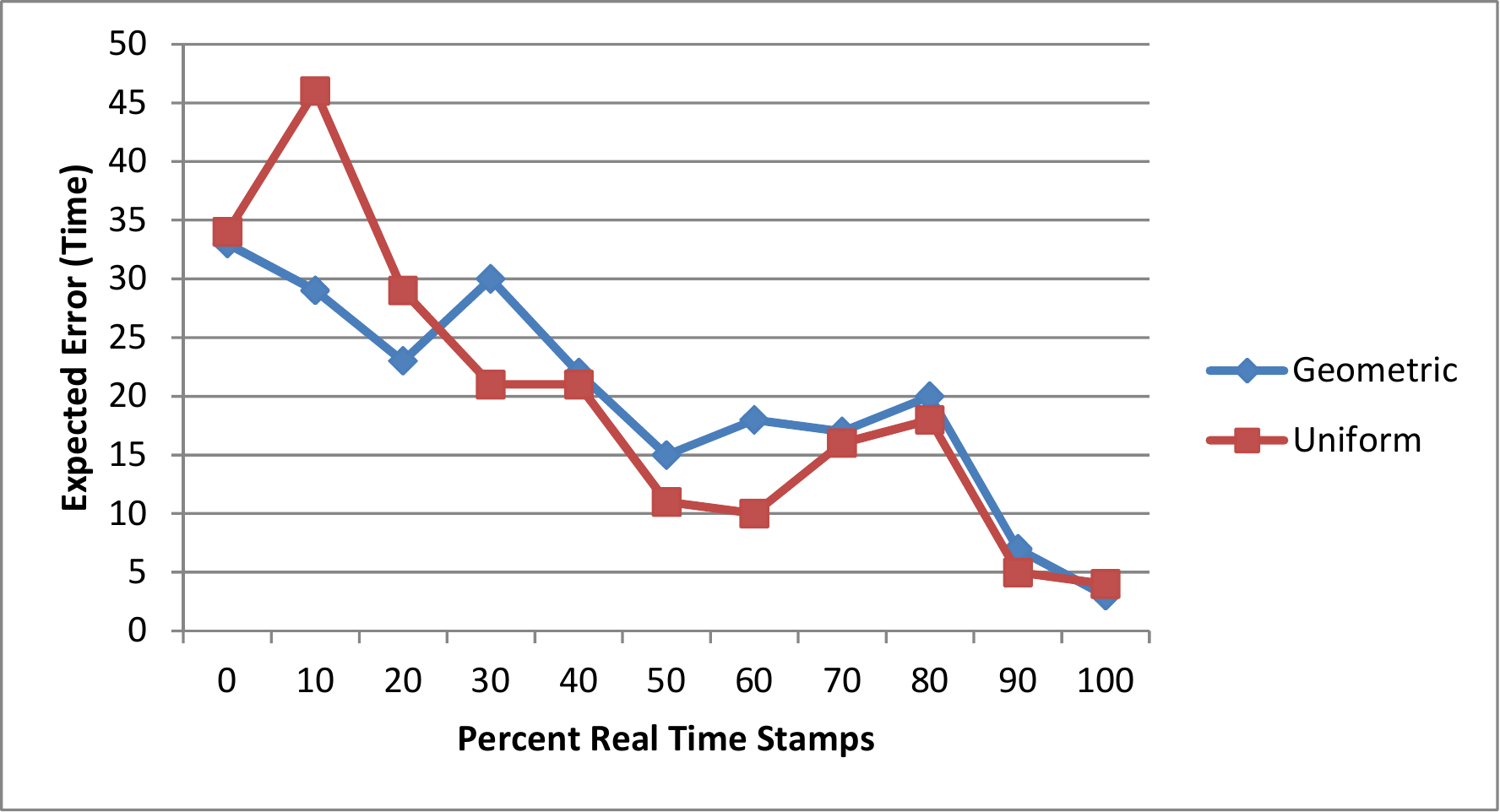}
\label{LineageTimeError}
}
\subfigure[Probability of Error]{
\includegraphics[width=.3\textwidth]{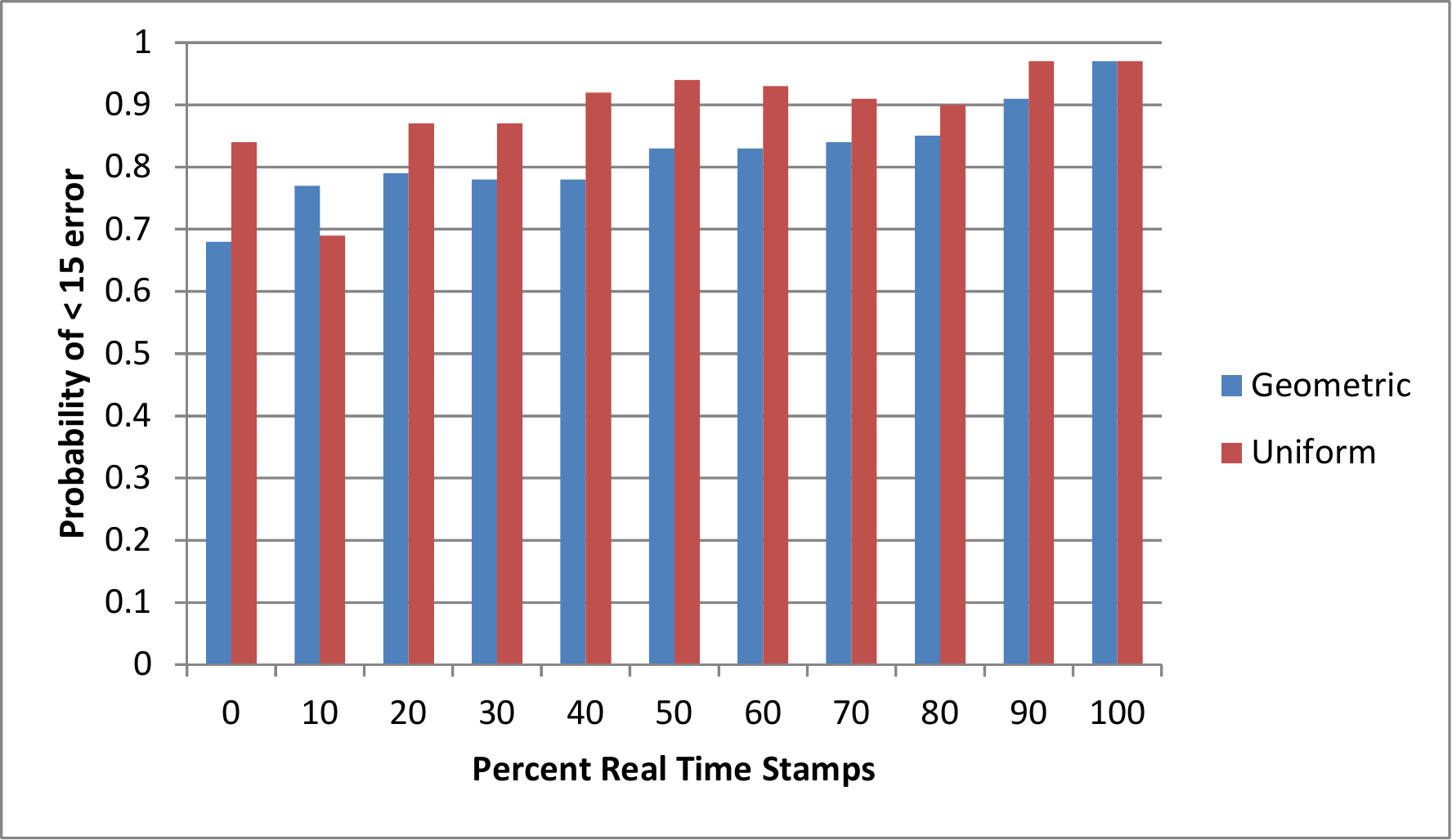}
\label{LineageTimeProb}
}
\subfigure[Reduction in Error after Lineage]{
\includegraphics[width=.3\textwidth]{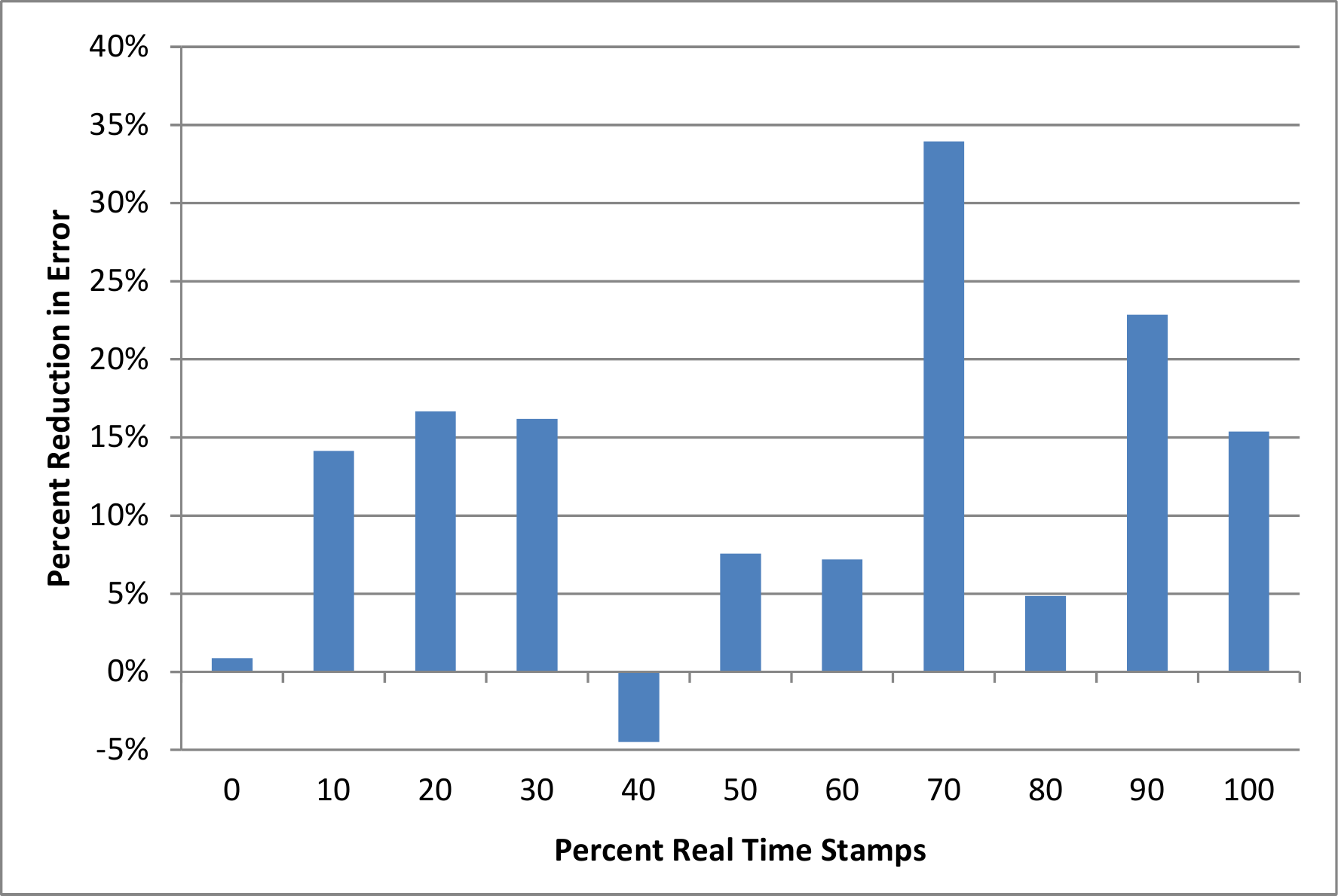}
\label{LineageTimeImprove}
}
\caption[]{Results of the creation time inference on the synthetic data set. \label{lineageTimeResults}}
\end{figure}

We generated 50 random lineages, and computed the expected error between the inferred time stamps and the actual time stamps. The results are shown in Fig.~\ref{LineageTimeError}. As can be seen, as the data set becomes less obfuscated, we are able to estimate the actual creation time of the binaries with higher accuracy. We also see that the expected errors are very similar using either a geometric or uniform distribution to generate the time seen information (and note that the algorithm has no knowledge of the generative model of time seen from creation time). 

To put this error in perspective, we also show the probability that the estimated creation time is within 15 time “units” of the actual creation time in Fig.~\ref{LineageTimeProb}. The number 15 is significant here as this the expected time between a parent to its child (this number was set during the random lineage generation). The uniform model of time seen generation has less error in this case since it has lower variance than the geometric. 

Finally, during the lineage phase of our algorithm, we iterate between improving the lineage and improving the creation time. So we wanted to determine the improvement of the creation time estimation before and after the lineage portion of the algorithm is run, shown in Fig.~\ref{LineageTimeImprove}. As can be seen, for most of the tests, the error in the creation time estimation was reduced after running the lineage estimation. This demonstrates that our initial lineage hypothesis is correct; lineage can inform the creation times as much as creation times can inform the lineage.

\paragraph{Inferring Lineage}

Next, we tested the actual lineage estimation algorithm. Since we do not have ground truth wild malware, we tested our algorithm on 17 lineages generated by MITLL. Since the binaries in this data set are clean room generated data (i.e., non--wild) they do not have data for the time they were first seen in the wild. To test our algorithm on this data, we had to create an artificial time system that generated new creation time, time stamp, and time seen data for the malware. This synthetic time generation was the same process as described for the creation time testing. Although the time data was artificial, it was still generated according to the ground--truth lineage. We also obfuscated a portion of the time stamps to make it more realistic. Fig.~\ref{LineageResults} shows the results of the test.

\begin{figure}
\centering 
\includegraphics[width=.5\columnwidth]{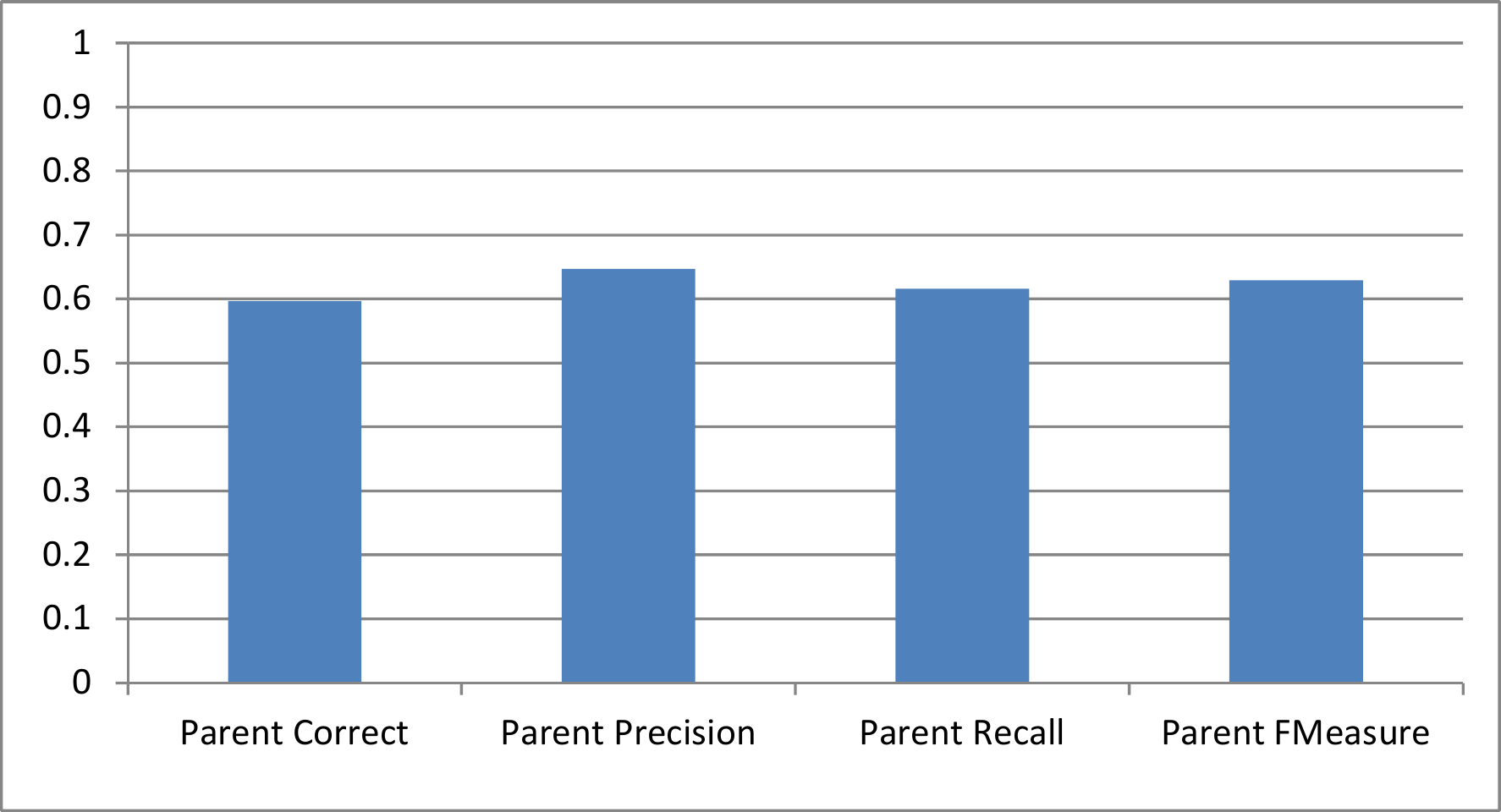}
\vspace{-1mm}
\caption{Results of the lineage estimation on the MITLL data. \label{LineageResults}}
\end{figure} 

The figure shows the accuracy of the lineage algorithm on four metrics. As can be seen, our lineage estimation algorithm can reconstruct lineages of malware with fairly high accuracy. These results demonstrate that our lineage analysis tool is able to reconstruct the lineage of a set of malware, and can be successfully used to analyze a set of wild malware.

As indicated by our experiments, AI tools and methods can be used to accurately estimate the lineage of a set of malware. This type of tool will prove invaluable to malware analysts in their quest to understand the provenance and evolution of malware. As malware authors increase the volume and sophistication of their attacks, we see AI playing a large role reconstructing increasingly complex lineages.

\subsection{Trend Prediction}

Currently, cyber defense is virtually always responsive. The attacker initiates an attack and the defenders need to find a way to respond to the attack, mitigate its damage, and prevent it from reoccurring, a process that invariable favors the attacker. Preemptive cyber defense has the potential to shift the advantage away from the attacker and on to the defender. In preemptive defense, the defender tries to predict future attacks and respond to them before they happen, yet requires an effective capability to predict future cyber developments. 

To date, most malware trend analyses are retrospective analyses of existing malware~\cite{cruz2014mcafee}, or are predictions that reflect the opinions of human experts and have a broad, high level focus~\cite{blanch2014mcafee,websense2012}. This type of prediction may not be of practical use to a malware analyst, as they are not able to quantify the threat level of a malware family or trend. A number of tools and systems~\cite{kang2012malware,deerman2012,juzonis2012specialized} have been proposed for detecting novelty or anomaly among individual malware binaries. These do not include the notion of a malware family, do not consider changes in the families over time, nor do they make long term predictions about the behavior of a family over time.

The goal of this analysis tool is to predict how malware families will evolve, and to determine which families will become the most harmful. Prediction is challenging since the space of possibilities is enormous and attackers try to be unpredictable. Nevertheless, the MAAGI trend analysis tool uses modern machine learning methods to make effective predictions even in difficult situations. As prediction about the behavior of a large family of malware is generally not useful (most of the damage has already been done), this tool only needs a single binary from a malware family to provide significant and actionable information about the future behavior of the family, such as the overall prevalence of the family. 

\subsubsection{Prediction Problem}
\label{predprob}

\begin{figure}[t]
\centering
\includegraphics[width=0.75\textwidth]{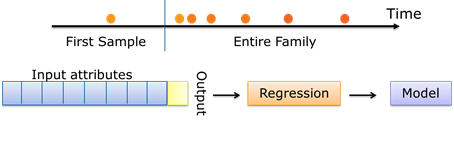}
\vspace{-1mm}
\caption[]{Setup of prediction experiment: Families are ordered temporally based on their first-seen date, then the earliest appearing sample in each family is used to predict its future characteristics}
\label{modelsetup}
\end{figure} 

When a new malware binary is encountered(i.e., a malware not clustered into an existing family), the MAAGI prediction tool attempts to predict various characteristics of the entire family from the single instance from the family, as shown in Fig.~\ref{modelsetup}. The tool first extracts a set of features from the binary which are used as the input values for the prediction algorithm. Attributes at the family level are ultimately derived from the features of the binaries they contain. Temporal aspects of the family are obtained from an external source such as VirusTotal. 

The list of features extracted from a malware binary is as follows:
\begin{enumerate}
 \item The number of blocks, library imports, strings, and procedures contained in a binary.
 \item The average length of a block in a binary.
 \item The number of rare and unique blocks, strings, imports and procedures contained in a binary. We consider a rare feature to be one that appears in less than 10 other binaries.
 \item The average procedure complexity of a binary, using a measure such as cyclomatic complexity.
 \item Whether the compiler timestamp is plausible.
 \item The similarity of a binary to representative binaries of a existing malware families. We use the similarity measure used in the hierarchical clustering technique described in~\ref{clustering}.
\end{enumerate}

We defined seven primary features of malware families that the prediction algorithm estimates from a single malware:

\begin{enumerate}

 \item \textit{Continued Growth} is a binary value which indicates whether or not any binaries are added to the family in the future. 
 \item The \textit{Future Size} of the family is the number of binaries which will be added to the family in the future. 
 \item The \textit{Distinctiveness} is the level of similarity between a family and all other families. 
 \item The \textit{Mobility} of a family is the degree to which it changes over time. Computationally, this is the level of similarity of the family in the present to itself in the future. 
 \item \textit{Sophistication} is the degree to which the binaries in the family become more complex. This is computed by computing the average number of blocks which the family contains.
 \item The \textit{Time Till Stability} is the time at which novel binaries are no longer being added to the family. This is determined by looking at the first--seen date of the oldest binary in a family.
 \item \textit{Search Engine Response} is the number of hits that a major search engine (Bing) returns for a query on the binary’s MD5. This value is averaged out across the family. 
\end{enumerate}

The MAAGI system can make predictions about these characteristics individually, or they can be combined together to define an impact score.

\subsubsection{Prediction Methods and Training}

The prediction framework consists a set of prediction algorithms applied to a malware binary supplied by the user. Before this can occur, however, we need to select a set of training data, a prediction model or algorithm, and a testing set of data. We first use a feature aggregation program to extract features from a feature database and create a dataset suitable as input to machine learning models. A set of families and binaries are provided as input to the aggregation program, along with a date range and cutoff time. Each model is trained on the input features of the earliest binary in each cluster (family) and predicts the corresponding output characteristics of the full family. The models are trained separately against each output. %However, we can also define an impact score as a single field combining all of the outputs described in section \ref{results}. 

We have applied a variety of machine learning methods to learn the output characteristics of families from a single binary. This includes generalized linear models such as logistic\cite{hosmer2000introduction} and Poisson regression\cite{draper1981applied}; genetic algorithms learning formulas of input features\cite{schmidt2009distilling}; neural networks\cite{dreiseitl2002logistic}; support vector machines with sigmoid, polynomial, radial basis and linear kernels\cite{gunn1998support}\cite{CC01a}; CN2 rules\cite{clark1989cn2}, and classification trees and random forests\cite{quinlan1986induction}. Among the SVM models, we found that using the radial basis kernel exhibited the best results. %We primarily employed three software packages to build our models: We used packages available in R\cite{team2012r} and Orange\cite{JMLR:demsar13a} to create the models in our experiments. To create models learned by genetic algorithms, we used Eureqa\cite{schmidt2009distilling}, a software suite for genetic programming.

Each of these machine learning algorithms is input to an ensemble learning process. This means that after we trained a first set of 'tier--1' models, we used the output values produced as inputs to a second tier of models. We refer to the second set as the `tier--2' models. There are temporal considerations here as well: We sub-divide the training and testing data a second time to ensure that none of the models in either tier are trained on families having data in the test set. 

In each experiment, the data is split into training and testing sets. The training set contains 70\% of the data set binaries, with the remaining 30\% for testing. However, care must be taken that we do not leak information about future binaries in the training data set. That is, we must ensure that all binaries in the test data set appear temporally after the oldest binary in the test data set. To ensure the temporal correctness of the data, we mark the point in time after 70\% of the binaries have appeared, and use families first before this point in time for training, and the remaining binaries for testing. We must account for changes in the features space with respect to time. Because we consider features such as the number of unique blocks or unique strings contained by a binary, we build the features incrementally. Features are counted as unique if they were unique when the binary first appeared. 
 
\subsubsection{Experiments and Results}

While the MAAGI system can make predictions about the features explained in Sec.~\ref{predprob}, we only detail the experiments that focused on the anticipated continued growth of the family. That is, these experiments detail a classification problem: Given a single binary, will the family ever grow beyond one?

\paragraph{Data Sets} We were provided by DARPA with a large dataset totaling over a million malware binaries. To obtain information about when a family was born, we used the family time stamps assigned by VirusTotal, an online virus database. We assume that the first binary in the family is the one with the earliest `first-seen' date. Many binaries did not have a `first-seen' field and were discarded. The dataset contained families which first appeared between 2006 and 2012, with the most binaries appearing in 2008 and 2009. The data contained a considerable amount of benignware and many binaries were packed (and thus could not be analyzed). For the experiments described here, we selected approximately 100k unpacked binaries with temporal information and families and supplied them to our hierarchical clustering technique described in Sec.~\ref{clustering}. %Not every binary in this subset had a label assigned from each anti--virus vendor, and so the total number of binaries varies from one family definition to another.

\paragraph{Evaluation} We evaluated the models' accuracy, precision and recall. These metrics are all derived from the true positive (TP) and false positive (FP) rates of the classifier. Accuracy the fraction of correct predictions out of the total number of predictions. Precision is the fraction of actual positives out of the number predicted positives. Recall is the fraction of correctly predicted positives out of total number of actual positives. The calculations used to obtain these performance metrics are shown below: 

\begin{equation*}
\scriptsize
\begin{split}
  & Accuracy = (TP + TN) / (TP + FP + TN + FN)\\
  & Precision = TP / (TP + FP) \\
  & Recall = TP / (TP + FN) \\
\end{split}
\end{equation*}

Another useful metric for evaluation is the F-measure~\cite{powers2011evaluation}, which is defined as 

\begin{equation*}
	F = 2*\frac{Precision \cdot Recall}{Precision+Recall}
\end{equation*}
F-measure assesses the accuracy of a model by considering rate of true positives identified by the model. A higher value of F-measure indicates a better score.

We chose to target F-measure and recall as the most important metrics in our evaluation. Intuitively, we believe that correctly identifying a highly threatening family is more important than correctly classifying a sample as unimportant or benignware. We can tolerate a higher rate of false positives if the true positives are correctly identified.

\paragraph{Results}

To classify a single binary family as either one with continued growth or not, we first train a ‘tier--1’ set of classifiers on the original data. The results from these classifiers were added to the test data to create a new data set, which was used for 'tier--2' learning to see if we could learn a way to combine these results to perform better than any single learning algorithm.

As a baseline, we include the metrics obtained by a majority classifier which always assigned the most common occurring value in the dataset. In this data, the most commonly occurring classification is that the families remain singletons. Because the majority classifier always outputs the same value, it has a score of 0 for the metrics besides accuracy. 

Figures \ref{fig:classifier1} and \ref{fig:classifier2} show the performance of the models. The 'tier--1' classifiers are shown in Figure \ref{fig:classifier1}, followed by 'tier--2' in Figure \ref{fig:classifier2}. 

\begin{figure}
\centering 
\subfigure[Evaluation of tier 1 models]{
\includegraphics[width=.45\textwidth]{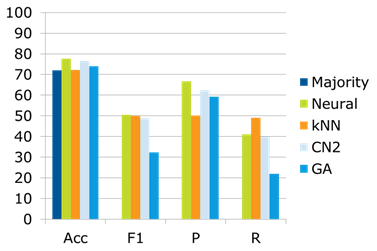}
\label{fig:classifier1}
}
\subfigure[Evaluation results of tier 2 models]{
\includegraphics[width=.45\textwidth]{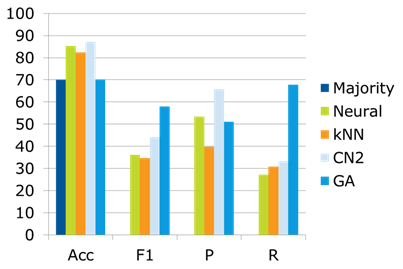}
\label{fig:classifier2}
}
\caption[]{Evaluation of the prediction tool.}
\end{figure}

The basic classifiers were able to perform only a mild level above the baseline majority classifier. However, the difference between the first and second tier classifiers shows promising results. Using an ensamble classifier based on several 'tier--1' classifiers, the MAAGI system is able to obtain better results for all metrics, and a significantly stronger value of recall using a genetic algorithm. In the malware defense domain, it is much more important to correctly recognize a significant family than an unimportant one, and so we feel this is a strong result.

The prediction analysis tool in the MAAGI system represents an important landmark towards cyber prediction and pre-emptive malware defense. To our knowledge, this capability is the first of its kind in malware analysis and defense, prediction of future malware developments can greatly improve cyber--defensive efforts. The MAAGI system has only scratched the surface of malware prediction, integration of more advanced artificial intelligence technology into the system could enable even more fruitful prediction possibilities.

\subsection{Functional Analysis}

The goal of functional analysis (FA) is to more thoroughly and deeply analyze malware with the goal of characterizing its function and context. Functional analysis detects multiple purposes that exist simultaneously in a malware sample. It uses structured, contextual reasoning to infer complex purposes. For example, credential theft is a complex purpose that involves credential capturing and data exfiltration, each of which may be realized in many ways. Included in this reasoning are basic characterizations of the attack and attacker, such as degree of sophistication.

Functional analysis uses only the results of static analysis. We have found that in many cases, a malware sample does not exhibit its behaviors in a sandbox, so we cannot rely on dynamic analysis. However, using only static analysis is more difficult, because API calls in the static code are not inherently ordered, so we cannot assume ordering of API calls in the analysis.

For functional analysis, we use parsing techniques derived from natural language processing (NLP). Because we only use static analysis results, we cannot directly use NLP methods that depend on the ordering of words in a sentence. We must treat a sentence using a “bag of words” model and try to identify components of meaning out of it. Also, just like in language we can analyze writing competence based on the sentence and vocabulary complexity, so in malware can we identify attack and attacker characteristics based on the words used and their organization.

Our representation is based on the Systemic Functional Linguistics, which organizes Grammars (SFGs) according to the function and context of language. They provide a functional hierarchy for meaning and purpose, and an overlaid context hierarchy for context. So, it is possible to reason about the purpose and context from the words. Similar mechanisms can be used to reason about function and characterizations of malware. One nice feature of SFGs for our purposes is that unlike most NLP frameworks, which require a fully ordered sentence, SFGs accommodate partial orderings by treating orderings as constraints. In domains such as malware where no ordering is available, we can simply not specify any ordering constraints.

Our approach to functional analysis is as follows:
\begin{enumerate}
  \item Static analysis creates the call graph and a set of searchable items, such as system API calls and strings.
  \item A lexification system extracts items of interest to functional analysis, such as particular strings that are indicative of specific functions. These items are annotated with semantic tags, where relevant.
  \item These unordered items are parsed using SFGs. This parsing process can produce multiple alternative parses. Each parse represents a different way of accomplishing an attack.
  \item We prioritize these alternative parses based on the minimum graph span in the call graph. This is based on the heuristic that related behaviors will be taken at times that are close to each other.
  \item We produce a human-readable report that summarizes the results.
\end{enumerate}

\subsubsection{Functional Analysis Design}

\begin{figure}
\centering 
\includegraphics[width=.9\columnwidth]{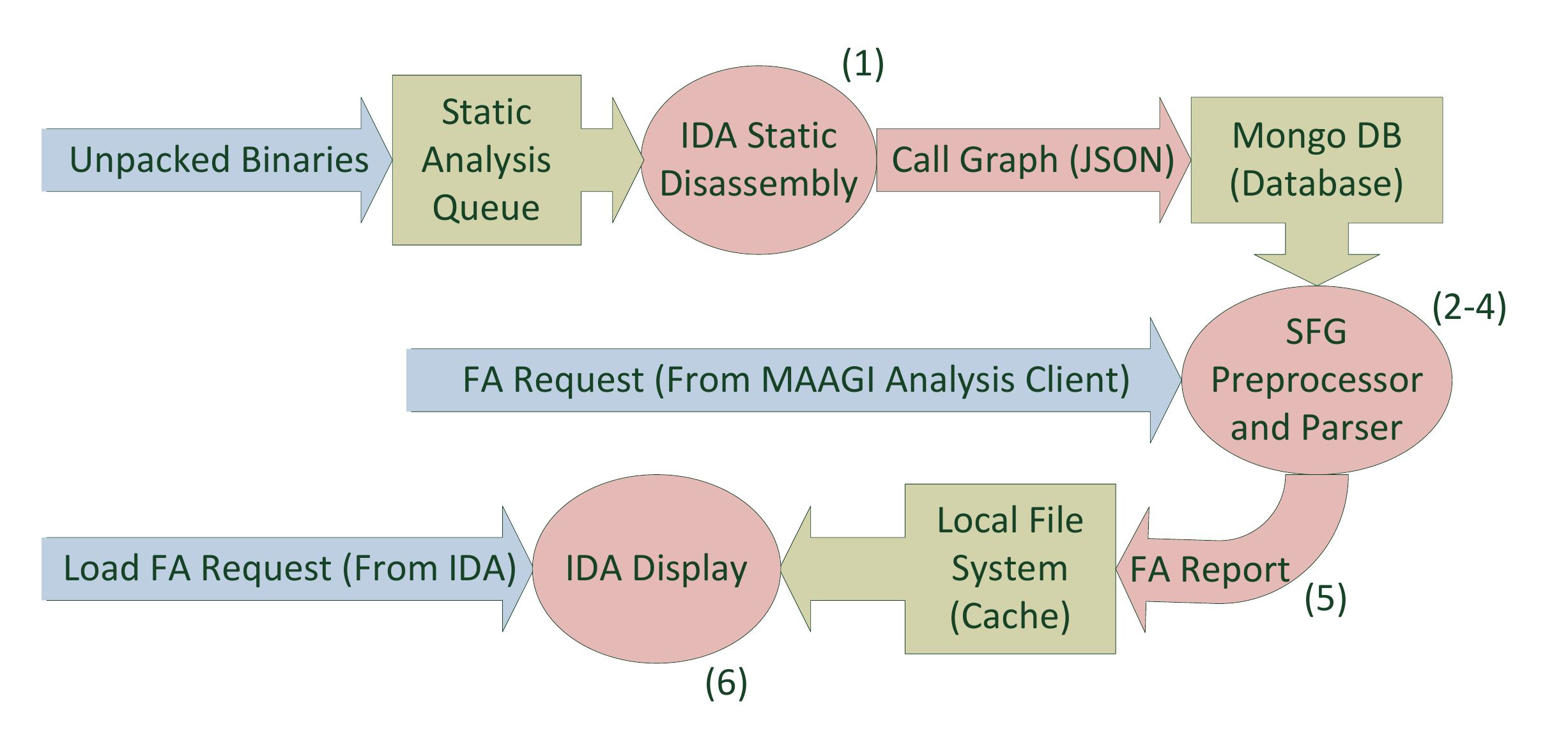}
\vspace{-5mm}
\caption{Functional analysis workflow.\label{faworkflow}}
\end{figure}

\begin{figure}
\centering 
\includegraphics[width=.7\columnwidth]{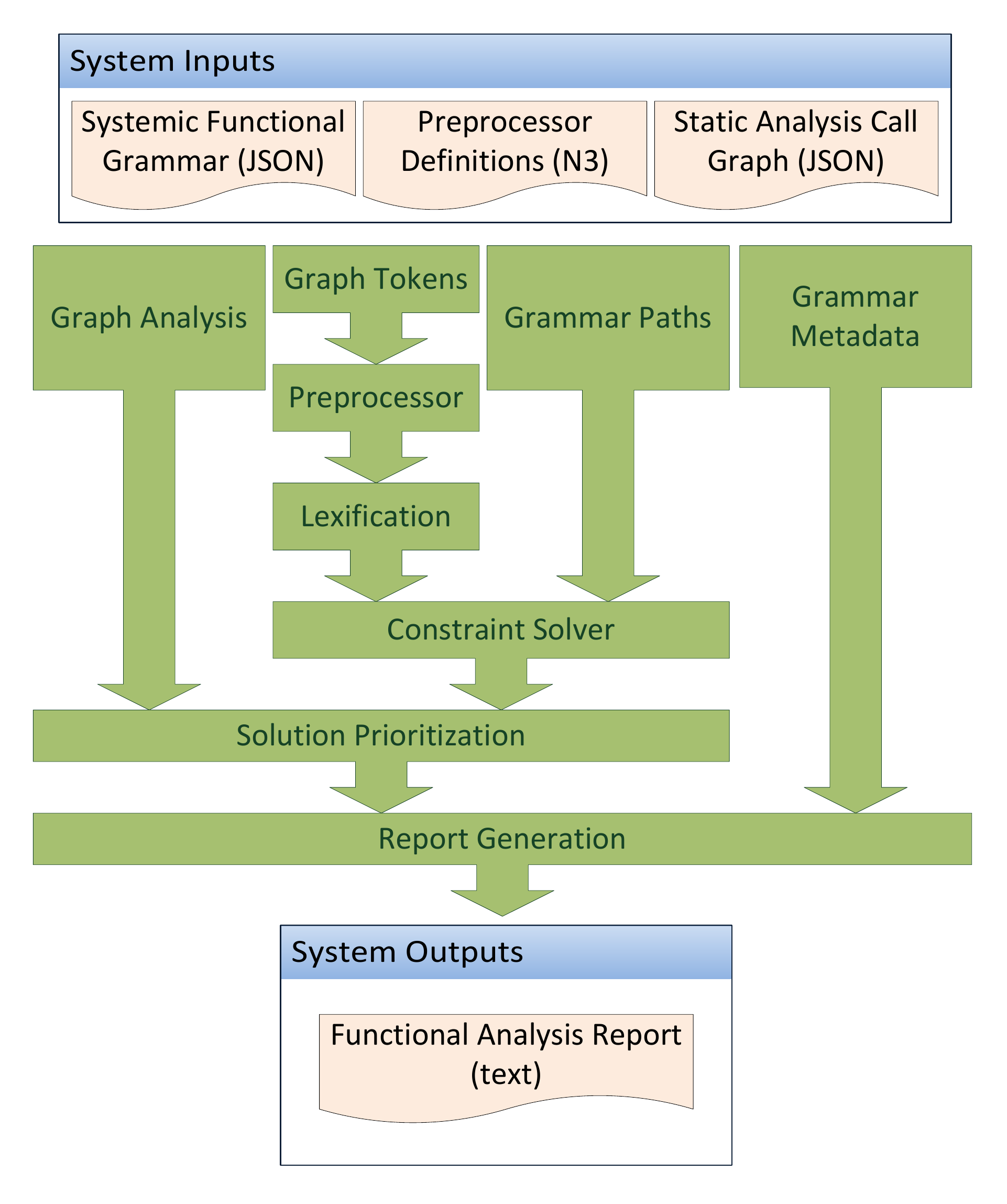}
\vspace{-5mm}
\caption{Functional analysis workflow.\label{fadataflow}}
\end{figure}

Fig.\ref{faworkflow} describes the functional analysis workflow. Creation of a static analysis call graph, Step 1, above, is transparently triggered whenever a sample is submitted for static analysis. Steps 2-5: Lexical Analysis, Parsing, Parse Prioritization, and report generation; are initiated by a request from the user, and occurs within the functional analysis software component (detailed in the next section).  
  
%This integration was natural due to IDA’s existing role as both a static analysis component and a consumer of analysis results. Automatic analysis through IDA was thus extended to produce a JSON call graph representation as well. The MAAGI User Interface was modified to launch Functional Analysis, and the IDA plugin was modified to display existing FA results.

As seen in Fig.~\ref{fadataflow}, the system performs analysis in several distinct stages, each passing its results to the next. The system takes as inputs three sets of data. The first two represent domain knowledge and the inferences that the system is attempting to make. That is, the Systemic Functional Grammar which captures generalizations and characterizations of functions of malicious software, and the Preprocessor Definitions file, which supplied preprocessing instructions and grammar metadata. These files represent domain knowledge and tend to be static unless an analyst wishes to introduce new domain knowledge. The third input, a Static Analysis Call Graph, is generated by the static analysis pipeline for each sample to be analyzed.
 
The output of the system is a report which summarizes the parses that were satisfied by the tokens found within the call graph. This report is human readable and is displayed in the MAAGI user interface, in addition to being able to export and loaded the report into external malware analysis tools (such as IDA, a disassembler). The current templates for report generation make use of IDA’s context--aware console window to provide active content. Active content allows users to click within the console window and be directed to the binary location which contains identified functions.

\subsubsection{Grammar Representation}

In the context of the MAAGI system, a grammar is composed of two interconnected \textit{strata} – a \textit{grammatical stratum} and a \textit{contextual stratum}. The contextual stratum and the grammatical stratum are structurally similar, with each an instance of a system network.  Each plays a different role in the overall representation, but both are built from the same set of components. We describe that structure below.

\paragraph{System Networks}

\begin{figure}
\centering 
\includegraphics[width=.6\columnwidth]{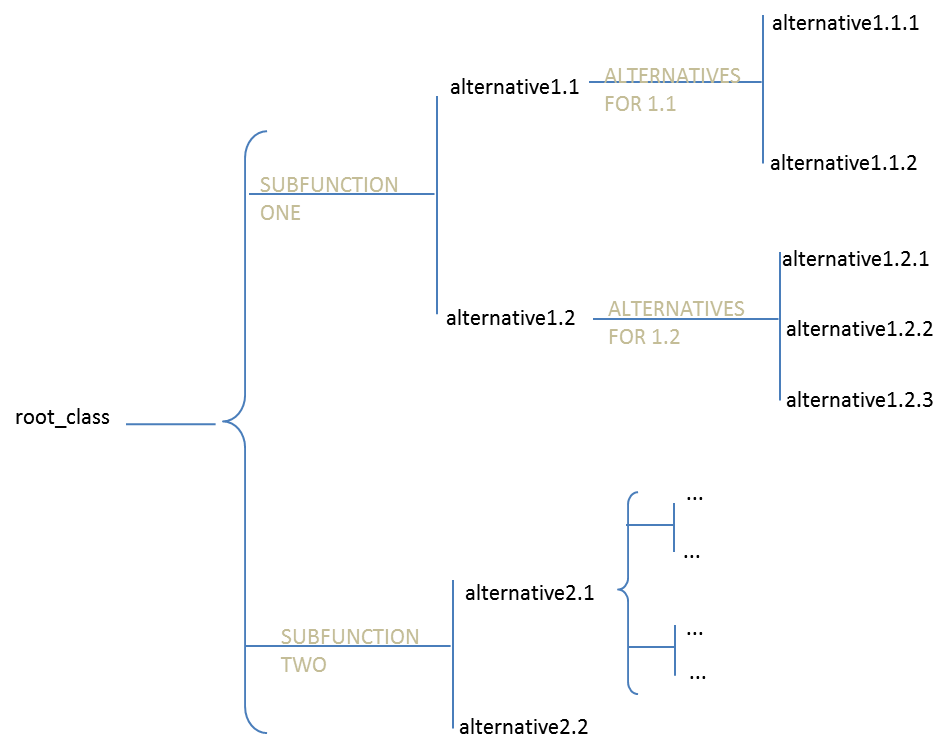}
\vspace{-5mm}
\caption{Functional analysis workflow.\label{exampleNetwork}}
\end{figure}

\begin{figure}
\centering 
\includegraphics[width=.6\columnwidth]{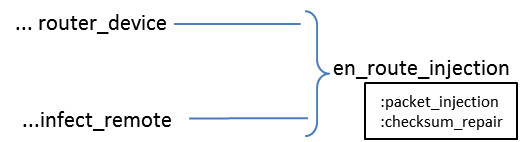}
\vspace{-5mm}
\caption{Functional analysis workflow.\label{otherExample}}
\end{figure}

A system network describes a hierarchy of systems, which represent exclusive alternatives that might hold in the presence of a given feature in the input. For example, if a sentence contains a noun, that noun can be countable or uncountable, common or proper, concrete or abstract, etc. These are orthogonal properties that can each be represented by a system with “noun” as the entry condition. An example system network is shown in Fig.~\ref{exampleNetwork}.

As noted in the figure, orthogonal choices (indicated by curly braces) represent an "AND" relationship. In the case in Fig.~\ref{exampleNetwork}, for example, this indicates that a choice must be selected in both the "SUBFUNCTION ONE" and "SUBFUNCTION TWO" systems. Exclusive choices (indicated by “T” shaped branches) represent an exclusive choice or "XOR" relation. At most one feature class must be satisfied for each such instance in the grammar. A system network optionally defines forward chaining rules, or gates, which take multiple feature classes as entry conditions (described in the rules section below).

The grammar in the MAAGI system is composed of two system networks (grammatical and contextual), as indicated above. We refer to these as strata, as they represent two distinct, but interrelated levels of abstraction for viewing the input. The grammatical stratum is what many people mean when they use the term “grammar” in SFG.  It contains a hierarchy of functional choices.  These functional choices are realized when the constraints, i.e. realization rules, associated with them hold for a given set of input. A valid parse is said to exist when there is a path to a leaf where all realization rules hold for every orthogonal choice in the stratum. The contextual stratum uses the same notations as the grammatical stratum but it represents a hierarchy of the salient contexts. These contexts may alter the likelihood or validity of particular choices in the grammatical stratum.

\paragraph{Realization Rules}

Realization rules are constraints which define which elements in the input correspond to given functional roles, the mappings or equivalences between those roles and others, as well as the order in which those roles must appear. Some example rules include:
\begin{itemize}
  \item Lexification: Lexification draws a correspondence between a string present in the input and a textual role. 
  \item Insertion: Insertion requires that a particular role must be observed in the data. Typically the role is one that is lexified directly, or one that is lexified indirectly via expansion. This is used as a mechanism for representing expectation.
  \item Expansion: Expansion allows for groups of roles to be referred to collectively.
  \item Preselection: Preselection is used to map contextual classes to grammatical classes. In such cases, preselections define a many-to-many mapping between contextual classes and grammatical classes.
  \item Gates: Gates are forward chaining rules that represent classes that stem from an arbitrary Boolean combination of other classes. They are directly analogous to logic gates. For example, in Fig.~\ref{otherExample}, the class "gate1" is inferred if alternative1.2 is chosen with either alternative2.1 or alternative2.2.
\end{itemize}

Other rules include definitions for Adjacency, Partition, and Conflation.

\subsubsection{Parsing and Classification}

In essence, SFG strata are analogous to decision trees, with realization rules serving as queries for establishing which conditions (vis a vis features) are observed in a given set of input. Each permutation of choices in the contextual stratum represents a possible context in a given domain, and the permutation of choices in the grammatical stratum that remain after preselections from that context are applied represents the space of possible interpretations of the data for the context. Conversely, given an interpretation that is known to hold, we can infer the contexts that support it. We refer to the set of features in an interpretation and its context as a traversal. The set of realization rules associated with a traversal are the set of queries that must be satisfied in order for the features of that traversal to be valid for the input data. As noted above, these rules take the form of constraints, and thus a traversal describes a Boolean expression over features. 

The MAAGI system classifies an attack by generating traversals through the grammatical stratum of the grammar, constructing the associated expression from the queries associated with its features, and passing the result to a constraint solver (Copris) (CITE). From the set of possible valid traversals, we can then infer the possible contexts in which the associated attacks could have been carried out. Multiple realizations can exist, however. For this reason, realizations are ordered by the average proximity of the binary's API calls lexified in the traversal. Applicable contexts to each are ordered heuristically.

\subsubsection{Testing and Experimentation}

The functional analysis in MAAGI has been able to reliably operate over a set of 283 unpacked Banker-family malware samples (CITE) while providing accurate characterizations of its command and control behavior, as well as its theft of personal information. After capturing a large number of alternative approaches for these and related malicious activates, we have been able to note some of the limits and successes of functional analysis in MAAGI.

Presently, the MAAGI Systemic Functional Grammar (SFG) Parser and grammar representation language are capable of expressing an impressive variety of malware behaviors and actions. In our current environment, however, the lack of data--flow and temporal information can limit the variety of behaviors expressed. When two behaviors differ only with respect to data--flow or runtime temporal characteristics, the grammatical representation of these two behaviors is identical. Often, this information is the key discriminator between a malicious sample and a benign sample with the same capabilities. Currently, the system compensates for the lack of this data by using call graph distance as a proxy for temporal proximity. This issue should not be interpreted as a limitation of functional analysis, but rather as a limitation of using modern static analysis for producing tokens/annotations later consumed by the parser. Temporal and data flow features can absolutely be represented in the grammar and could be leveraged to better discriminate samples with meaningful differences with respect to those features. However, it is difficult or impossible to glean this information from Static Analysis alone. 

Despite these limitations, the functional analysis tool will provide malware analysts the ability to understand the behavior and function of the malware. Incorporating more advanced AI techniques, however, will allow the user to fully take advance of the wealth of information that can be extracted from malware using SFGs.

\section{Conclusion}

In this paper, we introduced the MAAGI malware analysis system. This system takes advantage of many of the artificial intelligence advances over the past several decades to provide a comprehensive system for malware analysts to understand the past, present and future characteristics of malware.

As a practical application of artificial intelligence, the MAAGI system will have an enormous impact on the cyber--defense community. While the cyber--defense community has always employed artificial intelligence methods to analyze malware, to the best of our knowledge, nobody has attempted to integrate these technologies into a single system where various analyses use the results of other analyses to model, learn and predict characteristics of malware. It is our hope that the MAAGI system also inspires other innovative applications of artificial intelligence to the cyber--defense community. Towards this goal, the MAAGI system is also expandable, and we envision that new and intelligent methods to analyze malware will be incorporated into the system in the future, giving cyber--defenders a powerful and continually evolving tool to combat malware--based attacks.

\section{Acknowledgments}
This work was supported by DARPA under US Air Force contract FA8750-10-C-0171, with thanks to Mr. Timothy Fraser. The views expressed are those of the authors and do not reflect the official policy or position of the Department of Defense or the U.S. Government.

%\section*{References}

\bibliography{MAAGI_aij}

\begin{thebibliography}{10}
\expandafter\ifx\csname url\endcsname\relax
  \def\url#1{\texttt{#1}}\fi
\expandafter\ifx\csname urlprefix\endcsname\relax\def\urlprefix{URL }\fi
\expandafter\ifx\csname href\endcsname\relax
  \def\href#1#2{#2} \def\path#1{#1}\fi

\bibitem{szolovits1988artificial}
P.~Szolovits, R.~S. Patil, W.~B. Schwartz, Artificial intelligence in medical
  diagnosis, Annals of internal medicine 108~(1) (1988) 80--87.

\bibitem{tyugu2011artificial}
E.~Tyugu, Artificial intelligence in cyber defense, in: Cyber Conflict (ICCC),
  2011 3rd International Conference on, IEEE, 2011, pp. 1--11.

\bibitem{cruz2014mcafee}
B.~Cruz, P.~Greve, B.~Kay, H.~Li, D.~McLean, F.~Paget, C.~Schmugar, R.~Simon,
  D.~Sommer, B.~Sun, J.~Walter, A.~Wosotowsky, C.~Xu, Mcafee labs threats
  report: Fourth quarter 2013, McAfee Labs.

\bibitem{jang_bitshred:_2011}
J.~Jang, D.~Brumley, S.~Venkataraman, {BitShred:} feature hashing malware for
  scalable triage and semantic analysis, in: Proceedings of the 18th {ACM}
  conference on Computer and communications security, {CCS} '11, {ACM}, New
  York, {NY}, {USA}, 2011, p. 309–320.
\newblock \href {http://dx.doi.org/10.1145/2046707.2046742}
  {\path{doi:10.1145/2046707.2046742}}.

\bibitem{bayer2009scalable}
U.~Bayer, P.~M. Comparetti, C.~Hlauschek, C.~Kruegel, E.~Kirda, Scalable,
  behavior-based malware clustering., in: NDSS, Vol.~9, Citeseer, 2009, pp.
  8--11.

\bibitem{rieck2011automatic}
K.~Rieck, P.~Trinius, C.~Willems, T.~Holz, Automatic analysis of malware
  behavior using machine learning, Journal of Computer Security 19~(4) (2011)
  639--668.

\bibitem{perdisci2013scalable}
R.~Perdisci, D.~Ariu, G.~Giacinto, Scalable fine-grained behavioral clustering
  of http-based malware, Computer Networks 57~(2) (2013) 487--500.

\bibitem{sahoo2006incremental}
N.~Sahoo, J.~Callan, R.~Krishnan, G.~Duncan, R.~Padman, Incremental
  hierarchical clustering of text documents, in: Proceedings of the 15th ACM
  international conference on Information and knowledge management, ACM, 2006,
  pp. 357--366.

\bibitem{dumitras2011experimental}
T.~Dumitras, I.~Neamtiu, Experimental challenges in cyber security: a story of
  provenance and lineage for malware, Cyber Security Experimentation and Test.

\bibitem{karim2005malware}
M.~E. Karim, A.~Walenstein, A.~Lakhotia, L.~Parida, Malware phylogeny
  generation using permutations of code, Journal in Computer Virology 1~(1-2)
  (2005) 13--23.

\bibitem{cybergenome}
DARPA,
  \href{http://www.darpa.mil/Our_Work/I2O/Programs/Cyber_Defense_(Cyber_Genome).aspx}{Cyber
  genome program} (2010).
\newline\urlprefix\url{http://www.darpa.mil/Our_Work/I2O/Programs/Cyber_Defense_(Cyber_Genome).aspx}

\bibitem{lakhotia2013fast}
A.~Lakhotia, M.~{Dalla Preda}, R.~Giacobazzi, Fast location of similar code
  fragments using semantic `juice', in: SIGPLAN Program Protection and Reverse
  Engineering Workshop, ACM, 2013, p.~5.

\bibitem{broder1997resemblance}
A.~Z. Broder, On the resemblance and containment of documents, in: Compression
  and Complexity of Sequences 1997. Proceedings, IEEE, 1997, pp. 21--29.

\bibitem{pfeffer2011practical}
A.~Pfeffer, Practical probabilistic programming, in: Inductive Logic
  Programming, Springer, 2011, pp. 2--3.

\bibitem{blondel2008fast}
V.~D. Blondel, J.-L. Guillaume, R.~Lambiotte, E.~Lefebvre, Fast unfolding of
  communities in large networks, Journal of Statistical Mechanics: Theory and
  Experiment 2008~(10) (2008) P10008.

\bibitem{gunn1998support}
S.~R. Gunn, et~al., Support vector machines for classification and regression.

\bibitem{pfeffer2012malware}
A.~Pfeffer, C.~Call, J.~Chamberlain, L.~Kellogg, J.~Ouellette, T.~Patten,
  G.~Zacharias, A.~Lakhotia, S.~Golconda, J.~Bay, et~al., Malware analysis and
  attribution using genetic information, in: Malicious and Unwanted Software
  (MALWARE), IEEE, 2012, pp. 39--45.

\bibitem{ruttenberg2014}
B.~E. Ruttenberg, C.~Miles, L.~Kellogg, V.~Notani, M.~Howard, C.~LeDoux,
  A.~Lakhotia, A.~Pfeffer, Identifying shared software components to support
  malware forensics, in: To appear in 11th Conference on Detection of
  Intrusions and Malware and Vulnerability Assessment, 2014.

\bibitem{jaccard1912distribution}
P.~Jaccard, The distribution of the flora in the alpine zone. 1, New
  phytologist 11~(2) (1912) 37--50.

\bibitem{manasse2007consistent}
M.~Manasse, F.~McSherry, K.~Talwar, Consistent weighted sampling, Unpublished
  technical report) http://research. microsoft. com/en-us/people/manasse.

\bibitem{saitou1987neighbor}
N.~Saitou, M.~Nei, The neighbor-joining method: a new method for reconstructing
  phylogenetic trees., Molecular biology and evolution 4~(4) (1987) 406--425.

\bibitem{theodoridis2008pattern}
S.~Theodoridis, K.~Koutroumbas,
  \href{http://books.google.com/books?id=QgD-3Tcj8DkC}{Pattern Recognition},
  Elsevier Science, 2008.
\newline\urlprefix\url{http://books.google.com/books?id=QgD-3Tcj8DkC}

\bibitem{newman2006modularity}
M.~E. Newman, Modularity and community structure in networks, Proceedings of
  the National Academy of Sciences 103~(23) (2006) 8577--8582.

\bibitem{hubert1985comparing}
L.~Hubert, P.~Arabie, Comparing partitions, Journal of classification 2~(1)
  (1985) 193--218.

\bibitem{jang2013towards}
J.~Jang, M.~Woo, D.~Brumley, Towards automatic software lineage inference, in:
  Proceedings of the 22nd USENIX conference on Security, USENIX Association,
  2013, pp. 81--96.

\bibitem{prufer1918neuer}
H.~Pr{\"u}fer, Neuer beweis eines satzes {\"u}ber permutationen, Arch. Math.
  Phys 27 (1918) 742--744.

\bibitem{blanch2014mcafee}
R.~Blanch, Malware threats, trend and predictions for 2014, McAfee Labs.

\bibitem{websense2012}
W.~S. Labs,
  \href{http://www.websense.com/content/websense-2013-security-predictions.html}{2013
  security predictions} (2012).
\newline\urlprefix\url{http://www.websense.com/content/websense-2013-security-predictions.html}

\bibitem{kang2012malware}
B.~Kang, T.~Kim, H.~Kwon, Y.~Choi, E.~G. Im, Malware classification method via
  binary content comparison, in: Proceedings of the 2012 ACM Research in
  Applied Computation Symposium, ACM, 2012, pp. 316--321.

\bibitem{deerman2012}
J.~Deerman, Advanced malware detection through attack life cycle analysis.

\bibitem{juzonis2012specialized}
V.~Juzonis, N.~Goranin, A.~Cenys, D.~Olifer, Specialized genetic algorithm
  based simulation tool designed for malware evolution forecasting, in: Annales
  UMCS, Informatica, Vol.~12, 2012, pp. 23--37.

\bibitem{hosmer2000introduction}
D.~W. Hosmer, S.~Lemeshow, R.~X. Sturdivant, Introduction to the logistic
  regression model, Wiley Online Library, 2000.

\bibitem{draper1981applied}
N.~R. Draper, H.~Smith, Applied regression analysis 2nd ed.

\bibitem{schmidt2009distilling}
M.~Schmidt, H.~Lipson, Distilling free-form natural laws from experimental
  data, science 324~(5923) (2009) 81--85.

\bibitem{dreiseitl2002logistic}
S.~Dreiseitl, L.~Ohno-Machado, Logistic regression and artificial neural
  network classification models: a methodology review, Journal of biomedical
  informatics 35~(5) (2002) 352--359.

\bibitem{CC01a}
C.-C. Chang, C.-J. Lin, {LIBSVM}: A library for support vector machines, ACM
  Transactions on Intelligent Systems and Technology 2 (2011) 27:1--27:27,
  software available at \url{http://www.csie.ntu.edu.tw/~cjlin/libsvm}.

\bibitem{clark1989cn2}
P.~Clark, T.~Niblett, The cn2 induction algorithm, Machine learning 3~(4)
  (1989) 261--283.

\bibitem{quinlan1986induction}
J.~R. Quinlan, Induction of decision trees, Machine learning 1~(1) (1986)
  81--106.

\bibitem{powers2011evaluation}
D.~M. Powers, Evaluation: from precision, recall and f-measure to roc,
  informedness, markedness and correlation.

\end{thebibliography}

\end{document}